\documentclass[11pt,letterpaper]{article}

\usepackage{etex}
\usepackage{bbm}

\usepackage{xspace,enumerate}
\usepackage[dvipsnames]{xcolor}
\usepackage[T1]{fontenc}
\usepackage[full]{textcomp}
\usepackage[american]{babel}
\usepackage{mathtools}
\usepackage{amsthm}
\newtheorem{theorem}{Theorem}[section]
\newtheorem*{theorem*}{Theorem}

\newtheorem{proposition}[theorem]{Proposition}
\newtheorem*{proposition*}{Proposition}
\newtheorem{lemma}[theorem]{Lemma}
\newtheorem*{lemma*}{Lemma}
\newtheorem{corollary}[theorem]{Corollary}
\newtheorem*{conjecture*}{Conjecture}
\newtheorem{fact}[theorem]{Fact}
\newtheorem*{fact*}{Fact}

\newtheorem*{hypothesis*}{Hypothesis}

\theoremstyle{definition}
\newtheorem{definition}[theorem]{Definition}
\newtheorem*{definition*}{Definition}

\newtheorem{question}[theorem]{Question}

\newtheorem{algorithm}[theorem]{Algorithm}

\theoremstyle{remark}
\newtheorem{claim}[theorem]{Claim}
\newtheorem*{claim*}{Claim}

\newtheorem*{remark*}{Remark}

\newtheorem*{observation*}{Observation}
\usepackage[
letterpaper,
top=1.2in,
bottom=1.2in,
left=1in,
right=1in]{geometry}
\usepackage{newpxtext} 
\usepackage{textcomp} 
\usepackage[varg,bigdelims]{newpxmath}
\usepackage[scr=rsfso]{mathalfa}
\usepackage{bm} 
\let\mathbb\varmathbb
\usepackage{microtype}
\usepackage[
pagebackref,
colorlinks=true,
urlcolor=blue,
linkcolor=blue,
citecolor=OliveGreen,
]{hyperref}
\usepackage[capitalise,nameinlink]{cleveref}
\crefname{lemma}{Lemma}{Lemmas}
\crefname{fact}{Fact}{Facts}
\crefname{theorem}{Theorem}{Theorems}
\crefname{corollary}{Corollary}{Corollaries}
\crefname{claim}{Claim}{Claims}
\crefname{example}{Example}{Examples}
\crefname{algorithm}{Algorithm}{Algorithms}
\crefname{problem}{Problem}{Problems}
\crefname{definition}{Definition}{Definitions}
\usepackage{paralist}
\usepackage{turnstile}
\usepackage{mdframed}
\usepackage{algorithm}
\usepackage{algpseudocode}
\usepackage{tikz}

\newcommand{\Authornote}[2]{}

\newcommand{\Snote}{\Authornote{S}}
\newcommand{\SKnote}{\Authornote{SK}}

\usepackage{boxedminipage}
\newcommand{\paren}[1]{(#1)}
\newcommand{\Paren}[1]{\left(#1\right)}

\newcommand{\Brac}[1]{\left[#1\right]}

\newcommand{\Abs}[1]{\left\lvert#1\right\rvert}

\newcommand{\Set}[1]{\left\{#1\right\}}

\newcommand{\norm}[1]{\lVert#1\rVert}
\newcommand{\Norm}[1]{\left\lVert#1\right\rVert}

\newcommand{\iprod}[1]{\langle#1\rangle}
\newcommand{\Iprod}[1]{\left\langle#1\right\rangle}

\newcommand{\Esymb}{\mathbb{E}}
\newcommand{\Psymb}{\mathbb{P}}
\newcommand{\Vsymb}{\mathbb{V}}
\DeclareMathOperator*{\E}{\Esymb}
\DeclareMathOperator*{\Var}{\Vsymb}
\DeclareMathOperator*{\ProbOp}{\Psymb}
\DeclareMathOperator*{\V}{\Vsymb}
\renewcommand{\Pr}{\ProbOp}

\newcommand{\tensor}{\otimes}

\def\nnewcolor{0}

\ifnum\nnewcolor=1
\newcommand{\new}[1]{{\color{red} #1}}
\fi
\ifnum\nnewcolor=0
\newcommand{\new}[1]{#1}
\fi

\newcommand{\mper}{\,.}
\newcommand{\mcom}{\,,}

\newcommand\bdot\bullet

\DeclareMathOperator{\Ind}{\mathbf 1}

\DeclareMathOperator{\Tr}{Tr}

\DeclareMathOperator{\poly}{poly}

\DeclareMathOperator{\polylog}{polylog}

\newcommand{\N}{\mathbb N}
\newcommand{\R}{\mathbb R}

\newcommand{\cA}{\mathcal A}
\newcommand{\cB}{\mathcal B}
\newcommand{\cC}{\mathcal C}

\newcommand{\cN}{\mathcal N}

\renewcommand{\leq}{\leqslant}

\renewcommand{\geq}{\geqslant}

\let\epsilon=\varepsilon
\numberwithin{equation}{section}
\newcommand\MYcurrentlabel{xxx}
\newcommand{\MYstore}[2]{%
  \global\expandafter \def \csname MYMEMORY #1 \endcsname{#2}%
}
\newcommand{\MYload}[1]{%
  \csname MYMEMORY #1 \endcsname%
}
\newcommand{\MYnewlabel}[1]{%
  \renewcommand\MYcurrentlabel{#1}%
  \MYoldlabel{#1}%
}
\newcommand{\MYdummylabel}[1]{}
\newcommand{\torestate}[1]{%
  \let\MYoldlabel\label%
  \let\label\MYnewlabel%
  #1%
  \MYstore{\MYcurrentlabel}{#1}%
  \let\label\MYoldlabel%
}
\newcommand{\restatetheorem}[1]{%
  \let\MYoldlabel\label
  \let\label\MYdummylabel
  \begin{theorem*}[Restatement of \cref{#1}]
    \MYload{#1}
  \end{theorem*}
  \let\label\MYoldlabel
}
\newcommand{\restatelemma}[1]{%
  \let\MYoldlabel\label
  \let\label\MYdummylabel
  \begin{lemma*}[Restatement of \cref{#1}]
    \MYload{#1}
  \end{lemma*}
  \let\label\MYoldlabel
}
\newcommand{\restateprop}[1]{%
  \let\MYoldlabel\label
  \let\label\MYdummylabel
  \begin{proposition*}[Restatement of \cref{#1}]
    \MYload{#1}
  \end{proposition*}
  \let\label\MYoldlabel
}
\newcommand{\restatefact}[1]{%
  \let\MYoldlabel\label
  \let\label\MYdummylabel
  \begin{fact*}[Restatement of \prettyref{#1}]
    \MYload{#1}
  \end{fact*}
  \let\label\MYoldlabel
}
\newcommand{\restate}[1]{%
  \let\MYoldlabel\label
  \let\label\MYdummylabel
  \MYload{#1}
  \let\label\MYoldlabel
}

\newcommand{\e}{\epsilon}
\newcommand{\eps}{\epsilon}

\allowdisplaybreaks
\newcommand*{\Id}{\mathrm{Id}}

\newcommand*{\dtv}{\text{TV}}
\newcommand*{\tmu}{\tilde{\mu}}

\newcommand*{\loweredwidetildehelper}[2]{\hbox{\csname dimen@\endcsname\accentfontxheight#1%
  \accentfontxheight#11.25\csname dimen@\endcsname
  $\csname m@th\endcsname#1\widetilde{#2}$%
  \accentfontxheight#1\csname dimen@\endcsname
  }%
}
\newcommand*{\accentfontxheight}[1]{\fontdimen5\ifx#1\displaystyle \textfont \else\ifx#1\textstyle \textfont \else\ifx#1\scriptstyle \scriptfont \else \scriptscriptfont \fi\fi\fi3
}

\DeclareMathOperator*{\pE}{\tilde{\mathbb{E}}}

\newcommand{\proves}[1]{\vdash_{#1}}

\title{Robustly Learning any Clusterable Mixture of Gaussians}

\author{Ilias Diakonikolas\thanks{UW Madison. \url{ilias@cs.wisc.edu}. Supported by NSF Award CCF-1652862 (CAREER) and a Sloan Research Fellowship. Some of this research was performed while visiting the Simons Institute on the Theory of Computing.} 
\and 
Samuel B. Hopkins\thanks{UC Berkeley. \url{hopkins@berkeley.edu}. Supported by a Miller Postdoctoral Fellowship.} 
\and 
Daniel Kane\thanks{UC San Diego. \url{dakane@ucsd.edu}. 
Supported by NSF Award CCF-1553288 (CAREER) and a Sloan Research Fellowship.} 
\and 
Sushrut Karmalkar\thanks{UT Austin. \url{sushrutk@cs.utexas.edu}. 
Supported by NSF CNS 1414082. Some of this research was performed 
while visiting USC and the Simons Institute on the Theory of Computing.}
}

\begin{document}

\maketitle
\thispagestyle{empty}

\begin{abstract}
We study the efficient learnability of high-dimensional Gaussian mixtures in the outlier-robust setting, 
where a small constant fraction of the data is adversarially corrupted. We resolve the polynomial learnability of this problem 
when the components are pairwise separated in total variation distance. Specifically, we provide an algorithm that, 
for any constant number of components $k$, runs in polynomial time and learns the components of an $\eps$-corrupted $k$-mixture 
within information theoretically near-optimal error of $\tilde{O}(\eps)$, under the assumption that the overlap between any pair of 
components $P_i, P_j$ (i.e., the quantity $1-TV(P_i, P_j)$) is bounded by $\mathrm{poly}(\eps)$. 

Our separation condition is the qualitatively weakest assumption under which accurate clustering of the samples is possible. 
In particular, it allows for components with arbitrary covariances and for components with identical means, 
as long as their covariances differ sufficiently.
Ours is the first polynomial time algorithm for this problem, even for $k=2$.

Our algorithm follows the Sum-of-Squares based \emph{proofs to algorithms} approach. 
Our main technical contribution is a new robust identifiability proof of clusters from a Gaussian mixture, 
which can be captured by the constant-degree Sum of Squares proof system. The key ingredients of this proof 
are a novel use of \emph{SoS-certifiable anti-concentration} and a new characterization of pairs of Gaussians 
with small (dimension-independent) overlap in terms of their parameter distance.
\end{abstract}

\newpage

\tableofcontents
\thispagestyle{empty}

\newpage
\pagenumbering{arabic}


\section{Introduction} \label{sec:intro}

\subsection{Background} \label{ssec:background}
A mixture of Gaussians or {\em Gaussian mixture model (GMM)} is
a convex combination of Gaussian distributions, i.e.,
a distribution on $\R^d$ of the form $F = \sum_{i=1}^k w_i \cN(\mu_i, \Sigma_i)$, where
the weights $w_i$, mean vectors $\mu_i$, and covariance matrices $\Sigma_i$ are unknown.
GMMs are arguably {\em the} most extensively studied latent variable model in the statistics 
and computer science literatures, starting with the pioneering work of Karl Pearson in 1894~\cite{Pearson:94}.

The problem of efficiently learning a high-dimensional GMM from samples has received tremendous attention in theoretical
computer science and machine learning. A long line of work initiated by Dasgupta~\cite{Dasgupta:99, AroraKannan:01, VempalaWang:02, AchlioptasMcSherry:05, KSV08, BV:08} gave efficient clustering algorithms for GMMs under various separation assumptions. 
Subsequently, efficient parameter estimation algorithms were obtained~\cite{KMV:10, MoitraValiant:10, BelkinSinha:10, HardtP15} 
under minimal information-theoretic conditions. The related problems of density estimation and proper learning
have also been well-studied~\cite{FOS:06, SOAJ14, DK14, MoitraValiant:10, HardtP15, LiS17}. 

The correctness of the aforementioned learning algorithms crucially relies on the assumption 
that the generative model is correctly specified, i.e., that the input dataset is a multi-set
of independent samples drawn from an unknown mixture of Gaussians. 
However, this simplifying assumption is at best only approximately valid, as 
real datasets are typically exposed to some source of contamination.
The more realistic setting that the model can be slightly misspecified (or, equivalently, that there exists 
a small fraction of outliers in the data) turns out to be significantly more challenging. 
In fact, the existence of even a {\em single} arbitrary outlier in the dataset completely compromises 
the behavior of the aforementioned algorithms. 

Robust statistics~\cite{HampelEtalBook86, Huber09} is a classical sub-field of statistics
focusing on the design of estimators that perform well in the presence of outliers (or under model misspecification).
While the information-theoretic limits of robust estimation have been well-understood in 
several settings of interest,  even the most basic algorithmic questions had remained wide open until recently.

Recent work in theoretical computer science, starting with~\cite{DKKLMS16, LaiRV16}, 
gave the first efficient and outlier-robust learning algorithms for a variety
of high-dimensional distributions, including high-dimensional Gaussians. 
Specifically,~\cite{DKKLMS16} gave the first robust learning algorithms 
for a single high-dimensional Gaussian that can tolerate a {\em constant} fraction of outliers, independent of the dimension. 
Since the dissemination of~\cite{DKKLMS16, LaiRV16}, there has been an explosion 
of research activity on algorithmic robust high-dimensional estimation (see, e.g.,~\cite{DK20-survey} for a  recent survey on the topic). 

Most relevant to this paper are the prior works~\cite{DKKLMS16, HopkinsL18, KothariSS18, DiakonikolasKS18-mixtures}.
\cite{DKKLMS16} gave a robust density estimation algorithm for a mixture of spherical Gaussians. 
On the other hand, \cite{HopkinsL18, KothariSS18, DiakonikolasKS18-mixtures} gave robust parameter estimation
algorithms for mixtures of (essentially) spherical Gaussians under near-optimal separation assumptions. 
(See Section~\ref{ssec:related} for a detailed summary of prior work.)
Despite this recent progress, the algorithmic problem of robustly learning a mixture of a constant number
(or even two) arbitrary Gaussians has remained open and is recognized
as a central open problem in this field~\cite{DVW19-vignette}. 

In this work, we focus on the important special case of this problem 
where the mixture components are ``separated''. Various notions of separation have been used in the literature.
Here we focus on the following definition: 
We say that a $k$-mixture of Gaussians is {\em separated} if 
the {\em overlap} between any pair of  components $P, Q$ (i.e., $1-\dtv(P, Q)$, 
where $\dtv(P, Q)$ is the total variation distance 
between $P$ and $Q$) is a small constant --- independent of the dimension.
\new{We note that this is qualitatively} the weakest possible separation assumption under which accurate clustering of the components 
is information-theoretically possible --- even without outliers.

The preceding discussion motivates the following question, whose resolution is the main
result of this work:

\begin{question} \label{q:sep-gmm}
Is there a $\poly(d, 1/\eps)$-time robust learning algorithm for a mixture of any constant number of 
(or even two) arbitrary {\em separated} Gaussians on $\R^d$?
\end{question}

\subsection{Our Results and Techniques} \label{ssec:results}

In this paper, we provide an affirmative answer to Question~\ref{q:sep-gmm}.
Our main result is the following theorem:

\begin{theorem}[Main Result]\label{thm:main-informal}
For every $w_{\min} > 0$ there are $F(w_{\min}),f(w_{\min}) > 0$ such that there is an algorithm which takes $n = d^{F(w_{\min})}/\poly(\e)$ $\e$-corrupted samples from any $d$-dimensional mixture of Gaussians $\sum_{i \leq k} w_i G_i$ with $w_i \geq w_{\min}$, runs in time $n^{F(w_{\min})}$, and returns $k$ hypothesis Gaussians $H_1, \ldots, H_k$ and weights $u_1,\ldots,u_k$ such that with high probability
there exists a permutation $\pi: [k] \to [k]$ for which $\max_i \dtv(G_i,H_{\pi(i)}) \leq \tilde{O}(\epsilon)$ and $\sum_{i \leq k} |w_i - u_{\pi(i)}| \leq O(\epsilon)$,
so long as $\e \leq f(w_{\min})$ and $\min_{i \neq j} \dtv(G_i,G_j) \geq 1 - f(w_{\min})$.
\end{theorem}

\noindent Before we describe our techniques, some comments are in order. For now we restrict attention to the case that the weights $w_i$ are uniform, so $1/w_{\min}$ in \cref{thm:main-informal} can be taken equal to $k$.
Then the sample complexity becomes $d^{F(k)}/\poly(\eps)$, and the running time is $n^{F(k)}$.
If $k$ is constant, this yields an algorithm with polynomial running time and sample complexity.
We note that prior to this work, even the $k=2$ special case of Theorem~\ref{thm:main-informal} remained open -- recent work of Jia and Vempala addresses a special case of the $k=2$ setting \cite{JV19}.


\paragraph{Our Techniques.}
Here we provide a detailed intuitive overview of our approach in tandem 
with a brief comparison to prior work.

We will describe an algorithm to robustly learn the parameters of an unknown $k$-mixture $\sum_{i=1}^k w_i G_i$
on $\R^d$, where the individual components $G_i$, $i \in [k]$, can be arbitrary Gaussians 
satisfying the constraint that the pairwise overlap $1-\dtv(G_i , G_j)$ is a sufficiently small, 
but (importantly) {\em dimension-independent} quantity. \new{This condition is 
qualitatively} the weakest separation assumption under which accurate clustering is 
information-theoretically possible --- even without corruptions. 

It should be noted that the only prior algorithm in the literature that can (non-robustly) 
learn Gaussian $k$-mixtures under the separation we consider is the algorithm of~\cite{MoitraValiant:10}, 
which in fact solves the (more challenging) parameter estimation task even in the regime that 
$\min_{i \neq j} \dtv(G_i,  G_j)$ is close to zero. 
However, as we explain
in Section~\ref{ssec:related}, the approach of~\cite{MoitraValiant:10} is inherently non-robust. That is,
to develop a {\em robust} learning algorithm for our setting, it is in some sense necessary to develop 
a genuinely new {\em non-robust} algorithm that is \new{amenable to ``robustification''.}

For simplicity of the proceeding exposition, we will focus on uniform weights, 
i.e., the special case when $w_i = 1/k$ for all $i \in [k]$. We note that our techniques straightforwardly 
generalize to arbitrary weights (see \cref{sec:arbit-weights} \new{for a sketch}).

We now provide a high-level outline of our approach 
followed by a somewhat more technical description of the various steps.
Our general strategy is to learn an {\em accurate clustering} of the samples.
We show (Proposition~\ref{prop:const-acc}) that it in fact suffices 
to find a ``rough'' clustering of our samples, i.e., one that correctly classifies a $1-1/\poly(k)$ fraction of the points. 
It should be noted that our rough clustering algorithm actually produces a list of $F(k)$ candidate clusterings, 
for some function $F(k)$, with the guarantee that at least one of them achieves the desired accuracy. 
By leveraging Proposition~\ref{prop:const-acc} and a standard hypothesis testing routine, 
we show that even this weaker guarantee suffices for our purposes.
Intuitively, such an implication holds because of the assumed separation between the components.
Given this rough clustering, we can use existing robust estimation algorithms~\cite{DKKLMS16} 
to learn rough approximations to the components. We can then obtain a better clustering of the points (based on which hypothesis component assigns a given point higher probability mass). 
This new clustering will have error $O(\eps)$, and so feeding these new samples into our robust
estimation algorithms will give us the final error of $\tilde{O}(\eps)$, as desired.

To find our desired clustering, we use the Sums-of-Squares (SoS) ``proofs to algorithms'' method.
In the subsequent discussion, we will first focus on the non-robust version of our method for the
following reasons. First, the non-robust version is already quite challenging. Second,
as a byproduct of our approach, we can ``robustify'' our non-robust learning algorithm by leveraging
a key insight from the recent high-dimensional robust statistics literature~\cite{DKKLMS16}. 

Our goal will be to efficiently find a partition of the samples into ``good'' clusters of points -- a good cluster is one with the property that its 
low-degree moments are close to those of a Gaussian in a certain technical sense. It should be noted 
that our definition of what it means for the moments to be {\em close} is necessarily stronger than the notion which 
was used in prior work for the spherical setting~\cite{HopkinsL18, KothariSS18}. 
Roughly speaking, these prior works only imposed
(Gaussian) upper bounds on the low-degree moments of a cluster. While this condition 
was sufficient for the spherical setting, it inherently fails
in our more challenging regime (see, e.g., Section~\ref{sec:mom-upperbounds-not-enough} 
for an intuitive explanation).

Given our definition of a ``good'' cluster, a necessary condition for our approach to be viable
is that any good cluster is in fact close to (i.e., has large overlap with) 
a true cluster (i.e., a cluster of points from the same Gaussian component).
In the present context, a mathematical \emph{proof} that good clusters are close to true clusters is called a {\em (cluster) identifiability proof}.
The gist of the Sums-of-Squares method is that if we have a ``sufficiently simple'' 
identifiability proof --- specifically one that can be captured by the constant-degree SoS proof system ---
we can mechanically translate it to a polynomial-time SDP-based algorithm for the clustering problem.
In more detail, we formulate a system of polynomial inequalities 
in cluster-assignment variables whose solutions are good clusters, 
and give a constant-degree SoS identifiability proof that any solution $w$ to our polynomial
system has high overlap with a true cluster.

\paragraph{Cluster Identifiability}
Our key technical contribution is an SoS proof of cluster identifiability, 
which we now describe in more detail, still in the non-robust setting.
The main challenge is to find a proof of identifiability which is sufficiently simple 
that one may even hope to accomplish it within the SoS proof system.
In this overview, we describe this simple proof, and defer to the body 
of the paper the details of how it can be encoded into a formal proof system.
The mathematical statements in the following overview are informal, 
with many constants and lower-order terms left out for simplicity of exposition. 
The reader is referred to the lemma statements in the body of the paper for formal statements 
and proofs.

Let $X_1,\ldots,X_n$ be i.i.d. samples from the mixture $\tfrac 1 k \sum_{i=1}^k G_i$, 
and let $S_1,\ldots,S_k$ be the partition of $X_1,\ldots,X_n$ into {\em true clusters}, i.e., 
the set $S_i$ corresponds to the samples drawn from $G_i$.
Informally, the main identifiability statement is:

\begin{samepage}
\begin{lemma}[Cluster identifiability, informal]
\label[lemma]{lem:ident-intro}
For every pair of true clusters $S_i, S_j$, if $TV(G_i,G_j) \geq 1-\delta$ then we have:
\begin{quote}
(*) For every subset $S \subseteq \{X_1,\ldots,X_n\}$, if the $t$-th empirical moments of $S$ are approximately Gaussian, 
i.e., if for all  $t \leq K(k) = \poly(k)$ it holds
\[
\E_{X \sim S} X^{\tensor t} \approx \E_{X \sim \cN(\mu_S,\Sigma_S)} X^{\tensor t} \text{, where } \mu_S, \Sigma_S \text{ are the empirical mean/covariance of } S  \;,
\]
then either $\frac{|S \cap S_i|}{|S|} \leq \delta'(\delta) \cdot K(k) + \frac 1 {k^{10}}$ or 
$\frac{|S \cap S_j|}{|S|} \leq \delta'(\delta) \cdot K(k) + \frac 1 {k^{10}}$ \;,
\end{quote}
where $\delta'(\delta) = 1/\polylog(1/\delta) \rightarrow 0$ as $\delta \rightarrow 0$. 
\end{lemma}
\end{samepage}

Using the SoS clustering approach from prior works \cite{HopkinsL18,KothariSS18}, an SoS proof of (*) from \cref{lem:ident-intro} would suffice to obtain an algorithm for clustering samples from $\tfrac 1 k \sum_{i=1}^k G_i$ which correctly clusters at least a $1-1/\poly(k)$ fraction of the samples (taking $\delta$ sufficiently small with respect to $k$).
\cite{HopkinsL18,KothariSS18} develop this approach and apply it in the much simpler 
setting that $G_i,G_j$ have variance at most $1$ \emph{in every direction} 
and that there is a direction $v$ in which their means are separated by $k^{\Omega(1)}$.
Our separation assumption $\dtv(G_i,G_j) \geq 1 -\delta(k)$ means we must handle a significantly more challenging set of possible relationships among clusters $S_i, S_j$.

For technical reasons, we are not aware of an SoS proof of  the statement (*) without some additional assumptions.
(This point was overlooked in a previous version of the present paper -- see \cref{sec:concurrent}.)
For the sake of this intuitive exposition, we will first sketch a simple proof of \cref{lem:ident-intro}, 
and then we will describe a technical modification which does have an SoS proof and 
suffices for our clustering purposes.

The starting point of our proof of \cref{lem:ident-intro} is to consider three cases, 
using the following new structural result (see Lemma~\ref{lem:gaussian-cases} for the formal statement)
characterizing pairs of Gaussians $P,Q$ with small overlap in terms of their parameter distance.
Our structural lemma says that if $\dtv(P,Q) \geq 1-\delta$,
then we are in one of the following three cases, for some $\delta'(\delta)$:
\begin{itemize}
\item[(i)] There exists a hyperplane approximately separating the components,
i.e., a direction $v$ such that $\iprod{v,\mu_P - \mu_Q}^2 \geq \iprod{v,\Sigma_P v}/\delta'$ 
and $\iprod{v, \mu_P - \mu_Q}^2 \geq \iprod{v, \Sigma_Q v}/\delta'$.
\item[(ii)] There exists a direction $v$ such that the variances of the $v$-projections of $P$
and $Q$ differ significantly, i.e., $\iprod{v,\Sigma_P v} \leq \delta' \cdot \iprod{v,\Sigma_Q v}$.

\item[(iii)] The covariance matrices of $P$ and $Q$ differ significantly in relative Frobenius norm, 
i.e., $\|I - \Sigma_P^{-1/2} \Sigma_Q \Sigma_P^{-1/2}\|_F \geq 1/\delta'$. 
\end{itemize}
See \cref{fig:separation-types} for an illustration of these three distinct cases.
In \cref{lem:mean_sep}, \cref{lem:var_sep}, and \cref{lem:covar_sep}, 
we prove (the formal version of an SoS-provable modification of) \cref{lem:ident-intro} in cases (i), (ii), and (iii), respectively.

In the following paragraphs, we give an overview of these proofs.

  \begin{figure}[!htb]
	\center{\includegraphics[width=\textwidth]
		{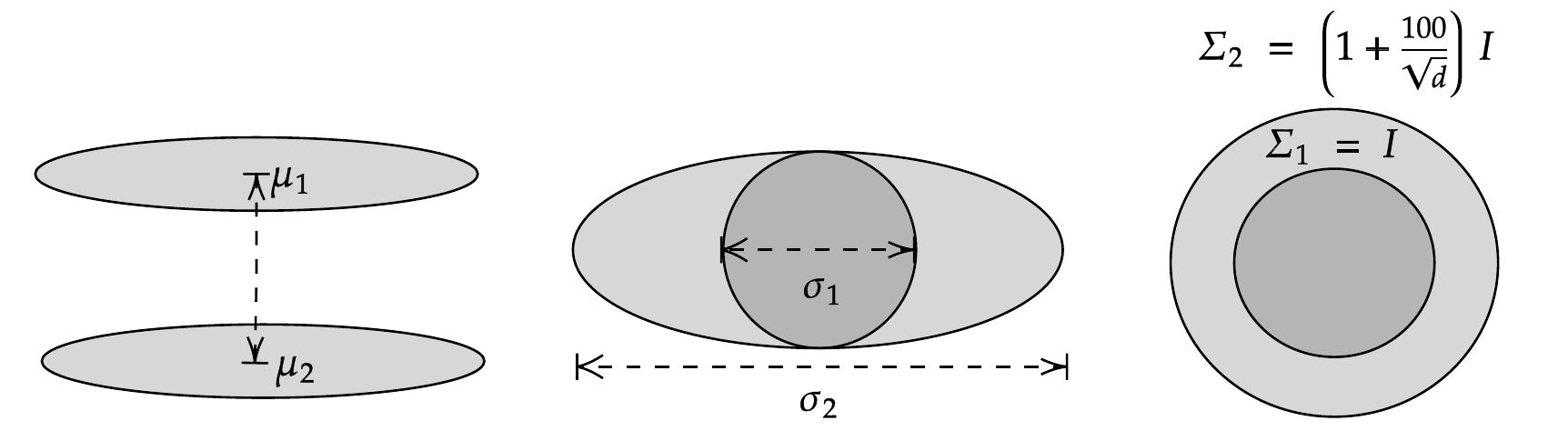}}
\caption{\label{fig:separation-types} The only three ways in which a pair of Gaussians can be separated in total variation distance. 
Either there is a direction in which the Gaussians differ significantly in their means (left), 
or there is such a direction for the variances (middle), 
or the covariance matrices differ significantly in the relative Frobenius norm (right).}
\end{figure}

\paragraph{Witnessing Polynomials and Variance Bounds}
For $G_i, G_j$, $i \neq j$, in each of cases (i),(ii),(iii), we first identify a constant-degree 
\emph{witnessing polynomial} $p_{ij} \, : \, \R^d \rightarrow \R$, which can certify \emph{non-Gaussianity} 
of any set of samples $S$ such that $|S_i \cap S| \geq \delta' K(k)|S|$ and $|S_j \cap S| \geq \delta' K(k)|S|$.
In particular, we will find a polynomial $p_{ij}$ such that the variance of \new{$p_{ij}$} under the uniform distribution on $S$ 
is much larger than allowed for any $S$ which has both low-degree Gaussian empirical moments \new{(up to order $K(k)$)}
and a $\delta' K(k)$-fraction of its mass on each of $S_i, S_j$.
Specifically, in cases (i) and (ii), we have $p_{ij}(x) = \iprod{x,v}$, \new{where $v$ is the direction of separation,}
and in case (iii), we have $p_{ij}(x) = \iprod{x, Ax}$, 
where $A \approx \Sigma_i - \Sigma_j$.\footnote{For technical reasons, we in fact take $A = \Sigma_i^{-1/2} (I - \Sigma_i^{-1/2} \Sigma_j \Sigma_i^{-1/2}) \Sigma_i^{-1/2}$.}
We will henceforth drop the index $ij$ and use $p$ for the polynomial in question.

\subparagraph{Cases (i) and (ii): hyperplane and variance separation}
Let us start with cases (i) and (ii), which capture most of the technical insight, see \cref{lem:mean_sep,lem:var_sep}.
(In case (iii), we will be able to re-use much of the work from (i) and (ii).)
The key to cases (i) and (ii) are two technical lemmas (\cref{lem:fake_var_upperbound,lem:fake_var_lowerbound}), which establish both upper and lower bounds on the variance of $p(x)$ under the uniform distribution on any subset of samples $S$ with both Gaussian $t$-th moments (for $t \leq K(k)$) and significant overlap with $S_i$ and $S_j$.
If $G_i,G_j$ are in either of cases (i),(ii), these upper and lower bounds cannot hold simultaneously, and so any $S$ with Gaussian $t$-th moments cannot have significant overlap with both $S_i$ and $S_j$.

The first lemma, which is the more technically challenging of the two, 
says that the variance of $\iprod{x,v}$ under \new{the uniform distribution on} $S$ 
cannot be much larger than the variance under either of $G_i, G_j$.

\begin{lemma}[Informal version of \cref{lem:fake_var_upperbound}, variance upper bound on $\iprod{x,v}$]
\label[lemma]{lem:fake_var_upperbound_intro}
For every $\tau > 0$,
\begin{equation} \label{eqn:fake-var-ub-intro}
 \V_{X \sim S} [\iprod{X,v}] \cdot \Paren{\frac{|S_i \cap S|}{|S|} - \tau} \cdot \tau'(\tau) \lesssim \V_{X \sim G_i} [\iprod{X,v}]\mper
\end{equation}
Furthermore, \eqref{eqn:fake-var-ub-intro} has an SoS proof of degree $(\tau'(\tau))^{-1}$, where the indeterminates are assignment variables for the subset $S$.
\end{lemma}
Let us briefly discuss the proof of \cref{lem:fake_var_upperbound_intro}, since it is the most technically subtle in this paper.  
To prove \cref{lem:fake_var_upperbound_intro}, we rely on Gaussian anti-concentration, 
which implies the following basic fact about Gaussian distributions: 
If $P = \cN(0,\sigma^2)$ has $\Pr_{X \sim P}(|X| \leq \alpha) \geq \beta$, then $\sigma \leq O(\alpha/\beta)$.
In the context of \cref{lem:fake_var_upperbound_intro}, if we imagine temporarily that the distribution of $X \sim S$ is Gaussian 
with covariance $\Sigma_S$ (rather than just having its first $K(k)$ Gaussian-like moments), 
this means that the variance of $\iprod{X,v}$ cannot be too large under $S$. This follows from the fact that 
the uniform distribution on $S$ contains $|S \cap S_i|$ samples from $G_i$, and most of these 
will lie in an interval of width $O(\Var_{X \sim G_i} [\iprod{X,v}]^{1/2})$.

The above sketch suggests that in \cref{lem:fake_var_upperbound_intro} 
we should obtain the bound $\V_{X \sim S} [\iprod{X,v}] \Paren{\frac{|S_i \cap S|}{|S|} }^2 \leq O(\V_{X \sim G_i} [\iprod{X,v}])$.
Our actual result, however, loses an additive $\tau$ and a multiplicative factor of $\tau'(\tau)$.
This is because the distribution $X \sim S$ is not actually Gaussian, but instead has its first $K(k)$ moments close to those of a Gaussian.
To prove our lemma, we need to use a version of anti-concentration which applies to any such distribution (and which has an SoS proof) \cite{karmalkar2019list,raghavendra2020list,bakshi2020list,raghavendra2020listB}.
The SoS-provable anti-concentration inequality we use, which relies only on the first $1/\tau'(\tau)$ moments,
loses these additive and multiplicative factors.
However, since we only wish to conclude that $|S_i \cap S| |S_j \cap S|/ |S|^2 \leq 1/\poly(k)$, 
this ultimately suffices for our arguments.

Moving on to the second lemma, the following lower bound on the variance of $\iprod{X,v}$ contradicts the upper bound of \cref{lem:fake_var_upperbound_intro} when $G_i,G_j$ are in cases (i),(ii), unless one of $|S_i \cap S|, |S_j \cap S|$ is small.

\begin{lemma}[Informal version of \cref{lem:fake_var_lowerbound}, variance lower bound on $\iprod{x,v}$]
\label[lemma]{lem:fake_var_lowerbound_intro}
\begin{equation}\label{eqn:fake-var-lb-intro}
  \Var_{X \sim S}[ \iprod{X,v} ] \gtrsim \frac 1 {K(k)} \cdot \Paren{\frac{|S_i \cap S| \cdot |S_j \cap S|}{|S|^2}} \cdot \Var_{X \sim 0.5 G_i + 0.5 G_j} [\iprod{X,v}]\mper
\end{equation}
Furthermore, \eqref{eqn:fake-var-lb-intro} has a constant-degree SoS proof, 
where the indeterminates are assignment variables for the subset $S$.
\end{lemma}
The proof of \cref{lem:fake_var_lowerbound_intro} follows by arguing that the samples from $S_i, S_j$ in $S$ 
contribute $\Var_{X \sim 0.5 G_i + 0.5 G_j} [\iprod{X,v}]$ to the variance.

To complete this proof sketch of \cref{lem:ident-intro} in cases (i) and (ii), 
consider first what happens when $G_i,G_j$ are in case (i).
By direct calculation, $\Var_{X \sim 0.5 G_i + 0.5 G_j} [\iprod{X,v}] \geq \Omega(\iprod{\mu_{G_i} - \mu_{G_j},v}^2)$, so by \cref{lem:fake_var_lowerbound_intro},
\[
\Var_{X \sim S}[ \iprod{X,v} ] \gtrsim \frac 1 {K(k)} \cdot \Paren{\frac{|S_i \cap S| \cdot |S_j \cap S|}{|S|^2}} \cdot \Omega(\iprod{\mu_{G_i} - \mu_{G_j},v}^2) \;,
\]
and putting this together with \cref{lem:fake_var_upperbound_intro}, we get
\[
\tau + K(k) \cdot \frac{\iprod{v, \Sigma_{G_i} v} + \iprod{v,\Sigma_{G_j} v}}{\iprod{\mu_{G_i} - \mu_{G_j},v}^2} \gtrsim \Paren{\frac{|S_i \cap S| \cdot |S_j \cap S|}{|S|^2}}^2\mper
\]
When $G_i,G_j$ are in case $(i)$, with sufficiently large $TV(G_i,G_j) \geq 1-\delta(k)$ and sufficiently-small $\tau$, 
the left-hand side above is at most $\delta'(k) K(k)$, which is what we wanted to show.
For case (ii), the same argument applies, using the bound $\Var_{X \sim 0.5 G_i + 0.5 G_j} [\iprod{X,v}] \geq \iprod{v, \Sigma_i v}$.

Let us pause to note that the aforementioned discussion 
gives a simple \emph{but not yet SoS} proof of \cref{lem:ident-intro} in cases (i) and (ii).
Interestingly enough, we are not aware of a way to put together \cref{lem:fake_var_upperbound_intro,lem:fake_var_lowerbound_intro} to obtain \cref{lem:ident-intro} within the SoS proof system.
After our description of case (iii) below, we describe a modification of \cref{lem:ident-intro} 
for which we can obtain a low-degree SoS proof and which suffices for our purposes.

\subparagraph{Case (iii): separation by degree-2 polynomial}
Let us turn to case (iii) (\cref{lem:covar_sep}), where the covariance matrices $\Sigma_i, \Sigma_j$ of $G_i,G_j$ 
have large difference in relative Frobenius norm --- i.e., $\|I - \Sigma_i^{-1/2} \Sigma_j \Sigma_i^{-1/2} \|_F \geq \delta'(\delta)$, 
if $\dtv(G_i,G_j) \geq 1-\delta$.
Because we have handled cases (i) and (ii), we may now assume that for each $\ell \in [k]$ such that $G_i,G_\ell$ are in cases (i) or (ii), if $|S \cap S_i|/|S| \geq \delta'(\delta) K(k) + k^{-10}$ then $|S \cap S_\ell|/|S| \leq \delta'(\delta) K(k) + k^{-10}$.

For this overview, let us imagine that in fact for every $\ell$ such that $G_i,G_\ell$ are in cases (i),(ii), 
we actually have $|S \cap S_\ell| = 0$, as this does not change the argument significantly.
This implies, roughly speaking, that
\begin{align}
\label{eq:intro-1}
\Sigma_S \precsim \frac{\Sigma_i}{\delta'(\delta)} \;,
\end{align}
since $S \subseteq [n]$ is supported only on true clusters $S_\ell \subseteq [n]$ 
such that $\iprod{\mu_\ell - \mu_i,v}^2 \leq \iprod{v,\Sigma_i + \Sigma_\ell v}/\delta'(\delta)$ 
and $\Sigma_\ell \cdot \delta'(\delta) \preceq \Sigma_i \preceq \Sigma_\ell / \delta'(\delta)$.

Similarly to cases (i) and (ii), we now establish upper and lower bounds 
on the variance of a witnessing polynomial $p$ under the (centered) uniform distribution on $S$ 
which cannot hold simultaneously if $|S \cap S_i|,|S \cap S_j| \geq \delta'(\delta) K(k) + k^{-10}$.
We take our witnessing polynomial to be $p(x) = \iprod{x,Ax}$, 
where $A = \Sigma_i^{-1/2} (I - \Sigma_i^{-1/2} \Sigma_j \Sigma_i^{-1/2}) \Sigma_i^{-1/2}$.
(We encourage the reader to think of the case that $\Sigma_i = I$, in which case $A = \Sigma_i - \Sigma_j$.)

Our lower bound on the variance of $p(x)$ comes from our assumption that $S$ has Gaussian-like low-degree 
moments and the following formula for the variance of degree-$2$ polynomials under mean-zero Gaussian distributions: 
For any $\Sigma$ and any matrix $M$
\[
\V_{\cN(0,\Sigma)}[\iprod{x,Mx}] = \|\Sigma^{1/2} M \Sigma^{1/2}\|_F^2\mper
\]
It follows that (see \cref{lem:fake_var_upperbound_quadratic})
\begin{align}
\label{eq:intro-2}
  \V_{x \sim \cN(0,\Sigma_S)}[p(x)] = \|\Sigma_S^{1/2} \Sigma_i^{-1/2} (I - \Sigma_i^{-1/2} \Sigma_j \Sigma_i^{-1/2}) \Sigma_i^{-1/2} \Sigma_S \|_F^2 \leq \frac 1 {\delta'(\delta)} \|I - \Sigma_i^{-1/2} \Sigma_j \Sigma_i^{-1/2}\|_F^2\mcom
\end{align}
where for the inequality we have used \cref{eq:intro-1}.
On the other hand, by similar logic to \cref{lem:fake_var_lowerbound_intro}, we 
show (see \cref{lem:fake_var_lowerbound_quadratic})
\begin{align}
\label{eq:intro-3}
  \V_{x \sim \cN(0,\Sigma_S)} [p(x)] = \V_{X \sim S}[p(X - \mu_S)] \gtrsim \frac{|S_i \cap S||S_j \cap S|}{|S|^2} \cdot \iprod{\Sigma_i - \Sigma_j, A}^2 = \|I - \Sigma_i^{-1/2} \Sigma_j \Sigma_i^{-1/2}\|_F^4 \mper
\end{align}
It follows from \cref{eq:intro-2,eq:intro-3} that
\[
\frac{|S_i \cap S||S_j \cap S|}{|S|^2} \lesssim \delta'(\delta)\mcom
\]
which proves \cref{lem:ident-intro} in case (iii).
In \cref{lem:covar_sep}, we carry out this reasoning using SoS-provable 
versions of \cref{eq:intro-1,eq:intro-2,eq:intro-3}.

\paragraph{Modifying Lemma~\ref{lem:ident-intro} for SoS-Friendliness}
The preceding discussion has sketched a simple cluster identifiability proof.
To obtain our main clustering algorithm, we translate this proof into the SoS proof system.
Generally speaking, using powerful tools like SoS-provable anti-concentration and \cite{karmalkar2019list,raghavendra2020list,bakshi2020list,raghavendra2020listB} 
and the SoS proof of hypercontractivity of low-degree polynomials under the Gaussian measure \cite{barak2012hypercontractivity}, this task is largely mechanical.
However, as we have already alluded to, a technical challenge arises when we put together our upper and lower bounds (\cref{lem:fake_var_upperbound_intro,lem:fake_var_lowerbound_intro}) on the variance of linear functions $\iprod{x,v}$ to prove \cref{lem:ident-intro} in cases (i) and (ii).
We now describe this challenge and our resolution to it; see also \cref{sec:concurrent}.

Let us re-examine the proof sketch of \cref{lem:ident-intro} in cases (i) and (ii), 
which we gave using \cref{lem:fake_var_upperbound_intro,lem:fake_var_lowerbound_intro}.
Formally, this proof required multiplying both sides of~\cref{eqn:fake-var-lb-intro} by the quantity 
$(|S_i \cap S|/|S| - \delta)$ 
and then using the upper bound on $\V_{X \sim S} [\iprod{X,v}] (|S_i \cap S|/|S| - \delta)$ 
afforded by \cref{eqn:fake-var-ub-intro}.
This multiplication is valid without changing the sign of the inequality only if $|S_i \cap S|/|S| - \delta \geq 0$.\footnote{This subtlety was overlooked in a previous version of this paper -- see \cref{sec:concurrent}.}
This step did not present an obstacle in our non-SoS proof, since if $|S_i \cap S|/|S| \leq \delta$ we are already done, but it turns out to be a technical obstacle in our SoS proof.

To avoid this multiplication, we instead employ a modification of \cref{lem:ident-intro} for cases (i) and (ii).
We first replace the conclusion of \cref{lem:fake_var_upperbound_intro} with 
\[
\V_{X \sim S}[\iprod{X,v}] \cdot \frac{|S_i \cap S|}{|S|} \leq \tau'(\tau) \V_{X \sim G_i}[\iprod{X,v}] + \tau \V_{X \sim S}[\iprod{X,v}]\mper
\]
This carries the advantage that the $|S_i \cap S|/|S| \geq 0$ (unlike $|S_i \cap S|/|S| - \delta$), so we can easily combine this inequality with \cref{lem:fake_var_lowerbound_intro} to obtain a modified \cref{lem:ident-intro}, which concludes (roughly speaking)
\[
\frac{|S \cap S_i| |S \cap S_j|}{|S|^2} \lesssim \delta'(\delta)K(k) + \tau \cdot \frac{\V_{X \sim S}[\iprod{X,v}]}{\V_{X \sim 0.5 G_i + 0.5 G_j}[\iprod{X,v}]}\mcom
\]
where $v$ is the direction of hyperplane or variance separation between components $G_i,G_j$.
(As usual, see \cref{lem:mean_sep,lem:var_sep} for the formally correct version of this statement.)

Of course, with the above modification, we have the disadvantage that $\V_{X \sim S}[\iprod{X,v}]$ appears in the inequality -- this strategy will work only if we can obtain an \emph{a priori} upper bound on $\V_{X \sim S}[\iprod{X,v}]$.\footnote{Bakshi and Kothari \cite{BK20a} also note that an \emph{a priori} upper bound on $\V_{X \sim S}[\iprod{x,v}]$ is needed to give an SoS proof of an inequality like \cref{lem:ident-intro}. We emphasize that 
our techniques to obtain such an upper bound are significantly different from theirs. See \cref{sec:concurrent}.}
Since we may choose $\tau$ to be quite small, a relatively weak bound will suffice.
For this we take the following strategy.

We show (\cref{cor:TV_param}) that every mixture of $k$ Gaussians where every pair of components $G_i,G_j$ has small overlap and some pair is either in case (i) or case (ii) can in fact be partitioned into two nontrivial subsets $A,B$ of components such that one of the following holds.
\begin{itemize}
\item[(a)] There exists a hyperplane $v$ approximately separating the \emph{mixture} -- every pair of components $G_a \in A$ and $G_b \in B$ has $\iprod{\mu_a - \mu_b,v} \gg \sqrt{\iprod{v, \Sigma_a v}} + \sqrt{\iprod{v, \Sigma_b v}}$, and furthermore $\iprod{\mu_a - \mu_b,v} \geq f(k) \cdot \iprod{v, \Sigma_{\text{mix}} v}$, where the last quantity is the variance of the entire mixture in the direction $v$, for some function $f(k) > 0$.
That is, the mean separation between pairs $G_a \in A, G_b \in B$ is comparable to the overall variance of the mixture.

\item[(b)] The mixture can be separated into high-variance and low-variance components -- every pair of components $G_a \in A$ and $G_b \in B$ has $\iprod{v,\Sigma_a v} \gg \iprod{v, \Sigma_b v}$, and furthermore $\iprod{v, \Sigma_a v} \gg f(k) \iprod{v, \Sigma_{\text{mix}} v}$.
That is, the variance of the large-variance components is comparable to the variance of the entire mixture.

\end{itemize}
The above statements turn out to be a relatively straightforward corollary of our argument 
that any pair of Gaussians with small overlap are in one of cases (i),(ii),(iii) (\cref{lem:TV_param}).
Moreover, since it is not hard to argue that $\V_{X \sim S} [\iprod{X,v}] \leq F(k) \iprod{v, \Sigma_{\text{mix}} v}$, this modified statement gives our desired upper bound on $\V_{X \sim S}[\iprod{X,v}]$, 
allowing us to prove a modified version of \cref{lem:ident-intro}, 
which applies only to pairs $i, j$ on opposing sides of some nontrivial partition of the mixture.
We show that, while such a statement may not be enough to recover all the clusters at once, it does suffice (together with the SoS clustering approach) to split the mixture into two non-trivial sub-mixtures, which can then be clustered recursively.

This concludes our discussion of our SoS identifiability proof.

\paragraph{Robustification}
We now briefly discuss how the foregoing cluster identifiability proof (and its SoS counterpart) 
can be made {\em robust} to the setting that a small constant fraction of samples 
are replaced with adversarially-chosen outliers.
The resulting robust identifiability proof will lead immediately to a robust learning algorithm.

Our robustification of the above SoS proof uses a key insight leveraged in the early work 
on robust statistics~\cite{DKKLMS16} and first brought into the SoS framework by \cite{HopkinsL18,KothariSS18}.
Informally, this insight is that if a set of samples $S \subseteq [n]$ has bounded $2t$-th moments, 
then the influence of any $\e$-fraction of samples on its $t$-th moments is bounded by a dimension-independent function of $\e$.

Slightly more formally, we repeatedly use the following reasoning to robustify our arguments.
The following lemma, versions of which appear in several previous works, can be proved straightforwardly 
by the Cauchy-Schwarz inequality:

\begin{lemma}[Informal, see \cref{lem:mom_decorruption_diff}]
\label[lemma]{lem:decorrupt-intro}
Let $S,T \subseteq \R^d$, $|S| = |T| = m$, be such that $|S \cap T| \geq (1-\e)m$.
For every function $f \, : \, \R^d \rightarrow \R$, we have that
  \[
  \Abs{\E_{X \sim S} f(X) - \E_{X \sim T} f(X) } \leq O(\sqrt{\epsilon}) \cdot \Brac{ \Paren{\E_{X \sim S} f(X)^2}^{1/2} + \Paren{\E_{X \sim T} f(T)^2 }^{1/2} }\mper
  \]
\end{lemma}
Some comments are in order.
In a typical application of Lemma~\ref{lem:decorrupt-intro}, the set $S$ will be a subset of a set of \emph{corrupted} samples 
and the set $T$ will be a subset of uncorrupted samples.
If $f$ is taken as a low-degree polynomial and the moments of the uniform distribution on $S$ and $T$ are bounded, 
then the right-hand side of the conclusion of \cref{lem:decorrupt-intro} is small.

\new{Importantly, when the sets $S$ and $T$ have Gaussian-like low-degree moments (as in our clustering identifiability proof), 
\cref{lem:decorrupt-intro} gives multiplicative error guarantees.}
In particular, for polynomials $f$ of constant degree, we will have the hypercontractivity property that 
$\E_{X \sim S} f(X)^2 \leq O(\E_{X \sim S} f(X)^4)^{1/2}$ and similarly for $T$.
Hence, we will have
\[
\Abs{ \E_{X \sim S} f(X)^2 - \E_{X \sim T} f(X)^2 } \leq O(\sqrt{\epsilon}) \cdot \Brac{ \E_{X \sim S} f(X)^2 + \E_{X \sim T} f(T)^2 } \;,
\]
and therefore
\begin{align}
\label{eq:intro-4}
\E_{X \sim S} f(X)^2 = (1 \pm O(\sqrt{\epsilon})) \E_{X \sim T} f(X)^2
\end{align}
for low-degree polynomials $f$.
Thus, making our cluster identifiability proof robust just requires that (a) it can be accomplished using squares of low-degree polynomials and (b) it is robust to small multiplicative errors as in \cref{eq:intro-4}.
We have already seen most of (a) in the foregoing overview, and (b) is not difficult to attain in light of the 
simplicity of our identifiability proof.

%
%
%
%
%
\subsection{Related and Prior Work} \label{ssec:related}
In this subsection, we provide a detailed summary of the most relevant prior work.

The algorithmic question of designing robust estimators in high-dimensions has been extensively studied in recent years. After the initial papers~\cite{DKKLMS16, LaiRV16}, 
a number of works~\cite{DKK+17, DiakonikolasKKLMS18, BDLS17, SteinhardtCV18, ChengDKS18, 
ChengDR18, KlivansKM18, DiakonikolasKS18-nasty, DKS19-lr, DHL19, DepLec19, CDGW19, PrasadSBR2018, DiakonikolasKKLSS2018sever, DKKPS19-sparse} have obtained efficient algorithms 
for a range of unsupervised and supervised probabilistic models.
Alas, the question of robustly learning even (two) separated Gaussians has remained open.

The known non-robust parameter estimation algorithms for $k$-mixtures of arbitrary
Gaussians~\cite{KMV:10, MoitraValiant:10} proceed by reducing the problem to a collection of univariate problems. This is achieved by taking many random one-dimensional projections, solving each univariate problem, and piecing together the information obtained from all these projections. 
The univariate problem is solved using the method of moments. 
Unfortunately, the accuracy required for each univariate problem for this approach to work
is inverse polynomial in the dimension $d$, which is information-theoretically impossible
to achieve in the presence of even a sub-constant fraction of outliers. In summary,
this approach is highly non-robust.

In the robust setting, significant progress has been made for mixtures of {\em spherical} Gaussians.
The work of~\cite{DKKLMS16} gave a robust density estimation algorithm for 
for a mixture of (any constant number of) spherical Gaussians. 
More recently,~\cite{HopkinsL18, KothariSS18, DiakonikolasKS18-mixtures}
obtained efficient robust parameter estimation algorithm for mixtures of spherical Gaussians 
under near-optimal separation assumptions. Our SoS-based clustering framework is identical
to that of~\cite{HopkinsL18}. Our contribution lies in our construction of a low-degree identifiability 
proof for the clusters that can handle arbitrary separated Gaussians. 

It should be emphasized that the identifiability proofs in~\cite{HopkinsL18, KothariSS18} 
are (essentially) constrained to the spherical setting and in particular cannot even handle 
(non-robust) parameter estimation of two hyperplane separated Gaussians. 
This is due to their definition of a good cluster that only imposes 
upper bounds on the low-degree moments of the clusters. 

\new{
Very recent work of~\cite{JV19} obtains robust learning algorithms for mixtures of {\em two} 
separated Gaussians, going beyond the spherical case of~\cite{HopkinsL18, KothariSS18}. 
We note that the algorithm of~\cite{JV19} only recovers a special case of our $k=2$ result.
In particular, it does not capture the second separation scenario in Figure~\ref{fig:separation-types}.}


Finally, we note that~\cite{DKS17-sq} gave an SQ lower bound, which provides evidence that 
an exponential dependence on $k$ is required for the sample complexity and runtime
of our problem, even for the hyperplane separated case without outliers.

\subsection{Concurrent Work}
\label[section]{sec:concurrent}

In independent work, Bakshi and Kothari obtained an algorithm with similar guarantees to the one we present here, also using the ``SoS proofs to algorithms'' method, and similarly devising an identifiability proof using upper and lower bounds on the variance of witnessing polynomials~\cite{BK20a}.
The algorithm of Bakshi and Kothari runs in time $d^{\log(\kappa) \poly(k/\eta)}$, 
where $\kappa = \max_v (\max_{i \in [k]} \iprod{v, \Sigma_i v} / \iprod{v, \Sigma_j v})$ measures the ``spread'' of the mixture of Gaussians $G_1,\ldots,G_k$, assuming that the components in the mixture $\dtv(G_i,G_j) \geq 1-2^{-\poly(k/\eta)}$, and outputs a clustering that misclassifies at most $O(k \eps)+\eta$ fraction of the points.
(We note that the parameter $\kappa$ need not be bounded by any function of $k$ or $d$.)
The running time of our algorithm avoids the dependence on $\kappa$, at the cost of worse 
dependence on $k$: the running time of our algorithm scales as $d^{F(k)}$ for some large function $F$.
We can na\"ively bound $F(k) \leq k^{k^{\ldots^k}}$, an exponential tower of height $\poly(k)$, 
but we have not attempted to optimize this bound.
Additionally, our analysis requires $\dtv(G_i,G_j) \geq 1-F(k)$, while \cite{BK20a} require the weaker $\dtv(G_i,G_j) \geq 1-2^{-\poly(k)}$.

An earlier version of our paper claimed an algorithm to robustly learn mixtures of $k$ separated 
Gaussians with running time $(d/\eps)^{\poly(k)}$.
However, our analysis of this algorithm contained a gap in the translation of our simple identifiability proof into an SoS identifiability proof, which we became aware of after seeing the manuscript~\cite{BK20a}.
Bakshi and Kothari point out that their SoS identifiability proof employing (in the language of this paper) the witnessing polynomial $p(x) = \iprod{v,x}$ for some $v \in \R^d$ works only in the presence of an \emph{a priori} upper bound on the maximum variance of any component of the mixture in the direction $v$. \cite{BK20a} obtain such a bound via the spread parameter $\kappa$ (see \cite{BK20a}, Lemmas 4.12 and 4.13).
The same observation applies to our SoS proof, although this was overlooked in an earlier version of this paper.

We resolve this technical issue differently from \cite{BK20a}, avoiding the dependence on the parameter $\kappa$ but incurring larger running time and sample complexity dependence on $k$.
See the above discussion on making \cref{lem:ident-intro} SoS-friendly.

Independently of our work, Bakshi and Kothari were able to extend their techniques to obtain an algorithm with running time that does not depend on the spread parameter $\kappa$, 
at the expense of a slightly worse clustering guarantee.
In particular, they give an algorithm with running time and sample complexity $d^{\poly(k/\eta)}$ that outputs a clustering misclassifying at most $O(k^{O(k)}( \eps+\eta))$ fraction of the points, assuming pairwise separation 
$1-2^{-k^{\Omega(k)}}$\cite{BK20b}.

\subsection{Organization}

The structure of the paper is as follows:
In Section~\ref{sec:prelims}, we provide the necessary definitions and technical facts.
In Section~\ref{sec:rough-clustering}, we give our main clustering algorithm and state our main technical lemmas.
The proofs of these lemmas are given in Sections~\ref{sec:sep-lemmas},~\ref{sec:tools}, and~\ref{sec:variance-bounds-proofs}.
Section~\ref{sec:tv-param} proves our structural result on separation in total variation distance.
Finally, Section~\ref{sec:constant-accuracy-and-main-theorem} shows that a rough clustering suffices and puts everything
together to prove our main result.
Some omitted proofs have been deferred to an Appendix.


\section{Preliminaries} \label{sec:prelims}

\paragraph{Notation} 
For $n \in \N$, we will use $[n]$ to denote $\{ i \in \N \mid i \leq n \}$. 
We use $\norm{ \cdot }$ to denote the $\ell_2$ (Euclidean) norm and $\norm{ \cdot}_F$ for the Frobenius norm. 
We will use $\iprod{\cdot, \cdot}$ to denote the Euclidean inner product for vectors and the standard trace product for matrices. 
For two polynomials $p, q$, the notation $p \preceq q$ means that $q-p$ is a sum of square polynomials. 
We will use $\cdot ^{\otimes s}$ to denote the standard Kronecker product. 
For $a, b \in \R$, we will write $a \gg b$ (or $b \ll a$) to mean that there exists a sufficiently large
constant $C>0$ such that $a \geq C b$.

\medskip

Throughout this work, we focus on the following strong contamination model that has been
extensively studied in prior work (see, e.g.,~\cite{DKKLMS16}).

\begin{definition}[Strong Contamination Model, $\epsilon$-corrupted]
We say that a set of vectors $Y_1,\ldots,Y_n$ is an $\e$-corrupted set of samples 
from a distribution $D$ on $\R^d$ if it is generated as follows: 
First, $X_1,\ldots,X_n$ are sampled i.i.d. from $D$. 
Then a (malicious, computationally unbounded) adversary 
observes $X_1,\ldots,X_n$, replaces any $\e n$ of them with any vectors she likes, 
then reorders the vectors arbitrarily to yield $Y_1,\ldots,Y_n$.
\end{definition}

We will use the Sum of Squares (SoS) proof system extensively.
We refer the reader to \cite{BarakSteurerNotes} for a complete treatment of basic definitions 
about the Sum of Squares hierarchy and Sum of Squares proofs.
Here we review the basics briefly.

\begin{definition}[Symbolic polynomial]
A degree-$t$ symbolic polynomial $p$ is a collection of indeterminates $\widehat{p}(\alpha)$, 
one for each multiset $\alpha \subseteq [n]$ of size at most $t$.
We think of it as representing a polynomial $p \, : \, \R^n \rightarrow \R$ whose coefficients are themselves 
indeterminates via $p(x) = \sum_{\alpha \subseteq [n], |\alpha| \leq t} \widehat{p}(\alpha) x^\alpha$.
\end{definition}

\begin{definition}[SoS Proof]
Let $x_1,\ldots,x_n$ be indeterminates and let $\cA$ be a set of polynomial inequalities $\{ p_1(x) \geq 0,\ldots,p_m(x) \geq 0 \}$.
  An SoS proof of the inequality $r(x) \geq 0$ from axioms $\cA$ is a set of polynomials $\{r_S(x)\}_{S \subseteq [m]}$ such that each $r_S$ is a sum of square polynomials and
  \[
  r(x) = \sum_{S \subseteq [m]} r_S(x) \prod_{i \in S} p_i(x)\mper
  \]
  If the polynomials $r_S(x) \cdot \prod_{i \in S} p_i(x)$ have degree at most $d$, we say that this proof is degree $d$.
  We write $\cA \proves{d} r(x) \geq 0$.
  When we need to emphasize what indeterminates are involved in a particular SoS proof, we sometimes write $\cA \proves{d}^x r(x) \geq 0$.
  We also often refer to $\cA$ containing polynomial equations $q(x) = 0$, by which we mean that $\cA$ contains both $q(x) \geq 0$ and $q(x) \leq 0$.
\end{definition}
We frequently compose SoS proofs without comment -- see \cite{BarakSteurerNotes} for basic facts about composition of SoS proofs and bounds on the degree of the resulting proofs.
Our algorithm also uses the dual objects to SoS proofs, \emph{pseudoexpectations}.
\begin{definition}
  Let $x_1,\ldots,x_n$ be indeterminates.
  A degree-$d$ pseudoexpectation $\pE$ is a linear map $\pE \, : \, \R[x_1,\ldots,x_n]_{\leq d} \rightarrow \R$ from degree-$d$ polynomials to $\R$ such that $\pE p(x)^2 \geq 0$ for any $p$ of degree at most $d/2$ and $\pE 1 = 1$.
  If $\cA = \{p_1(x) \geq 0,\ldots,p_m(x) \geq 0\}$ is a set of polynomial inequalities, we say that $\pE$ satisifies $\cA$ if $\pE s(x)^2 p_i(x) \geq 0$ for all squares $s(x)^2$ such that $s(x)^2 \cdot p_i(x)$ has degree at most $d$.
\end{definition}

Finally, we will rely on the following algorithmic fact.

\begin{theorem}[The SoS Algorithm, see \cite{BarakSteurerNotes}]
There is an algorithm which takes a natural number $d$ and a satisfiable system of polynomial inequalities $\cA$ in variables 
$x_1,\ldots,x_n$ with coefficients at most $2^n$ containing an inequality of the form $\|x\|^2 \leq M$ for some real number $M$ 
and returns in time $n^{O(d)}$ a degree-$d$ pseudoexpectation $\pE$ which satisfies $\cA$ 
up to error $2^{-n}$.\footnote{We did not define what it means for $\pE$ to satisfy $\cA$ up to error $2^{-n}$. 
The idea is that $2^{-n}$ slack is added to each constraint. Since the coefficients in all the SoS proofs in this paper 
have magnitude at most $n^{O(1)}$, these $2^{-n}$ errors are negligible and we will not treat them explicitly. 
See \cite{BarakSteurerNotes} for further discussion.}
\end{theorem}

We will also use the following Cauchy-Schwarz inequality for pseudoexpectations:

\begin{fact}[Cauchy-Schwarz for Pseudo-distributions]
	Let $f,g$ be polynomials of degree at most $d$ in indeterminate $x \in \R^d$. Then, for any degree d pseudoexpectation $\tmu$,
	$\pE_{\tmu}[fg] \leq \sqrt{\pE_{\tmu}[f^2]} \sqrt{\pE_{\tmu}[g^2]}$.
	\label{fact:pseudo-expectation-cauchy-schwarz}
\end{fact} 

The following fact is a simple corollary of the fundamental theorem of algebra:
\begin{fact} \label{fact:univariate}
For any univariate degree $d$ polynomial $p(x) \geq 0$ for all $x \in \R$, 
$\sststile{d}{x} \Set{p(x) \geq 0}$.
\end{fact}

We also use the following SoS proof; it is straightforward to prove by induction.
\begin{fact}\label[fact]{fact:sos-1}
  For all $s \in \N$ and $B \in [0,1]$,
  $\{0 \leq y \leq 1, a \geq 0, xy \leq a + Bx\} \proves{O(s)}^{x,y,a} xy^s \leq sa + B^s x$.
\end{fact}


\section{Robust Clustering to Constant Accuracy}\label[section]{sec:rough-clustering}

In this section, we describe and analyze our main robust clustering algorithm, as captured by \cref{lem:const_acc}.
For simplicity we present the case of uniform mixtures -- the modifications necessary in the case of mixtures with general weights are described in \cref{sec:arbit-weights}.
We also assume that our set of samples $X_1,\ldots,X_n \sim \frac 1 k \sum_{i=1}^k G_i$ consists of exactly $n/k$ samples from each Gaussian $G_i$ -- since our algorithm is robust to a small constant fraction of adversarial errors, guarantees in this sampling model imply the same guarantees in the usual model of i.i.d. samples.

\begin{lemma}[Constant-Accuracy Robust Clustering]\label[lemma]{lem:const_acc}
For every $\rho >0$ and $k \in \N$ there are $F(\rho,k), f(\rho,k)$ and an algorithm (\cref{alg:cluster}) which,
given an $\eps$-corrupted set of samples $X_1, \dots, X_n$ drawn from a mixture of $d$-dimensional Gaussians $\frac{1}{k}\sum_{i=1}^k G_i$ such that $\min_{i \neq j} \dtv(G_i, G_j) \geq 1-f(\rho,k)$ and such that $\eps < f(\rho,k)$ and $n \geq d^{F(\rho,k)}$,
the algorithm runs in time $n^{F(\rho,k)}$, and with probability $1-1/\poly(d)$ returns a list of at most $F(\rho,k)$ partitions $T_1,\ldots,T_k$ of $[n]$ into sets of size $n/k$ such that for at least one partition $T_1,\ldots,T_k$ in the list there is a permutation $\pi \, : \, [k] \rightarrow [k]$ with $|T_{\pi(i)} \cap S_i| \geq \tfrac n k (1 - \rho)$, where $S_1, \dots, S_k$ are subsets of $[n]$ that index samples drawn from $G_1, \dots G_k$.
\end{lemma}

To describe our algorithm, we need a system of polynomial inequalities which captures the robust clustering problem.
This is given in the following definition. 
  \begin{definition}[Clustering axioms $\cA$]
  \label[definition]{def:axioms}
    Let $X_1,\ldots,X_n \in \R^d$.
    Let $w_1,\ldots,w_n$, $z_1,\ldots,z_n$, $X_1',\ldots,X_n'$ and $\Sigma_{i,j}$, $\Sigma^{1/2}_{i,j}$, $\Sigma^{-1/2}_{i,j}$ for $i, j \in [d]$ be indeterminates ($\Sigma^{1/2}, \Sigma^{-1/2}$ will be thought of as $d \times d$ matrices of indeterminates).
    We define a system of polynomial inequalities in the variables $w,z,X'$ whose solutions correspond to subsets $S$ of $[n]$ and vectors $X_1',\ldots,X_n'$ such that (1) $X_i' = X_i$ for all but $\e n$ indices and (2) the low-degree moments of the uniform distribution on $\{ X_i' \, : \, i\in S \}$ are approximately Gaussian. More specifically, we have:
    \begin{enumerate}
      \item Let $t \in \N$ be even and let $\delta, \epsilon > 0$.
      \item Let $\cA_{\text{corruptions}} := \{ z_i^2 = z_i \}_{i \in [n]} \cup \{ z_i (X_i - X_i') = 0 \}_{i \in [n]} \cup  \{ \sum_{i \in [n]} z_i = (1-\e)(n/k)\}$.
      \item Let $\cA_{\text{subset}} := \{w_i^2 = w_i\}_{i \in [n]} \cup \left \{ \sum_{i=1}^n w_i = n/k \right \}$.
      \item Let $\mu(w) = \tfrac k n \sum_{i=1}^n w_i X_i'$.
      \item Let $\Sigma(w) = \tfrac k n \sum_{i=1}^n w_i (X_i' - \mu(w))(X_i' - \mu(w))^\top$.
      \item Let $\cA_{\text{matrices}} = \{ (\Sigma^{1/2})^2 = \Sigma(w), (\Sigma^{-1/2} \Sigma^{1/2})^2 = \Sigma^{-1/2} \Sigma^{1/2} \} \cup \{ \Sigma^{-1/2} \Sigma^{1/2} w_i \Paren{X_i' - \mu(w)}  = w_i \Paren{X_i' - \mu(w)} \}_{i \in [n]}$.
      \item Let $\cA_{\text{moments}}$ be the following collection of polynomial inequalities, for all  $s \leq t$:
      \[
      \Norm{\frac k n \sum_{i \in [n]} w_i [\Sigma^{-1/2} (X_i' - \mu(w))]^{\tensor s} - M_s}^2 \leq \delta \cdot d^{-2t}\mcom
      \]
      where $M_s = \E_{g \sim \cN(0,\Id)} g^{\tensor s}$ is the $s$-th moment tensor of the standard Gaussian.
    \end{enumerate}
    Finally, let $\cA(X_1,\ldots,X_n,t,\delta,\epsilon, k) = \cA_{\text{corruptions}} \cup \cA_{\text{subset}} \cup \cA_{\text{moments}}  \cup \cA_{\text{matrices}}$.
    Note that $\cA(X_1,\ldots,X_n,t,\delta,\epsilon, k)$ has degree at most $10t$.
  \end{definition}

\begin{algorithm}[htb]{}
	\begin{algorithmic}[1]
		\Function{Cluster}{$(X_1,\ldots,X_n,t,\delta,\epsilon,k,N)$}
                \State If $n \leq 1.1 N/k$, output $\{\{ \{ X_1,\ldots,X_n\} \}\}$.
                \State Else if $\textsc{Split}(\cA) = \textsc{Fail}$, return $\emptyset$.
                \State Else, let $R_1,\ldots,R_m = \textsc{Split}(\cA)$.
                \State For $S \subseteq [m]$, let $R_S = \cup_{a \in S} R_a$.
                \State Let $R_S'$ be the result of adding or removing samples from $R_S$ arbitrarily so that its size is an integer multiple of $n/k$.
                \State Return $\{ C_1 \times C_2 \, : \, S \subseteq [m], C_1 \in \textsc{Cluster}(R_S', t, \delta, \epsilon, k, N), C_2 \in \textsc{Cluster}(\overline{R_S'}, t, \delta, \epsilon, k, N) \}$
		\EndFunction
	\end{algorithmic}
	\caption{Clustering Algorithm: Returns a set of candidate clusterings of $X_1,\ldots,X_n$}
	\label{alg:cluster}
\end{algorithm}
\noindent

Our algorithm \textsc{cluster} uses a key subroutine, \textsc{split}, to find $\approx 2^k$ candidate partitions of the samples into two subsets.
Then it recursively finds a set of candidate clusterings of each half of each candidate partition and puts together the candidate clusterings.
The key subroutine, \textsc{split}, described in the following lemma.
\textsc{split} solves a semidefinite programming (SDP) relaxation of the polynomial system $\cA(X_1,\ldots,X_n)$, then uses a simple randomized rounding algorithm to find subsets $R_1,\ldots,R_m$ of $X_1,\ldots,X_n$.
The rounding scheme is standard so we postpone the analysis to \cref{sec:omitted}.

These subsets $R_i$ will have the following guarantee: if there is a partition of the clusters $[k]$ into two nontrivial subsets $S,T$ such that every pseudoexpectation satifying $\cA$ cannot be simultaneously supported on $S$ and $T$, then each of the subsets $R_1,\ldots,R_m$ is (approximately) supported on only one of $S,T$.
Additionally, $\cup R_i$ contains almost all of the samples $X_1,\ldots,X_n$.
Thus, there is a partition of $R_1,\ldots,R_m$ which approximates the partition $S,T$; \textsc{cluster} finds this partition by brute-force search.
This yields corrupted samples from the submixtures represented by $S,T$ respectively, which \textsc{cluster} then handles recursively.

\begin{lemma}[See proof in \cref{sec:omitted}]
\label[lemma]{lem:rounding}
  Let $n,k \in \N$ with $k \leq n$ and let $\eta > 0$.
  There is an algorithm \textsc{split} with the following guarantees.
  Suppose that $S_1,\ldots,S_k \subseteq [n]$ partition $[n]$, with $|S_i| = n/k$ for all $i$.
  Suppose further that $S,T \subseteq [k]$ partition $[k]$.
  Let $\cC$ be a set of degree-$2$ pseudoexpectations in variables $w_1,\ldots,w_n$ (and perhaps auxiliary variables) such that for all $\pE \in \cC$,
  \begin{compactenum}
  \item $\pE$ satisfies $w_i^2 = w_i$ and $\sum_{i \leq n} w_i = n/k$
  \item $\pE (\sum_{i \in \cup_{a \in S} S_a} w_i)(\sum_{i \in \cup_{a \in T} S_a} w_i) \leq \delta n^2$.
  \end{compactenum}
  Finally, suppose that the uniform distribution over indicator vectors for $S_1,\ldots,S_k$ is contained in $\cC$.
  Let $S' = \cup_{a \in S} S_a$ and $T' = \cup_{a \in T} S_a$.

  Then \textsc{split} makes $m \leq O(k \log 1/\eta)$ queries to an oracle which optimizes linear functions over $\cC$, runs in additional $\poly(n,k,\log 1/\eta)$ time, and returns a list of subsets $R_1,\ldots,R_m \subseteq [n]$ such that
  $\E |\bigcup_{t \leq m} R_t| \geq n (1-\eta)$,
  and for all $t \leq m$,
  \[
  \E |R_t| = n/k \text{ and } \E \min(|R_t \cap S'|,|R_t \cap T'|) \leq \delta n\mcom
  \]
 where the expectation is taken over randomness used by the algorithm.
\end{lemma}

In light of \cref{lem:rounding}, the bulk of our efforts will be to show that for a set of corrupted samples $X_1,\ldots,X_n$ from a GMM $\tfrac 1 k \sum_{i \leq k} G_i$ where $\dtv(G_i,G_j) \geq 1 - \delta$ for some small $\delta$, there is a nontrivial partition $S,T \subseteq [k]$ of the mixture components such that $\pE (\sum_{i \in \cup_{a \in S} S_a} w_i ) (\sum_{i \in \cup_{a \in T} S_a} w_i)$ is close to $0$, for any $\pE$ satisfying the clustering axioms.

This will occur under the following deterministic conditions on the un-corrupted samples $X_1,\ldots,X_n$, which we will show hold with high probability.

\begin{definition}[Deterministic conditions for \textsc{cluster}]
\label[definition]{def:deterministic}
  Fix Gaussians $G_1, \dots, G_k$ on $\R^d$.
  For $\delta,\xi > 0$ and $t \in \N$, the $(\delta, \xi, t)$-deterministic conditions 
  (with respect to $G_1,\ldots,G_k$) on $X_1,\ldots,X_n \in \R^d$ are the following:
  \begin{enumerate}
  \item For each $i \in [k]$, there exists a partition $S_i$ of $\{X_1,\ldots,X_n\}$ into $k$ pieces each of size $n/k$ such that for all $s \leq t$,
  \[
  \Norm{\frac kn \sum_{j \in S_i} [\overline{\Sigma}_j^{-1/2} (X_j - \overline{\mu}_j)]^{\tensor s} - M_s}_F^2 \leq d^{-2t} \delta \;, \] 
  where $M_s = \E_{g \sim \cN(0,\Id)} g^{\tensor s}$ is the $s$-th moment tensor of the standard Gaussian and
  \[
  \overline{\mu}_i = \frac 1 {|S_i|} \sum_{j \in S_i} X_j \, , \text{ and } \overline{\Sigma}_i = \frac 1 {|S_i|} \sum_{j \in S_i} (X_j - \overline{\mu}_j)(X_j - \overline{\mu}_j)^\top \mper
  \]
  \item For $a \in [k]$, $v \in \R^d$, and $A \in \R^{d \times d}$, we define
  \begin{enumerate}
  	\item $E_a(v) := \{  X_i \in S_a \mid~\iprod{X_i - \mu_a, v}^2 \leq O(1) \log(1/\xi) \iprod{v, \Sigma_a v} \}$,
  	\item $F_{a}(v) := \{ (X_i, X_j) \in S_a^2 \mid~\iprod{X_i-X_j,v}^2 \geq \Omega(1) \cdot \xi \iprod{ v, \Sigma_a v}   \}$,
  	\item $G_a(A) := \{ (X_i, X_j) \in S_a^2 \mid \iprod{X_i - X_j, A (X_i - X_j)} = 2\iprod{\Sigma_a, A} \pm O(1) \cdot \log(1/\xi) \cdot \| \Sigma_a A\|_F \}$.
  \end{enumerate} 
  Then for every $v \in \R^d, A \in \R^{d \times d}$, we have $|E_a(v) \geq (1-\xi)(n/k)$ and $|G_a(A)|,|F_a(V)| \geq (1-\xi)(n/k)^2$.
\end{enumerate}
\end{definition}
We provide a proof of the following in \cref{sec:omitted} for completeness using only standard concentration arguments.
\begin{lemma}\label{lem:concentration}
  For all even $t \in \N$, if $n \geq \log(1/\gamma)^{Ct} d^{10kt}/\delta^2$ for some sufficiently large constant $C$ and $\xi \geq \delta$, then $X_1,\ldots,X_n \sim \tfrac 1 k (\sum_{i=1}^k G_i)$ sampled iid satisfy the deterministic conditions \cref{def:deterministic} with probability at least $1-\gamma$. 
\end{lemma}

Now we arrive at the key lemmas.
The first is a structural result which translates our small overlap assumption on each pair of Gaussians $G_i,G_j$ into information about the parameters $\mu_i,\Sigma_i,\mu_j,\Sigma_j$.

\begin{lemma}[See proof in~\cref{sec:tv-param}]
\label[lemma]{lem:gaussian-cases}\label{lem:TV_param}
Suppose that $P,Q$ are $d$-dimensional Gaussian distributions with $\dtv(P,Q) \geq 1 - \eps$.
Then one of the following holds:
\begin{enumerate}
        \item ($\eps$-Mean Separation) There is a unit vector $v \in \R^d$ such that
        \[
        |\iprod{v,\mu_P} - \iprod{v,\mu_Q}| \geq \frac{1}{10} (\log 1/\epsilon)^{1/6} \cdot \sqrt{ \iprod{v,(\Sigma_P + \Sigma_Q) v} }\mper
        \]
        \item ($\eps$-Variance Separation) There is a unit vector $v \in \R^d$ such that
        \[
        \max \{ \iprod{v,\Sigma_Q v}, \iprod{v, \Sigma_P v} \} \geq (\log 1/\epsilon)^{1/6} \cdot \min \{ \iprod{v, \Sigma_Q v}, \iprod{v, \Sigma_P v} \}
\mper        \]
        \item ($\eps$-Covariance Separation) We have that 
        \[
        \| I  - \Sigma_P^{-1/2}\Sigma_Q \Sigma_P^{-1/2}\|_F^2 \geq \Omega\Paren{\frac{\log 1/\epsilon}{\log \log 1/\epsilon}}\mper
        \]
\end{enumerate}
\end{lemma}

\cref{lem:TV_param} applies to two Gaussians at a time, but we can establish as a corollary that the entire mixture of $k$ Gaussians (where each pair has overlap at most $\e$) is also in one of three cases. 
Note that in the following, the distances between parameters of the components $G_i$ are related both to the overlap of pairs of components and to the covariance of the overall mixture.
The proof is a straightforward case analysis using \cref{lem:TV_param}, so we defer it to \cref{sec:omitted}.

\begin{corollary}\label[corollary]{cor:TV_param}
  Let $k \in \N$, $k \geq 2$.
  Suppose $G_1 = \cN(\mu_1,\Sigma_1),\ldots,G_k = \cN(\mu_k,\Sigma_k)$ Gaussians distributions.
  Let $\Sigma$ be the covariance of the uniform mixture of $G_1,\ldots,G_k$.
  For any $C, C' > 0$, if $\min_{i \neq j} \dtv(G_i, G_j) \geq 1 - 2^{-k^{10(C')^k + 10C + 10}}$, then we have at least one of:
  \begin{enumerate}
  \item The mixture is partitioned by a hyperplane.
  That is, there is a direction $v \in \R^d$ and a nontrivial partition of $[k]$ into $S,T$ such that for all $a \in S$ and $b \in T$,
  \[
 \iprod{\mu_a  - \mu_b, v}^2 \geq \max \left \{ k^C ( \iprod{v, \Sigma_a v} + \iprod{v, \Sigma_b v}), \frac{\iprod{v, \Sigma v}}{ k^2 } \right \}\mper
 \]
  \item The mixture can be partitioned into high-variance and low-variance components.
  That is, there is a direction $v \in \R^d$ and a nontrivial partition of $[k]$ into $S,T$ such that for all $a \in S$ and $b \in T$,
 \[
 \iprod{v, \Sigma_a v} \geq k^{C'} \iprod{v, \Sigma_b v} \text{ \quad and \quad }
 \frac{\iprod{v, \Sigma_b v}}{\iprod{v, \Sigma v}} \leq \Paren{\frac{\iprod{v, \Sigma_a v}}{\iprod{v, \Sigma v}}}^{C'}\mper
 \]

  \item The covariance of each component is comparable to $\Sigma$, and all pairs are covariance separated.
  That is, if $A_{ab} = \Sigma_a^{-1/2} (I - \Sigma_a^{-1/2} \Sigma_b \Sigma_a^{-1/2}) \Sigma_a^{-1/2}$, then every $a,b \in [k]$ has
  \[
  \|I - \Sigma_a^{-1/2} \Sigma_b \Sigma_a^{-1/2}\|_F^2 \geq \Omega(k^{(C')^k + C + 1}) \cdot \max \left \{ \|\Sigma_a^{1/2} A_{ab} \Sigma_a^{1/2} \|_F, \|\Sigma_b^{1/2} A_{ab} \Sigma_b^{1/2} \|_F, \|\Sigma^{1/2} A_{ab} \Sigma^{1/2}\| \right \} \mper
  \]
\end{enumerate}
\end{corollary}

Finally, the following lemmas, capturing the bulk of our technical work, give upper bounds on $\pE (\sum_{i \in \cup_{a \in S} S_a} w_i)(\sum_{i \in \cup_{a \in T} S_a} w_i)$ for $S,T$ being the partition of clusters from \cref{cor:TV_param} and any pseudoexpectation $\pE$ satisfying the clustering axioms $\cA$.

\paragraph{Notation}
  In the following \cref{lem:mean_sep,lem:var_sep,lem:covar_sep}, let $X_1,\ldots,X_n \in \R^d$ satisfy the $(\delta,\xi,t)$-deterministic conditions (\cref{def:deterministic}) with respect to Gaussians $G_1,\ldots,G_k$, where $G_i = \cN(\mu_i,\Sigma_i)$.
We assume that $\xi \leq \xi_0$ and $\tau \leq \tau_0$ for sufficiently small universal constants $\xi_0,\tau_0 > 0$.
  Let $S_1,\ldots,S_k$ be the partition of $X_1,\ldots,X_n$ guaranteed by \cref{def:deterministic}.
  Let $Y_1,\ldots,Y_n$ be an $\e$-corruption of $X_1,\ldots,X_n$ and let $\cA$ be the clustering axioms (\cref{def:axioms}) for $Y_1,\ldots,Y_n$.
  For indeterminates $w_1,\ldots,w_n$, let $\alpha_i(w) = \sum_{j \in S_i} w_j = \iprod{a_i,w}$, where $a_i \in \{0,1\}^n$ is the indicator for $S_i$.

The first lemma gives an upper bound on $\pE \alpha_a(w) \alpha_b(w)$ when components $G_a, G_b$ are mean separated.
\begin{lemma}[Mean Separated Components, see proof in \cref{sec:sep-lemmas}]
\label[lemma]{lem:mean_sep}
For every $\tau > 0$ there is $s = \tilde{O}(1/\tau^2)$ such if $\epsilon, \delta \leq s^{-O(s)} k^{-20}$, then for all $a,b \in [k]$ and all $v \in \R^d$ and all sufficiently-small $\rho > 0$, if $\iprod{\mu_a - \mu_b,v}^2 \geq \rho \E_{X,X' \sim \tfrac 1 k \sum_{i \leq k} G_i} \iprod{X - X',v}^2$,
\[
  \cA \proves{O(s)} \Paren{\frac{\alpha_a(w) \alpha_b(w)}{n^2}}^{s}  \leq s^{O(s)} \cdot (\log 1/\xi)^{O(s)} \cdot \Paren{ \frac{\iprod{v,\Sigma_a v} + \iprod{v, \Sigma_b v}}{ \iprod{\mu_a - \mu_b,v}^2 }}^{\Omega(s)}+ \rho^{-O(s)} (\tau^{\Omega(s)} + \epsilon^{\Omega(s)} k^{O(s)} s^{O(s^2)} + \xi^{\Omega(s)} ) \mper
\]
\end{lemma}

The second lemma gives an upper bound on $\pE \alpha_a(w) \alpha_b(w)$ in terms of the variance separation.
\begin{lemma}[Variance Separated Components, see proof in \cref{sec:sep-lemmas}]
  \label[lemma]{lem:var_sep}
For every $\tau > 0$ there is $s = \tilde{O}(1/\tau^2)$ such that if $\epsilon, \delta \leq s^{-O(s)} k^{-20}$,  then for all $a,b \in [k]$ and all $v \in \R^d$, if $\iprod{v,\Sigma_b v}^2 \geq \rho \E_{X,X' \sim \tfrac 1 k \sum_{i \leq k} G_i} \iprod{X - X',v}^2$,
\begin{align*}
\cA \proves{O(s)} \Paren{\frac{\alpha_a(w) \alpha_b(w)}{n^2}}^s \leq & \xi^{-O(s)} \cdot \Brac{ s^{O(s)} \Paren{\frac{\iprod{v, \Sigma_a v}}{\iprod{v, \Sigma_b v}}}^{\Omega(s)} + \rho^{-O(s)} (\tau^{\Omega(s)} + \e^{\Omega(s)} k^{O(s)} s^{O(s^2)} + \xi^{\Omega(s)} )}\mper
\end{align*}
\end{lemma}

The third lemma gives an upper bound on  $\pE \iprod{w,a_i}\iprod{w,a_j}$  in terms of the covariance separation.
\begin{lemma}[Covariance Separated Components, see proof in \cref{sec:sep-lemmas}]
\label[lemma]{lem:covar_sep}
Let $\Sigma_{\text{mix}}$ be the covariance of the mixture $\tfrac 1 k \sum_{i \leq k} G_i$.
If $\epsilon, \delta \leq k^{-O(1)}$, then for all $a,b \in [k]$ and $A \in \R^{d \times d}$,
\begin{align*}
\cA \proves{O(1)} & \Paren{\frac{\alpha_a(w) \alpha_b(w)}{n^2}}^{16} \leq O(\log 1/\xi)^8 \cdot \frac{\|\Sigma_{\text{mix}}^{1/2} A \Sigma_{\text{mix}}^{1/2}\|_F^8 + \|\Sigma_a^{1/2} A \Sigma_a^{1/2}\|_F^8 + \|\Sigma_b^{1/2} A \Sigma_b^{1/2}\|_F^8}{\iprod{\Sigma_a - \Sigma_b, A}^8} + O(\xi^4) + O(\e^2 k^{20})
\end{align*}
%
%
\end{lemma}

Now we can put everything together to prove \cref{lem:const_acc}.
To prove \cref{lem:const_acc}, we are going to rely on the following key proposition, which puts together the lemmas above to analyze the output of \textsc{split}.

\begin{proposition}\label[proposition]{prop:split}
  For every $\gamma > 0$ and $k \in \N$ there is $f(k,\gamma) > 0$ satisfing the following.
  If $k \geq 2$ and $G_1,\ldots,G_k$ are $d$-dimensional Gaussians such that $\min_{i \neq j} \dtv(G_i,G_j) \geq 1 - f(k,\gamma)$, then exists a nontrivial partition of $[k]$ into $A,B \subseteq [k]$ such that for all $\e,\delta,\xi \leq f(k,\gamma)$ and $t \geq 1/f(k,\gamma)$, if $Y_1,\ldots,Y_n \in \R^d$ satisfy the $(\delta,\xi,t)$ deterministic conditions with respect to $G_1,\ldots,G_k$, and hence are partitioned into ground-truth clusters $S_1,\ldots,S_k$, and $X_1,\ldots,X_n \in \R^d$ are any $\e$-corruption of $Y_1,\ldots,Y_n$, then given the set of degree-$t$ pseudoexpectations satisfying $\cA(X_1,\ldots,X_n,t,\delta,\e,k)$ and sufficiently-small error parameter $\eta \leq f(k,\gamma)$, with probability at least $1-1/k^{100}$ the algorithm \textsc{split} returns a list of subsets $R_1,\ldots,R_m \subseteq [n]$ such that there is a subset $S \subseteq [m]$ such that $R'$, the result of adding or removing elements arbitrarily from $\cup_{a \in S} R_a$ to make its size the nearest multiple of $n/k$, is a $\gamma$-corruption of $\cup_{a \in A} S_a$ -- that is, the samples $Y_i$ from the submixture defined by $A \subseteq [k]$.
\end{proposition}
\begin{proof}
  By \cref{cor:TV_param}, with $C = C(k,\gamma)$ and $C' = C'(k,\gamma)$ to be chosen later, the mixture $G_1,\ldots,G_k$ can be in one of three cases; we address each in turn.
  We will show that for any $\gamma$ of our choosing, by choosing sufficiently-large $C(k,\gamma), C'(k,\gamma)$ we can ensure that there is a partition $A,B \subseteq [k]$ such that any $\pE$ which satisfies $\cA$ has
\[
\pE \Paren{\sum_{i \in \cup_{a \in A} S_a} w_i}\Paren{\sum_{i \in \cup_{a \in B} S_a} w_i} \leq \gamma' n^2\mper
\]

Suppose first that case (1) of \cref{cor:TV_param} applies, and let $A,B$ be the given partition of $[k]$.
Then for every $a \in A, b \in B$, by \cref{cor:TV_param}, we have $\iprod{\mu_a - \mu_b,v}^2 \geq k^{C} (\iprod{v, \Sigma_a v} + \iprod{v, \Sigma_b v}), \iprod{v, \Sigma_{\text{mix}} v}/k^2$, where $\Sigma_{\text{mix}}$ is the covariance of the mixture $\tfrac 1k \sum_{i \leq k} G_i$.
So, choosing $\tau,\xi < \poly(\gamma,1/k)$, $\e \leq \poly(\gamma,1/k)^{\poly(1/\gamma,k)}$, and $\xi \leq \poly(\gamma,1/k)$, we have by \cref{lem:mean_sep} that
\begin{align}\label{eq:main-1}
\cA \proves{O(s)} \Paren{\frac{\alpha_a(w)\alpha_b(w)}{n^2}}^s \leq \gamma^{-O(s)} k^{O(s)} \cdot k^{-Cs} \mcom
\end{align}
where $s \leq \poly(\gamma,1/k)$.
For large-enough $C(k,\gamma)$, this is at most $(\gamma')^s$, for any $\gamma'(k,\gamma)$ we choose.

Next, suppose that case (2) of \cref{cor:TV_param} applies, for $C'(k,\gamma)$ to be chosen later, and let $A,B$ be the given partition of $[k]$.
The second part of case (2) ensures that for all $a \in A, b \in B$, we have
$\iprod{v, \Sigma_a v} / \iprod{v, \Sigma_b v} \leq (\iprod{v, \Sigma_a v} / \iprod{v, \Sigma_{\text{mix}} v})^{C' - 1}$
Taking $\rho = \min_{a \in A} \iprod{v, \Sigma_a v} / \iprod{v, \Sigma v}$, we can choose $C'(c')$ sufficiently large that if we choose $\tau = (\iprod{v, \Sigma_b v} / \iprod{v, \Sigma_a v})^{c'}$ for any $c' > 0$, then for small-enough $\e, \xi$, using $\iprod{v, \Sigma_a v} \geq k^{C'} \iprod{v, \Sigma b}$ and $C'$ sufficiently large,
\[
\cA \proves{O(s)} \Paren{\frac{\alpha_a(w)\alpha_b(w)}{n^2}}^s \leq \gamma^{-O(s)} k^{O(s)} \cdot k^{-C's}
\]
for any $\gamma'$ of our choosing.
Finally, case (3) follows the same argument, using \cref{lem:covar_sep}.

We have concluded that for any $\gamma'(k,\gamma)$, we can choose $C,C'$ sufficiently small to obtain \cref{eq:main-1} for any $\pE$ satisfying $\cA(X_1,\ldots,X_n,t,\delta,\e,k)$.
Now let us see what happens when the set of such $\pE$ is given to \textsc{split}.

Let $R_1,\ldots,R_m$ be the subsets output by \cref{lem:rounding}, with $\eta$ in \cref{lem:rounding} set to $\gamma'$ and hence $m \leq O(k \log 1/\gamma')$.
By Markov's inequality and a union bound and the guarantee of \cref{lem:rounding}, with probability at least $1 - 1/k^{99}$ over the randomness in \cref{lem:rounding}, every $R_t$ has $\min(|R_t \cap A'|, |R_t \cap B'|) \leq m k^{100} \gamma'n$, and $|\bigcup_{t \leq m} R_t| \geq (1- k^{100} \gamma') n$.

Let $A' = \cup_{a \in A} S_a$ and similarly for $B'$.
Let $R = \{ i \in R_t \, : \, |R_t \cap A'| \leq |R_t \cap B'| \}$.
We analyze $|R \cap B'|$ and $|\overline{R} \cap A'|$.
On one hand, $\bigcup_{t \leq m} R_t$ contains all but $k^{100} \gamma' n$ of $A'$.
On the other hand, for every $t$ such that $|R_t \cap B'| \leq |R_t \cap A'|$, we have $|R_t \cap B'| \leq m k^{100} \gamma' n$.
So $|R \cap A'| \geq |A'| - k^{100}\gamma'n - m^2 k^{100} \gamma' n$.
Now, $|A'| \geq n/k$, so $|R \cap A'| \geq (1 - \poly(m,k) \gamma') |A|$.
By similar analysis, $|R| \leq (1 + \poly(m,k)\gamma') |A|$.
So, adding or removing samples as necessary to arrive at a set size which is a multiple of $n/k$ finishes yields the conclusion of the proposition, for $\gamma'$ sufficiently small.
\end{proof}

\begin{proof}[Proof of \cref{lem:const_acc}]
By Lemma~\ref{lem:concentration}, by taking $d^{F(\rho,k)}$ samples we can assume that the $(\delta,\xi,t)$ deterministic conditions hold for any $\delta(k,\rho),\xi(k,\rho), t(k,\rho)$ of our choosing.
Note that this also implies that the deterministic conditions hold for the subset of samples from any submixture.

Now \cref{lem:const_acc} follows by induction on $k$ using \cref{prop:split}, which ensures that at each level of recursion in \textsc{cluster}, there is always a set $R_S$ which gives an $\e$-corrupted sample from a proper submixture.
\end{proof}

\section{Bounds on the Variance of Pseudoclusters}
\label[section]{sec:sep-lemmas}

Our goal in this section is to prove \cref{lem:mean_sep,lem:var_sep,lem:covar_sep}.
Each one of these follows from appropriate lower and upper bounds on the variance of linear and quadratic polynomials under the empirical distribution of samples in the ``pseudocluster'' captured by $\cA$.
We state those variance bounds and prove them in \cref{sec:variance-bounds-proofs}, once we have accumulated some additional tools in \cref{sec:tools}.

\paragraph{Notation}
  In the following \cref{lem:fake_var_upperbound,lem:fake_var_lowerbound,lem:fake_var_upperbound_quadratic,lem:fake_var_lowerbound_quadratic,lem:mixture_var}, let $Y_1,\ldots,Y_n \in \R^d$ satisfy the $(\delta,\xi,t)$-deterministic conditions (\cref{def:deterministic}) with respect to Gaussians $G_1,\ldots,G_k$, where $G_i = \cN(\mu_i,\Sigma_i)$.
  Let $S_1,\ldots,S_k$ be the partition of $Y_1,\ldots,Y_n$ guaranteed by \cref{def:deterministic}.
  Let $X_1,\ldots,X_n$ be an $\e$-corruption of $Y_1,\ldots,Y_n$ and let $\cA$ be the clustering axioms (\cref{def:axioms}) for $X_1,\ldots,X_n$.

The first two lemmas capture the variance of linear functions under the uniform distribution over the pseudocluster.

\begin{lemma}[Variance upper bound, linear functions]
\label[lemma]{lem:fake_var_upperbound}
For all $\tau > 0$ there exists an even $s = \tilde{O} \Paren{ 1/\tau^2}$, such that for all $a \in [k]$,  $\eps,\delta \leq s^{-O(s)} k^{-20}$ and $v \in \R^d$,
\begin{align*}
\cA &\proves{O(s)}  \iprod{v, \Sigma_a v}^{s}  \geq s^{-O(s)} \cdot \Brac{ \Paren{\frac{\alpha_a(w)}{n}}^8 - O(\tau^4) - \epsilon^2 k^{10} s^{O(s)} } \cdot \iprod{v, \Sigma(w)v}^{s} \mper
\end{align*}
\end{lemma}

\begin{lemma}[Variance lower bound, linear functions]
\label[lemma]{lem:fake_var_lowerbound}
For all $v \in \R^d$ and $\eps, \delta \ll 1/k^{10}$ and $t$ a sufficiently large constant, we have that
\begin{align}
	\cA &\proves{O(1)} \iprod{v, \Sigma(w)v}^4 \geq \Omega\Paren{\iprod{v, \Sigma_a v} }^4 \cdot \xi^4 \cdot \Brac { \Paren{\frac{ \alpha_a(w)}{n}}^8 - O(\xi^4) - O\Paren{\frac{\epsilon k^{10}}{\xi^4}}} \label{eq:fake_var_lb_1}
        \end{align}
        for every $a \in [k]$, and
        \begin{align}
        \cA \proves{O(1)} \iprod{v, \Sigma(w)v}^4 & \geq \Omega\Paren{ \iprod{\mu_a - \mu_b, v}^8} \cdot \Brac{ \Paren{\frac{\alpha_a(w)\alpha_b(w)}{n^2}}^4 - O(\xi^4) - O(\epsilon^2 k^{10}) } \nonumber \\
        & \qquad -  O(\epsilon^2 k^{10} + (\log 1/\xi)^2) \cdot \Paren {\iprod{v, \Sigma_a v}^4 + \iprod{v,\Sigma_b v}^4} \label{eq:fake_var_lb_2}
        \end{align}
        for every $a,b \in [k]$.
\end{lemma}

The next two lemmas give upper and lower bounds on the variance of \emph{quadratic} polynomials under the uniform distribution on the pseudocluster.
\cref{lem:covar_sep} follows by combining the two bounds and simplifying.

\begin{lemma}[Variance upper bound, quadratic functions]
\label[lemma]{lem:fake_var_upperbound_quadratic}
\Snote{TODO: revise lemma statement and proof}
	Let $p(x) := \iprod{x,Ax}$ where $A \in \R^{ d \times d}$ and $S \subseteq [k]$.
        Suppose $\Sigma_S = \E_{a,b \sim S} \E_{X \sim G_a, Y \sim G_b} (X-Y)(X-Y)^\top$, and $\alpha_S(w) := \sum_{i \in S} \Paren{ \sum_{j \in S_i} w_j}$ Then for $\eps, \delta < 0.1 k^{-10}$ and  $t$ a sufficiently large constant, we have
	\begin{align*}
\cA \proves{O(1)} &\Paren{ \Var_{X,Y \sim \cN(\mu(w), \Sigma(w))}[p(X-Y)]}^2 \leq O(\| \Sigma_S^{1/2} A \Sigma_S^{1/2} \|_F^{4}) + O\Paren{\frac{\alpha_{\overline{S}}} n } \cdot (\|\Sigma_S^{1/2} A \Sigma^{1/2} \|_F^4 + \|\Sigma^{1/2} A \Sigma^{1/2}\|_F^4 )\mper
\end{align*}

\end{lemma}

\begin{lemma}[Variance lower bound, quadratic functions]
\label[lemma]{lem:fake_var_lowerbound_quadratic}
\Snote{TODO: revise lemma statement and proof}
Let $p(x) := \iprod{x,Ax}$ where $A \in \R^{d \times d}$. Then for $\eps, \delta, \xi < 0.1 k^{-10}$, for every $a, b \in [k]$ and  $t$ a sufficiently large constant, we have
	
\begin{align*}
\cA \proves{O(1)} \Paren{ \Var_{x, y \sim \cN(\mu(w), \Sigma(w))}[p(x-y)] }^4 &\geq \Omega\Paren{ \frac{ \alpha_a(w)^{16} \alpha_b(w)^{16}}{n^{32}} - O(\xi^4) - O(\eps^2 \cdot k^{20})  }\iprod{\Sigma_a - \Sigma_b, A}^8 \\
&\qquad - O(\log(1/\xi)^8  + \eps^2 k^{20})  \cdot  \Paren{ \|\Sigma_a^{1/2} A \Sigma_a^{1/2}\|_F^8 +  \|\Sigma_b^{1/2} A \Sigma_b^{1/2} \|_F^8} \mper
\end{align*}
\end{lemma}

Before we can prove \cref{lem:mean_sep,lem:var_sep,lem:covar_sep}, we need one more (simpler) lemma relating $\iprod{v, \Sigma(w) v}$ to the variance of the mixture $G_1,\ldots,G_k$.

\begin{lemma}[See proof in \cref{sec:proofs-tools-section}]\label[lemma]{lem:mixture_var}
If $\epsilon \leq c k^{-10}$ for a sufficiently-small constant $c$, then for all $v \in \R^d$,
  \[
  \cA \proves{O(1)} \iprod{v, \Sigma(w) v}^4 \leq O \Paren{\E_{X,X' \sim \tfrac 1 k \sum_{i \leq k} G_i} \iprod{X - X', v}^2}^4\mper
  \]
\end{lemma}

Now we turn to the proofs of \cref{lem:mean_sep,lem:var_sep,lem:covar_sep}.

\begin{proof}[Proof of \cref{lem:mean_sep}]
  We start by rearranging \cref{lem:fake_var_upperbound} to have
  \[
  \cA \proves{O(s)} \Paren{\frac{\alpha_a(w)} n}^8 \iprod{v, \Sigma(w) v}^s \leq s^{O(s)} \iprod{v, \Sigma_a v}^s + (O(\tau^4) + \e^2 k^{10} s^{O(s)}) \iprod{v, \Sigma(w) v}^s\mper
  \]
  By \cref{fact:sos-1} (observing that the SoS proof in \cref{fact:sos-1} is constant-degree in the variable $x$), this implies
  \[
  \cA \proves{O(s)} \Paren{\frac{\alpha_a(w)} n}^s \iprod{v, \Sigma(w) v}^s \leq s^{O(s)} \iprod{v, \Sigma_a v}^s + (O(\tau)^{s/2} + \e^{\Omega(s)} k^{O(s)} s^{O(s^2)}) \iprod{v, \Sigma(w) v}^s\mper
  \]

  At the same time, by \cref{lem:fake_var_upperbound},
  \begin{align*}
  \cA \proves{O(1)} & \Paren{\frac{\alpha_a(w)\alpha_b(w)}{n^2}}^4 \iprod{\mu_a - \mu_b, v}^8
  \leq O(\iprod{v, \Sigma(w) v})^4 + O(\xi^4 + \e^2k^{10})\iprod{\mu_a - \mu_b,v}^8\\
   & + O(\e^2 k^{10} + (\log 1/\xi)^2) (\iprod{v, \Sigma_a v}^4 + \iprod{v, \Sigma_b v}^4)\mper
   \end{align*}
   Both sides of this inequality are sums of squares, so we can raise both sides to $s/4$ and obtain
   \begin{align*}
     \cA \proves{O(1)} & \Paren{\frac{\alpha_a(w)\alpha_b(w)}{n^2}}^s \iprod{\mu_a - \mu_b, v}^{2s} \leq O(\iprod{v, \Sigma(w) v})^s + O(\xi^4 + \e^{2} k^{10})^{s/4} \iprod{\mu_a - \mu_b,v}^{2s} \\
     & + O(\e^2 k^{10} + (\log 1/\xi)^2)^{s/4} (\iprod{v, \Sigma_a v}^s + \iprod{v, \Sigma_b v}^s)\mper
   \end{align*}
   We can multiply both sides by $(\alpha_a(w)/n)^s$ and then use $\alpha_a(w) / n \leq 1$ to obtain
  \begin{align*}
     \cA \proves{O(1)} & \Paren{\frac{\alpha_a(w)\alpha_b(w)}{n^2}}^{2s} \iprod{\mu_a - \mu_b, v}^{2s} \leq \Paren{\frac{\alpha_a(w)}{n}}^s O(\iprod{v, \Sigma(w) v})^s + O(\xi^4 + \e^{2} k^{10})^{s/4} \iprod{\mu_a - \mu_b,v}^{2s} \\
     & + O(\e^2 k^{10} + (\log 1/\xi)^2)^{s/4} (\iprod{v, \Sigma_a v}^s + \iprod{v, \Sigma_b v}^s) \mper
   \end{align*}
   We have arrived now at upper and lower bounds on $(\alpha_a(w)/n)^s \iprod{v, \Sigma(w) v}^s$.
   Putting them together,
  \begin{align*}
     \cA \proves{O(1)} & \Paren{\frac{\alpha_a(w)\alpha_b(w)}{n^2}}^{2s} \iprod{\mu_a - \mu_b, v}^{2s}
     \leq s^{O(s)} \iprod{v, \Sigma_a v}^s + (O(\tau)^{s/2} + \e^{\Omega(s)} k^{O(s)} s^{O(s)}) \iprod{v, \Sigma(w) v}^s\\
     & + O(\xi^4 + \e^{2} k^{10})^{s/4} \iprod{\mu_a - \mu_b,v}^{2s} \\
     & + O(\e^2 k^{10} + (\log 1/\xi)^2)^{s/4} (\iprod{v, \Sigma_a v}^s + \iprod{v, \Sigma_b v}^s)\mper
   \end{align*}
   By \cref{lem:mixture_var} and our hypothesis $\iprod{\mu_a - \mu_b, v}^2 \geq \rho \E_{X,X' \sim \tfrac 1 k \sum_{i=1}^k G_i} \iprod{X - X', v}^2$, we can bound the $\iprod{v, \Sigma(w) v}^s$ term to obtain
  \begin{align*}
     \cA \proves{O(1)} & \Paren{\frac{\alpha_a(w)\alpha_b(w)}{n^2}}^{2s} \iprod{\mu_a - \mu_b, v}^{2s}
     \leq s^{O(s)} \iprod{v, \Sigma_a v}^s + (O(\tau)^{s/2} + \e^{\Omega(s)} k^{O(s)} s^{O(s^2)}) \cdot \rho^{-s} \iprod{\mu_a - \mu_b, v}^{2s} \\
     & + O(\xi^4 + \e^{2} k^{10})^{s/4} \iprod{\mu_a - \mu_b,v}^{2s} \\
     & + O(\e^2 k^{10} + (\log 1/\xi)^2)^{s/4} (\iprod{v, \Sigma_a v}^s + \iprod{v, \Sigma_b v}^s)\mper
   \end{align*}
   Dividing by $\iprod{\mu_a - \mu_b, v}^{2s}$ finishes the proof.
\end{proof}

\begin{proof}[Proof of \cref{lem:var_sep}]
  The proof follows the same approach as the proof of \cref{lem:mean_sep}, using the other half of \cref{lem:fake_var_lowerbound}.
\end{proof}

\begin{proof}[Proof of \cref{lem:covar_sep}]
\cref{lem:covar_sep} follows immediately from \cref{lem:fake_var_upperbound_quadratic,lem:fake_var_lowerbound_quadratic}.
\end{proof}


\section{Toolkit: Decorruption, Hypercontractivity, and Anti-concentration}
\label[section]{sec:tools}

We assemble here a collection of SoS proofs using the clustering axioms $\cA$ which will be 
the building blocks of our proofs of the lemmas from \cref{sec:sep-lemmas}. 

\new{
The proof of all the lemmas in this section have been deferred to Appendix~\ref{sec:proofs-tools-section}.}

\paragraph{Notation}
In this section, we adopt the following notation.
$Y_1,\ldots,Y_n \in \R^d$ satisfy the $(\delta,\xi,t)$ deterministic conditions (\cref{def:deterministic}) with respect to Gaussians $G_1 = \cN(\mu_1,\Sigma_1),\ldots,G_k = \cN(\mu_k,\Sigma_k)$, yielding a partition of $[n]$ into $\{S_a\}_{a \in [k]}$.
$X_1,\ldots,X_n$ are an $\e$-corruption of $Y_1,\ldots,Y_n$, and $\cA=\cA(X_1,\ldots,X_n,t,\delta,\epsilon, k)$ are the clustering axioms (\cref{def:axioms}).
Throughout, we assume that $\epsilon,\delta \leq 0.1$.

\subsection{From Empirical to Population Moments}

The following \cref{lem:emp-pop,lem:emp-pop-2} relate the empirical distribution of samples $Y_1,\ldots,Y_n$ and what we call \emph{ghost samples} -- indeterminates $w_i X_i'$ in our polynomial system $\cA$ -- to appropriate population averages.

\begin{lemma}[Empirical to Population Moments, Ghost Samples]
  \label[lemma]{lem:emp-pop}
  Let $p$ be a degree-$(t/2)$ symbolic polynomial.
  Then
  \begin{equation}\label{eq:emp-to-pop-1}
  \cA \proves{O(t)} \Paren{\frac k n \sum_{i=1}^n w_i p(X_i')
  - \E_{g \sim \cN(\mu(w),\Sigma(w))} p(g)}^2 \leq  (2t)^{O(t)} \cdot \delta  \cdot \Paren{ \E_{g \sim \cN(\mu(w),\Sigma(w))} p(g)^2 }
  \end{equation}
  and
	\begin{equation}\label{eq:emp-to-pop-diff-1}
	\cA \proves{O(t)} \Paren{\Paren{\frac k n}^2 \sum_{i, j=1}^n w_iw_j p(X_i' - X_j')
		- \E_{g \sim \cN(0,2\Sigma(w))} p(g)}^2 \leq  (2t)^{O(t)} \cdot \delta \cdot \E_{g \sim \cN(0,2\Sigma(w))} p(g)^2 \mper
	\end{equation}
Additionally, if $p$ is itself a square, then
	\begin{equation}\label{eq:emp-to-pop-diff-2}
	\cA \proves{O(t)} \Paren{\Paren{\frac k n}^2 \sum_{i, j=1}^n w_iw_j p(X_i' - X_j')
		- \E_{g \sim \cN(0,2\Sigma(w))} p(g)}^2 \leq  (2t)^{O(t)} \cdot \delta \cdot \Paren{\E_{g \sim \cN(0,2\Sigma(w))} p(g) }^2 \mper
	\end{equation}
\end{lemma}

\begin{lemma}[Empirical to Population Moments, Real Samples]
  \label[lemma]{lem:emp-pop-2}
  Let $p$ be a degree-$(t/2)$ symbolic polynomial.
  For every $j \in [k]$,
  \begin{equation}
  \proves{O(t)} \Paren{\frac k n \sum_{i \in S_j} p(Y_i)
  - \E_{g \sim G_j} p(g)}^2 \leq \delta \cdot (2t)^{O(t)} \cdot \E_{g \sim G_j} p(g)^2\mper
  \end{equation}
  Furthermore,
  \begin{equation}
    \proves{O(t)} \Paren{\frac k n \sum_{i,l \in S_j} p(Y_i - Y_l)
  - \E_{g \sim \cN(0,2\Sigma_j)} p(g)}^2 \leq \delta \cdot (2t)^{O(t)} \cdot \E_{g \sim \cN(0,2\Sigma_j)} p(g)^2\mper
  \end{equation}
  And if $p$ is itself a square, then
  	\begin{equation}
  \cA \proves{O(t)} \Paren{\Paren{\frac k n}^2 \sum_{i, j=1}^n p(Y_i - Y_j)
  	- \E_{g \sim \cN(0,2\Sigma_j)} p(g)}^2 \leq  \delta \cdot (2t)^{O(t)} \cdot \Paren{\E_{g \sim \cN(0,2\Sigma_j)} p(g) }^2 \mper
  \end{equation}
  
\end{lemma}

\subsection{Decorruption}

\cref{lem:mom_decorruption_diff} captures the use of higher moments to remove the effect of adversarial corruptions.

\begin{lemma}[Moment Decorruption]\label[lemma]{lem:mom_decorruption_diff}
For all $a,b \in [k]$,
\begin{enumerate}
\item For all $s \leq t/10$ and all $v \in \R^d$, we have that
\begin{align}
\cA &\proves{t} \Paren{\frac{k^2}{n^2} \sum_{i \in S_a,j \in S_b} w_i w_j \Paren{ \iprod{X_i' - X_j', v}^{s} - \iprod{Y_i - Y_j, v}^{s}}} ^4 \\
&\leq \eps^2 \cdot k^{10} \cdot (2s)^{O(s)} \cdot \Paren{\iprod{v, \Sigma_a v}^{2s} + \iprod{v, \Sigma_b v}^{2s} + \iprod{ \mu_a - \mu_b, v}^{4s} + \iprod{v,\Sigma(w) v}^{2s} }\mper
\end{align}

\item For all $A \in \R^{d \times d}$, we have that 
\begin{align}
\cA &\proves{t} \Paren{\frac{k^4}{n^4} \sum_{(i,j) \in S_a,(r,l) \in S_b} w_i w_j w_r w_l \Paren{ \iprod{(Y_i - Y_j)^{\otimes 2}- (Y_r - Y_l)^{\otimes 2}, A}^2 - \iprod{(X'_i - X'_j)^{\otimes 2} - (X'_r - X'_l)^{\otimes 2}, A}^2 } }^4 \\
&\leq  O(\eps^2 \cdot k^{20}) \cdot  \Paren{  \|\Sigma_a^{1/2} A \Sigma_a^{1/2}\|_F^8  +  \|\Sigma_b^{1/2} A \Sigma_b^{1/2} \|_F^8 + \|\Sigma^{1/2} A \Sigma^{1/2}\|_F^8 + \iprod{\Sigma_a - \Sigma_b, A}^8} \mper
\end{align}
\end{enumerate}
For any $S \subseteq [k]$ and $T(S) = \bigcup_{a \in S} S_a$, let $\Sigma_S = \tfrac 1 {n^2} \sum_{i,j \in T(S)} (Y_i - Y_j)(Y_i - Y_j)^\top$ and $\tfrac 1 {n^2} \Sigma'_S = \sum_{i,j \in T(S)} w_i w_j (X_i' - X_j')(X_i' - X_j')^\top$.
Then,
\begin{enumerate}
\setcounter{enumi}{2}
\item
$\Paren{ \Tr(B^T \Sigma'_S B) - \Tr(B^T \Sigma_S B) }^2  \leq O(\eps^2 k^2) \cdot ( \Tr(B^T \Sigma_S B)^2 + \Tr(B^T \Sigma'_S B)^2 )$.
\end{enumerate}
\end{lemma}

\subsection{Hypercontractivity}
We also need the following corollaries of the SoS proof of hypercontractivity from \cite{KOTZ12}.

\begin{theorem} [\cite{KOTZ12}]
  For every $t \in\N$ and even integer $q$, there is a degree-$O(tq)$ SoS proof 
  \[
  \proves{O(kq)}^f \E_{g \sim \cN(0,I)} f(g)^q \leq O(q-1)^{qt/2} \cdot \Paren{\E_{g \sim \cN(0,I)} f(g)^2}^{q/2} \mcom
  \]
  where $f$ is a degree-$t$ symbolic polynomial.
\end{theorem}

The next corollary follows by the identities $\E_{g \sim \cN(0,I)} (\Sigma^{1/2} g + \mu)^{\otimes s} = \E_{g \sim \cN(\mu,\Sigma)} g^{\otimes s}$ as polynomials in $\mu,\Sigma$.

\begin{corollary}
\label[corollary]{fact:hypercontractivity}
  For every $t \in \N$ and even integer $q$, there is a degree-$O(tq)$ SoS proof
  \[
  \{(\Sigma^{1/2})^2 = \Sigma\} \proves{O(tq)}^{f,\mu,\Sigma,\Sigma^{1/2}} \E_{g \sim \cN(\mu,\Sigma)} f(g)^q \leq O(q-1)^{qt/2} \cdot \Paren{\E_{g \sim \cN(\mu,\Sigma)} f(g)^2}^{q/2} \mcom
  \]
  where $f$ is a degree-$t$ symbolic polynomial in $d$ variables, $\mu$ is a vector of $d$ indeterminates, and $\Sigma, \Sigma^{1/2}$ are matrices of indeterminates.
\end{corollary}

We will require an additional version of SoS hypercontractivity.

\begin{lemma}
\label[lemma]{lemma:hypercontractivity-sos-polys}
   For every $t \in \N$, there is a degree-$O(t)$ SoS proof
  \[
  \{(\Sigma^{1/2})^2 = \Sigma\} \proves{O(t)}^{f_1,\ldots,f_m} \E_{g \sim \cN(\mu,\Sigma)} \Paren{\sum_{i\leq m} f_i(g)^2}^2 \leq \exp(O(t)) \cdot \Paren{\E_{g \sim \cN(\mu,\Sigma)} \sum_{i \leq m} f_i(g)^2 }^2\mcom
  \]
  where $f_1,\ldots,f_m$ are degree-$t$ symbolic polynomials in $d$ variables, $\mu$ is a vector of $d$ indeterminates, and $\Sigma, \Sigma^{1/2}$ are matrices of indeterminates.
\end{lemma}
\begin{proof}
  It suffices to prove the lemma with $0,I$ substituted for $\mu,\Sigma$ respectively.
  The proof of Lemma 5.1 of \cite{barak2012hypercontractivity} shows that for two degree-$t$ symbolic polynomials $f_i,f_j$, we have $\proves{O(t)} \E_{\cN(0,I)} f_i^2 f_j^2 \leq \exp(O(t)) \cdot (\E_{\cN(0,I)} f_i^2) \cdot (\E_{\cN(0,I)} f_j^2)$.
  The lemma follows by expanding $\Paren{\sum_{i\leq m} f_i(g)^2}^2 = \sum_{i,j \leq m} f_i(g)^2 f_j(g)^2$.
\end{proof}

%

\subsection{Certifiable Anti-concentration}
The following lemma shows that anti-concentration properties of probability distributions on the real line can be captured by low-degree polynomials -- roughly, anti-concentration for an interval of width $\e$ is captured by a polynomial of degree $\tilde{O}(1/\epsilon^2)$.
Similar results are used in \cite{raghavendra2020list,karmalkar2019list}.
We provide a proof in \cref{sec:anticoncentration} for completeness.

\begin{lemma}
\label[lemma]{lem:q-main}
  There is a universal constant $C > 0$ such that for every $\e \in (0,1]$ there is a univariate polynomial $q_\e(x)$ of degree at most $(\log 1/\e)^C \cdot 1/\e^2$ such that
  \begin{enumerate}
  \item $q_\e$ is even.
  \item $q_\e(x) \geq 1-\e$ for $x \in [-\e,\e]$.
  \item $\E_{x \sim \cN(0,1)} q_\e(x) \leq C\e$.
  \item $q_\e$ is a square.
  \end{enumerate}
\end{lemma}

\section{Proofs of Variance Bounds on Pseudoclusters}
\label[section]{sec:variance-bounds-proofs}

With the tools from \cref{sec:tools} in hand, we are ready to prove the lemmas from \cref{sec:sep-lemmas}.


\subsection{Variance of Linear Functions}

\begin{proof}[Proof of \cref{lem:fake_var_upperbound}]
\cref{lem:q-main} implies there exist a constant $C > 0$ and an even square polynomial $q_{\tau}$ of degree $s = (\log 1/\tau)^C \cdot 1/\tau^2$ satisfying
\begin{align}\label{eqn:fake_anticoncentration}
\proves{2s}^{\iprod{v, \Sigma(w) v} } \iprod{v, \Sigma(w) v}^{2s} \E_{x \sim  \cN(0, 2\iprod{v, \Sigma(w) v})} \left [q_{\tau}\Paren{\frac{x}{2^{1/2}\iprod{v, \Sigma(w) v}^{1/2}}} \right]^2 \leq O(\tau^2) \cdot \iprod{v, \Sigma(w) v}^{2s}  \;.
\end{align}
Note that because $q$ is even with degree at most $s$, both sides of \eqref{eqn:fake_anticoncentration} are in fact polynomials in $\iprod{v,\Sigma(w)v}$, despite the appearance of $\iprod{v,\Sigma(w)v}^{1/2}$ in the denominator.
To prove it, we have used that the polynomial inequality $\sigma^{2s} (\E_{x \sim \cN(0,\sigma^2)} q_\tau(x/\sigma))^2 \leq O(\tau^2) \sigma^{2s}$ in the single variable $\sigma$ holds for every $\sigma \in \R$, so the difference of the RHS and LHS is a sum of squares, by \cref{fact:univariate}.

Since $q_{\tau}$ is a square polynomial, we can apply \cref{lem:emp-pop} to $\iprod{v, \Sigma(w) v}^{s} q_\tau \Paren{\frac{\iprod{X'_i - X'_j, v}}{\sqrt{2\iprod{v, \Sigma(w)v}}}}$.
We thus get
\begin{align*}
\cA &\proves{O(s)} \Paren{\frac{k^2}{n^2} \sum_{i,j=1}^n w_i w_j \iprod{v, \Sigma(w) v}^{s} q_\tau \Paren{\frac{\iprod{X'_i - X'_j, v}}{\sqrt{2\iprod{v, \Sigma(w)v}}}}}^2\\
&\leq 2 \Paren{\frac{k^2}{n^2} \sum_{i,j=1}^n w_i w_j \iprod{v, \Sigma(w) v}^{s} q_\tau \Paren{\frac{\iprod{X'_i - X'_j, v}}{\sqrt{2\iprod{v, \Sigma(w)v}}}} - \E_{g \sim \cN(0,2\iprod{v, \Sigma(w)v})} \iprod{v, \Sigma(w) v}^{s} q_{\tau}\Paren{\frac{g}{\sqrt{2\iprod{v, \Sigma(w)v}}}}}^2 \\
&+ 2\Paren{ \E_{g \sim \cN(0,2\iprod{v, \Sigma(w)v})} \iprod{v, \Sigma(w) v}^{s} q_{\tau}\Paren{\frac{g}{\sqrt{2\iprod{v, \Sigma(w)v}}}}}^2  \\
& \leq (2 + \delta s^{O(s)}) \cdot \Paren{ \E_{g \sim \cN(0,2\iprod{v, \Sigma(w)v})} \iprod{v, \Sigma(w) v}^{s} q_{\tau}\Paren{\frac{g}{\sqrt{2\iprod{v, \Sigma(w)v}}}}}^2 \;.
\end{align*}
Combining this with \eqref{eqn:fake_anticoncentration}, we obtain
\begin{align} \label{eqn:q_upper}
\cA \proves{O(s)} \iprod{v, \Sigma(w) v}^{2s} \cdot \Paren {\frac{k^2}{n^2} \sum_{i,j \in [n]} w_i w_j q_\tau \Paren{\frac{\iprod{X'_i - X'_j, v}}{\sqrt{2 \iprod{v, \Sigma(w)v}}}}}^2  \leq (2 + \delta s^{O(s)}) \cdot O(\tau^2) \cdot \iprod{v, \Sigma(w) v}^{2s} \;.
\end{align}
Now, observe that 
\[ 
\proves{2s}^x q_{\tau}(x) + \Paren{\frac{x}{\tau}}^{2s} \geq 1	
\]
by \cref{fact:univariate},
since $q_\tau(x) + (x/\tau)^{2s} \geq 1$ for all $x \in \R$ (by \cref{lem:q-main}).
Substituting $x = \Paren{\frac{\iprod{X'_i - X'_j, v}}{\sqrt{2\iprod{v, \Sigma(w)v}}}}$,
\[ 
\proves{2s}^{\Paren{\frac{\iprod{X'_i - X'_j, v}}{\sqrt{2\iprod{v, \Sigma(w)v}}}}}  q_{\tau}\Paren{\frac{\iprod{X'_i - X'_j, v}}{\sqrt{2\iprod{v, \Sigma(w)v}}}} + \Paren{\frac{\iprod{X'_i - X'_j, v}}{\tau \sqrt{2\iprod{v, \Sigma(w)v}}}}^{2s} \geq 1\;.
\]
Clearing the denominators in the SoS proof of the above inequality by multiplying throughout with $\tau^{2s} (2\iprod{v, \Sigma(w)v})^s$, we get
\begin{equation}
\proves{O(s)}^{\iprod{X_i'-X_j', v}, \iprod{v, \Sigma(w)v}} \tau^{2s}\iprod{v, \Sigma(w)v}^s q_\tau \Paren{\frac{\iprod{X'_i - X'_j, v}}{\sqrt{\iprod{v, \Sigma(w)v}}}} +  \iprod{X_i'-X_j',v}^{2s} \geq \tau^{2s} \iprod{v, \Sigma(w)v}^s \mper
\end{equation} 
Averaging the above over $i, j \in S_a$, we see that
\begin{align*}
&\proves{O(s)} \frac{k^2 \tau^{2s}}{n^2} \cdot \sum_{i, j \in S_a} w_i w_j \Paren{ \iprod{v, \Sigma(w)v}^s q_\tau \Paren{\frac{\iprod{X'_i - X'_j, v}}{\sqrt{\iprod{v, \Sigma(w)v}}}} +  \iprod{X_i'-X_j',v}^{2s}} 
\geq (k/n)^2 \alpha_a(w)^2 \cdot \tau^{2s} \iprod{v, \Sigma(w)v}^s \mper
\end{align*}
Squaring both sides and using $\proves{2} (a+b)^2 \leq 2 a^2 + 2 b^2$ to further upper bound the LHS and substituting Equation~\ref{eqn:q_upper}, we get
\begin{align*}
\cA &\proves{O(s)} (2 + \delta s^{O(s)}) \tau^{4s+2} \cdot \iprod{v, \Sigma(w) v}^{2s}+ \frac{2 k^4 \tau^{4s}}{n^4} \cdot\Paren{\sum_{i, j \in S_a} w_i w_j \iprod{X_i'-X_j',v}^{2s}}^2 \\
&\geq (k/n)^4 \alpha_a(w)^4 \tau^{4s} \iprod{v, \Sigma(w)v}^{2s} \;.
\end{align*}
We can divide by $\tau^{4s}$, use the assumption $\delta \leq s^{-O(s)}$, square and rearrange to get
\begin{align*}
\cA &\proves{O(s)} \frac{k^8}{n^8} \cdot\Paren{\sum_{i, j \in S_a} w_i w_j \iprod{X_i'-X_j',v}^{2s}}^4 \geq \Paren{\Omega ( \alpha_a(w)/n)^8 -  O(\tau^4)} \cdot \iprod{v, \Sigma(w)v}^{4s} \mper
\end{align*}
To decorrupt this, we apply \cref{lem:mom_decorruption_diff} with both $a, b$ in the lemma set to $a$, yielding
\begin{align*}
\cA &\proves{O(s)}  \Paren{\frac{k^2}{n^2}\sum_{i, j \in S_a} w_i w_j \iprod{Y_i-Y_j,v}^{2s}}^4 + \eps^2 (2s)^{O(s)} k^{10} \cdot \Paren{  \iprod{v, \Sigma_a v}^{4s} + \iprod{v, \Sigma(w) v}^{4s}} \\
& \geq \Paren{\Omega ( \alpha_a(w)/n)^8 -  O(\tau^4)} \cdot \iprod{v, \Sigma(w)v}^{4s} \mper
\end{align*}
Using $\cA \proves{O(1)} w_iw_j \leq 1$ and $\epsilon \leq (2s)^{-O(s)} k^{-20}$ and \cref{lem:emp-pop-2},
\begin{align*}
\cA &\proves{O(s)} (2s)^{O(s)} \cdot \iprod{v, \Sigma_a v}^{4s}  \geq \Brac{ \Omega(\alpha_a(w)/n)^8 - O(\tau^4) - \e^2 k^{10} (2s)^{O(s)} } \cdot \iprod{v, \Sigma(w)v}^{4s} \mper
\end{align*}
which is what we wanted to show.
\end{proof}

\begin{proof}[Proof of \cref{lem:fake_var_lowerbound}]
We will prove \cref{eq:fake_var_lb_1} and then describe how the proof can be easily modified to obtain \cref{eq:fake_var_lb_2}.
	Note that $2\iprod{v, \Sigma(w)v} = \Paren{\frac{k}{n}}^{2} \sum_{i, j \leq n} w_i w_j \iprod{X_i'-X_j', v}^2$. Since all the terms  in the sum are nonnegative,
	\begin{align*}
	\cA \proves{O(1)} \Brac{\Paren{\frac{k}{n}}^{2} \sum_{i, j \leq n} w_i w_j \iprod{X_i'-X_j', v}^2}^4 \geq \Brac{ \Paren{\frac{k}{n}}^{2} \sum_{(i, j) \in S_a^{2}} w_i w_j \iprod{X_i'-X_j', v}^2 }^4\mper
        \end{align*}
 By \cref{lem:mom_decorruption_diff}, 
         \begin{align*}
\cA \proves{O(1)} \Brac{\Paren{\frac{k}{n}}^{2} \sum_{(i, j) \in S_a^2 } w_i w_j \iprod{X_i'-X_j', v}^2}^4 & \geq \frac 1 2 \Brac{\Paren{\frac{k}{n}}^{2} \sum_{(i, j) \in S_a^2} w_i w_j \iprod{Y_i-Y_j, v}^2}^4\\
& \qquad - O(\epsilon^2 k^{10}) \cdot (\iprod{v,\Sigma_a v}^4 + \iprod{v,\Sigma(w)v}^4)\mper
         \end{align*}
         Since $Y_1,\ldots,Y_n$ satisfy the deterministic conditions (\cref{def:deterministic}), at least $(1-\xi)(n/k)^2$ elements of $S_a^2$ have $\iprod{Y_i - Y_j,v}^2 \geq \Omega(\xi) \iprod{v,\Sigma_a v}$.
         Restricting attention only to those pairs, we obtain
         \begin{align*}
           \cA \proves{O(1)} \sum_{i,j \in S_a^2} w_i w_j \iprod{Y_i - Y_j,v}^2 \geq (\alpha_a(w)^2 - \xi(n/k)^2) \cdot \Omega(\xi) \cdot \iprod{v,\Sigma_a v}\mper
         \end{align*}
         Putting this together gives
         \begin{align*}
         \cA \proves{O(1)}
\Brac{\Paren{\frac{k}{n}}^{2} \sum_{(i, j) \in S_a^2} w_i w_j \iprod{X_i-X_j, v}^2}^4 & \geq ([\alpha_a(w)/n]^8 - O(\xi)^4) \cdot \Omega(\xi)^4 \cdot \iprod{v,\Sigma_a v}^4 \\
& - O(\epsilon^2 k^{10}) (\iprod{v, \Sigma_a v}^4 + \iprod{v, \Sigma(w) v}^4)\mper
        \end{align*}
        Using $\epsilon \ll k^{-10}$, we can rearrange to obtain
         \begin{align*}
         \cA \proves{O(1)}
\Brac{\Paren{\frac{k}{n}}^{2} \sum_{(i, j) \in S_a^2} w_i w_j \iprod{X_i'-X_j', v}^2}^4 & \geq ([\alpha_a(w)/n]^8 - O(\xi)^4) \cdot \Omega(\xi)^4 \cdot \iprod{v,\Sigma_a v}^4 \\
& - O(\epsilon^2 k^{10}) (\iprod{v, \Sigma_a v}^4) \;,
        \end{align*}
        which rearranges to what we wanted to show.

        To obtain \cref{eq:fake_var_lb_2}, the above argument can be modified as follows.
        First, instead of restricting to $i,j \in S_a^2$, we restrict to $i,j \in S_a \times S_b$.
        Second, we use the deterministic conditions \cref{def:deterministic} to obtain $(1-\xi)(n/k)^2$ pairs $i,j \in S_a \times S_b$ which satisfy $|\iprod{Y_i - Y_j,v}| \geq |\iprod{\mu_a - \mu_b,v}| - \sqrt{\log(1/\xi)} (\sqrt{\iprod{v, \Sigma_a v}} + \sqrt{\iprod{v, \Sigma_b v}})$.
\end{proof}


\subsection{Variance of Quadratic Functions}
To prove \cref{lem:fake_var_upperbound_quadratic} we will need the following claims.

\begin{claim} \label[claim]{claim:w-to-good}
In the same notation as \cref{lem:fake_var_upperbound_quadratic}, for a $d \times d$ matrix of indeterminates $B$, we have
	\begin{align*}
	\cA \proves{O(1)}^B \Paren{ 1-O\Paren{\frac{\alpha_{\overline{S}}(w)}{n}}} \cdot (\Tr\Paren{ B^T \Sigma(w) B})^2  \leq  O(1) \cdot (\Tr(B^T \Sigma_S B))^2 \;.
	\end{align*}
\end{claim}
\begin{proof}
	Let $R := \{ i \mid i \in S_a, ~a \in S\}$. Let $\Sigma'_S := \sum_{a,b \in S} \Paren{\sum_{i,j \in S_a \times S_b} w_i w_j (X'_i-X'_j)^{\otimes 2}} = \sum_{i,j \in R}  (X'_i - X'_j)^{\otimes 2}$. An application of $(a+b)^2 \leq 2 \cdot a^2 + 2 \cdot b^2$ and the Cauchy Schwartz inequality 
	give us the following
	\begin{align*}
	\cA \proves{O(1)} \Tr\Paren{ B^T \Sigma(w) B}^2 &= \Paren{ \frac{1}{n^2} \cdot \sum_{i,j \leq n} w_i w_j \Tr( B^T (X'_i - X'_j) (X'_i - X'_j)^T B)}^2 \\
	&= \Paren{  \frac{1}{n^2} \cdot \sum_{i,j \leq n } w_i w_j \Tr( B^T (X'_i - X'_j) (X'_i - X'_j)^T B) \Paren{ \mathbbm{1}_{(i,j) \in R^2} + \mathbbm{1}_{(i,j) \in \overline{ R^2} }}}^2\\
	&\leq 2 \cdot (\Tr(B^T \Sigma'_{S} B))^2 \\
	&+ 2 \cdot \Paren{  \frac{1}{n^2} \cdot \sum_{(i,j) \in [n]^2} w_i w_j \mathbbm{1}_{(i,j) \in \overline{R^2} } } \Paren{  \frac{1}{n^2} \cdot \sum_{(i,j) \in [n]^2} w_i w_j [\Tr(B^T (X'_i - X'_j) (X'_i - X'_j)^T B)]^2} \;.
	\end{align*}
	Note that $\Tr(B^T (X'_i - X'_j) (X'_i - X'_j)^T B)$ is a sum of square polynomials in $(X'_i - X'_j)$, by \cref{lemma:hypercontractivity-sos-polys} and \cref{lem:emp-pop} $ \frac{1}{n^2} \cdot \sum_{(i,j) \in [n]^2} w_i w_j \Tr(B^T (X'_i - X'_j) (X'_i - X'_j)^T B)^2 \leq O(1) \cdot \Tr(B^T \Sigma(w) B)^2$. Additionally, observe that $\Paren{ \sum_{(i,j) \in [n]^2} w_i w_j \mathbbm{1}_{(i,j) \in \overline{R^2} }} = 2 \alpha_{S}(w) \alpha_{\overline{S}}(w) + \alpha_{\overline{S}}(w)^2 \leq n \cdot O(\alpha_{\overline{S}}(w))$. 
Substituting these back into the above inequality gives us 
	\begin{align*}
	\cA \proves{} \Tr\Paren{ B^T \Sigma(w) B}^2 \leq 2 \cdot (\Tr(B^T \Sigma'_{S} B))^2 + O(1) \cdot  \frac{ \alpha_{\overline{S}}(w)}{n} \cdot (\Tr(B^T \Sigma(w) B))^2 \;.
	\end{align*}
	Rearranging, we get 
	\begin{align*}
	\cA \proves{} (1-O(\alpha_{\overline{S}}(w))/n)	\Tr\Paren{ B^T \Sigma(w) B}^2  \leq  2 \cdot (\Tr(B^T \Sigma'_S B))^2 \;.
	\end{align*}
	Applying \cref{lem:mom_decorruption_diff,lem:emp-pop}, we get 
%
%
	\begin{align*} 
	(\Tr(B^T \Sigma_S' B))^2 &\leq 2 (\Tr(B^T \Sigma_S B))^2 + 2(\Tr(B^T \Sigma'_S B) - \Tr(B^T \Sigma_S B))^2 \\
	&\leq O( k^2 \eps^2) (\Tr(B^T \Sigma'_S B))^2 + O(1+ k^2 \eps^2) (\Tr(B^T \Sigma_S B))^2 \;.
	\end{align*}
	 Substituting this back and using the fact that $\eps \leq 0.1 \cdot k^{-10}$, we get 
	\begin{align*}
	\cA \proves{} (1-O(\alpha_{\overline{S}}(w))/n)	\Tr\Paren{ B^T \Sigma(w) B}^2  \leq O(1) \cdot (\Tr(B^T \Sigma_S B))^2 \;.
	\end{align*}
This completes the proof of the claim.	
\end{proof}

\begin{proof}[Proof of \cref{lem:fake_var_upperbound_quadratic}]
	Note that if $X,Y \sim \cN(\mu(w), \Sigma(w))$ then $Z=X-Y$ follows the distribution $N(0, 2\Sigma(w))$. 
	We can thus write
	\begin{align*} 
	\Var_{X,Y \sim \cN(\mu(w), \Sigma(w))}[p(X-Y)] &= \Var_{Z \sim \cN(0, 2\Sigma(w))}[\iprod{Z, A Z}]\\
	&= \E_{Z\sim \cN(0, 2\Sigma(w))}[\iprod{ ZZ^T, A}^2] - 4 \iprod{ \Sigma(w), A}^2\\
	&= 4 \Paren{\Tr\Paren{ \Sigma(w) A}}^2 + 8 \Tr\Paren{ \Sigma(w) A}^2- 4 \Paren{\Tr\Paren{ \Sigma(w) A}}^2\\
	&=  8 \Tr\Paren{ \Sigma(w) A \Sigma(w) A} \;.
	\end{align*}
	Recall that 
	$$\cA_{\text{matrices}} = \{ (\Sigma^{1/2})^2 = \Sigma(w), (\Sigma^{-1/2} \Sigma^{1/2})^2 = \Sigma^{-1/2} \Sigma^{1/2} \} \cup \{ \Sigma^{-1/2} \Sigma^{1/2} w_i \Paren{X_i' - \mu(w)}  = w_i \Paren{X_i' - \mu(w)} \}_{i \in [n]} \subset \cA \;.$$
	This implies 
	\begin{align*} 
	\Var_{X,Y \sim \cN(\mu(w), \Sigma(w))}[P(X-Y)] &= 8 \Tr\Paren{ \Sigma(w) A \Sigma(w) A}\\
	&= 8 \Tr\Paren{ \Sigma^{1/2} A \Sigma(w) A \Sigma^{1/2}} \;.
	\end{align*}	

 Repeated application of \cref{claim:w-to-good} with $B := A \Sigma^{1/2}$ and $B := A \Sigma_S^{1/2}$ completes the proof.
\end{proof}

\begin{proof}[Proof of \cref{lem:fake_var_lowerbound_quadratic}]
We can write
	\begin{align*}
	\Paren{ \Var_{x, y \sim \cN(\mu(w), \Sigma(w))}[p(x-y)] }^4 &= \Paren{ \frac{k^4}{n^4} \sum_{i,j,r,l \leq n} w_i w_j w_r w_l \Paren{ p(X'_i - X'_j) - p(X'_r - X'_l)}^2 }^4 \\
	&\geq \Paren{ \frac{k^4}{n^4} \sum_{(i,j,r,l) \in S_a^2 \times S_b^2} w_i w_j w_r w_l \Paren{ p(X'_i - X'_j) - p(X'_r - X'_l)}^2 }^4 \\
	&\geq \Paren { \frac{k^4}{n^4} \sum_{(i,j, r, l) \in S_a^2 \times S_b^2} w_i w_j w_r w_l \Paren{ p(Y_i - Y_j) - p(Y_r - Y_l)}^2 }^4 \\
	&\qquad - m_{a,b} \;,
	\end{align*}
where the final step is an application of \cref{lem:mom_decorruption_diff} and 
$$m_{a,b} = O(\eps^2 \cdot k^{20}) \cdot  \Paren{  \|\Sigma_a^{1/2} A \Sigma_a^{1/2}\|_F^8 +  \|\Sigma_b^{1/2} A \Sigma_b^{1/2} \|_F^8 + \|\Sigma^{1/2} A \Sigma^{1/2}\|_F^8 + \iprod{\Sigma_a - \Sigma_b, A}^8} \;.$$
The deterministic conditions imply that for 
$$E_a(A) := \{ (Y_i, Y_j) \in S_a^2 \mid p(Y_i - Y_j) = \iprod{\Sigma_a, A} \pm \log(1/\xi) \cdot \| \Sigma_a^{1/2} A \Sigma_a^{1/2}\|_F \} \;,$$ 
$|E_a| = (1-\xi) n^2$, similarly for $E_b(A)$ and $Y_r, Y_l$. Hence, using $x^2 - 2y^2 < 2(x-y)^2$, we see that
\begin{align*}
&\mathbbm{1}_{(Y_i, Y_j) \in E_a} \mathbbm{1}_{(Y_r, Y_l) \in E_b} \Paren{ p(Y_i - Y_j) - p(Y_r - Y_l)}^2 \\
&= \mathbbm{1}_{(Y_i, Y_j) \in E_a} \mathbbm{1}_{(Y_r, Y_l) \in E_b} \Paren{ \Paren{ \Sigma_a - \Sigma_b, A} \pm \log(1/\xi)\cdot \Paren{ \|\Sigma_a^{1/2}A \Sigma_a^{1/2}\|_F + \|\Sigma_b^{1/2}A \Sigma_b^{1/2}\|_F } }^2 \\
&\geq \frac{1}{2} \cdot \mathbbm{1}_{(Y_i, Y_j) \in E_a} \mathbbm{1}_{(Y_r, Y_l) \in E_b} \Paren{ \Paren{ \Sigma_a - \Sigma_b, A}^2 - 4\cdot \log(1/\xi)^2 \cdot \Paren{ \|\Sigma_a^{1/2}A \Sigma_a^{1/2}\|^2_F + \|\Sigma_b^{1/2}A \Sigma_b^{1/2}\|^2_F } } \;.
\end{align*}
Hence, in the case $(Y_i, Y_j) \in S_a^2$ and $(Y_r, Y_l) \in S_b^2$ we get
\begin{align*}
&\Paren{ \frac{k^4}{n^4} \sum_{(i,j, k, l) \in S_a^2 \times S_b^2} w_i w_j w_r w_l \Paren{ p(Y_i - Y_j) - p(Y_r - Y_l)}^2 }^4 \\
&\geq \Paren{ \frac{k^4}{n^4} \sum_{(i,j, k, l) \in S_a^2 \times S_b^2} w_i w_j w_r w_l \mathbbm{1}_{(Y_i, Y_j) \in E_a} \mathbbm{1}_{(Y_r, Y_l) \in E_b} \cdot (p(Y_i - Y_j) - p(Y_r - Y_l))^2 }^4 \\
& \geq \Paren{ \frac{k^4}{n^4} \sum_{(i,j, k, l) \in S_a^2 \times S_b^2} w_i w_j w_r w_l \mathbbm{1}_{(Y_i, Y_j) \in E_a} \mathbbm{1}_{(Y_r, Y_l) \in E_b}}^4 \cdot \min_{(Y_i,Y_j) \in E_a, (Y_r, Y_l) \in E_b} (p(Y_i - Y_j) - p(Y_r - Y_l))^8\\
&\geq \Paren{ (k/n)^4 \alpha_a(w)^4 - O(\xi) }^4 \cdot \Paren{ (k/n)^4 \alpha_b(w)^4 - O(\xi) }^4 \\
&\qquad \cdot \Paren{ \Paren{ \Sigma_a - \Sigma_b, A}^8 - O(1) \cdot \log(1/\xi)^8 \cdot \Paren{ \|\Sigma_a^{1/2}A \Sigma_a^{1/2}\|^8_F + \|\Sigma_b^{1/2}A \Sigma_b^{1/2}\|^8_F } }\\
&\geq \Paren{ (k/n)^{32} \alpha_a(w)^{16} \alpha_b(w)^{16} - O(\xi^4) } \\
&\qquad \cdot \Paren{ \Paren{ \Sigma_a - \Sigma_b, A}^8 - O(1) \cdot \log(1/\xi)^8 \cdot \Paren{ \|\Sigma_a^{1/2}A \Sigma_a^{1/2}\|^8_F + \|\Sigma_b^{1/2}A \Sigma_b^{1/2}\|^8_F } } \;.
\end{align*}
In the final few inequalities we use $ x^2 - 2y^2 < 2(x-y)^2$ repeatedly. Putting everything together, we see 
 \begin{align*}
 \Paren{ \Var_{x, y \sim \cN(\mu(w), \Sigma(w))}[p(x-y)] }^4 &\geq \Paren{ (k/n)^{32} \alpha_a(w)^{16} \alpha_b(w)^{16} - O(\xi^4) - O(\eps^2 \cdot k^{20})  }\iprod{\Sigma_a - \Sigma_b, A}^8 \\
 &\qquad- O(1) \cdot \Paren{ \eps^2 k^{20} + \log(1/\xi)^8 } \cdot  \Paren{ \|\Sigma_a^{1/2} A \Sigma_a^{1/2}\|_F^8 +  \|\Sigma_b^{1/2} A \Sigma_b^{1/2} \|_F^8 } \\
 &\qquad -O(\eps^2 k^{20}) \cdot \|\Sigma^{1/2} A \Sigma^{1/2}\|_F^8 \;.
 \end{align*}
Rearranging this, and using the fact that $\|\Sigma^{1/2} A \Sigma^{1/2}\|_F^8  = \Paren{ \Var_{x, y \sim \cN(\mu(w), \Sigma(w))}[p(x-y)] }^4$, completes the proof of our lemma.
\end{proof}


\section{From Small Overlap to Large Parameter Distance}
\label[section]{sec:tv-param}

In this section we prove \cref{lem:TV_param}.
Much of the argument is implicit in \cite{DKS17-sq}. Here we make it explicit.

\begin{proof}[Proof of \cref{lem:TV_param}]
If there is a unit vector $v$ such that $|\iprod{v,\mu_P - \mu_Q}| > 0.1 \sqrt{\log(1/\epsilon)} \cdot \sqrt{\iprod{v,(\Sigma_P + \Sigma_Q)v}}$, 
then case (1) holds and we are done.
Suppose no such $v$ exists.

Note that the total variation distance between distributions on $\R^d$ is unaffected by affine transformations.
Letting $P' = \cN(\Sigma_P^{-1/2} \mu_P,I)$ and $Q' = \cN(\Sigma_P^{-1/2} \mu_Q, \Sigma_P^{-1/2} \Sigma_Q \Sigma_P^{-1/2})$, 
we may assume $\dtv(P',Q') \geq 1 - \epsilon$.
Let $A = \Sigma_P^{-1/2} \Sigma_Q \Sigma_P^{-1/2}$.
Let $H^2(P, Q)$ denote the squared Hellinger distance between distributions $P$ and $Q$ and use the following formula for the squared Hellinger distance between two multivariate normal distributions
  	\[ H^2(\mathcal N(\mu_1, \Sigma_1), \mathcal N(\mu_2, \Sigma_2)) = 1 - \frac{ \det (\sum_1)^{1/4} \det (\sum_2) ^{1/4}} { \det \left( \frac{\sum_1 + \sum_2}{2}\right)^{1/2} }
  	\exp\left\{-\frac{1}{8}(\mu_1 - \mu_2)^T
  	\left(\frac{\sum_1 + \sum_2}{2}\right)^{-1}
  	(\mu_1 - \mu_2)
  	\right\} \;, \]
  	and so
  	\[ H^2(P', Q') = 1 - \frac{  \det (A) ^{1/4}} { \det \left( \frac{I + A}{2}\right)^{1/2} }
  	\exp\left\{-\frac{1}{8}(\mu_P - \mu_Q)^T \Sigma_P^{-1/2}
  	\left(\frac{I + A }{2}\right)^{-1}
  	\Sigma_P^{-1/2} (\mu_P - \mu_Q) \right\}. \]
  	Since $\sqrt{H^2(P', Q')(2-H^2(P',Q'))} \geq \dtv(P', Q') \geq 1 - \eps$, we have $H^2(P',Q') \geq 1 - \sqrt{2\eps}$ and hence
  	\begin{equation}\label{eqn:hell_upbound}
  	\frac{\det(A)^{1/4}} {\det\left(\frac{I+ A}{2}\right)^{1/2}} \leq \sqrt{2 \epsilon} \cdot\exp\left\{\frac{1}{8}(\mu_P - \mu_Q)^T \Sigma_P^{-1/2} \left(\frac{I+A}{2}\right)^{-1} \Sigma_P^{-1/2} (\mu_P - \mu_Q) \right\} \;.
  	\end{equation}
Since we are assuming condition (1) does not hold, if we let $v = \Sigma_P^{-1/2} \left(\frac{I + A}{2}\right)^{-1} \Sigma_P^{-1/2} (\mu_P - \mu_Q)$, then
  	\begin{align*}
  	\iprod{ \mu_P - \mu_Q,  v}&  \leq  \sqrt{ \log(1/\epsilon) } \cdot \sqrt{ \iprod{v,(\Sigma_P + \Sigma_Q) v} }\\
        & \leq \sqrt{ \log(1/\epsilon) } \cdot \sqrt{\iprod{\mu_P - \mu_Q, v}} \;,
  	\end{align*}
        using the definition of $v$.
  	And so, $\iprod{ \mu_P - \mu_Q,  v} \leq \log 1/\epsilon$.
        Plugging this back into Equation~\eqref{eqn:hell_upbound}, we get
  	\begin{equation*}
  	\frac{\det(A)^{1/4}} {\det\left(\frac{I+ A}{2}\right)^{1/2}} \leq \sqrt{2\epsilon} \cdot (1/\epsilon)^{1/8} \leq \sqrt 2 \cdot \e^{3/8} \mper
  	\end{equation*}
  	Let the eigenvalues of $A$ be given by $\lambda_1,\ldots,\lambda_d$.
        Expanding $\det A = \prod_{i=1}^d \lambda_i$ and $\det (I+A)/2 = 2^{-d} \prod_{i=1}^d(1+\lambda_i)$ and taking logs, then using an argument identical to \cite{DKS17-sq}, Lemma B.4, we get
  	\begin{equation}\label{eqn:hell_lb} 
          \sum_{i=1}^d \min(|\log \lambda_i|, (\log \lambda_i)^2) \geq \Omega(\log(1/\epsilon))\mper
        \end{equation}
  	Hence, when condition (1) does not hold, Equation~\eqref{eqn:hell_lb} holds.
    If condition (2) also fails, $\lambda_i \in [1/\sqrt{\log(1/\epsilon)},\sqrt{\log(1/\epsilon)}]$.
    Under these circumstances, we want to show that case (3) holds.
    By \cref{clm:tv-params-1},
    \[
    \| I - A \|_F^2 \geq \Omega\Paren{\frac{\log(1/\epsilon)}{\log \log(1/\epsilon)}}\mcom
    \]
    completing the proof.
\end{proof}

\begin{claim}\label[claim]{clm:tv-params-1}
    For all $a \geq e$ and all $x \in [1/a,a]$,
    \[
    \min(|\log x|,(\log x)^2) \leq O(\log a) \cdot (1-x)^2 \;.
    \]
\end{claim}
\begin{proof}
    To see this, first observe that $\min(|\log x|,(\log x)^2) = (\log x)^2$ if $x \in [1/e,e]$ and otherwise $\min(|\log x|, (\log x)^2) = |\log x|$.

    For $x \in [1/e,e]$, the maximum value of $(\log x)^2 / (1-x)^2$ is $1/(1/e-1)^2$ when $x = 1/e$.
    So we are done for $x \in [1/e,e]$.
    
    Next we must show $|\log x| \leq O(\log a) (1-x)^2$ for $x \in [1/a,1/e] \cup [e,a]$.
    Consider the case $x \in [e,a]$.
    Then we have $0 \leq \log x = \log(1 + (x-1)) \leq (x-1) \leq (x-1)^2$, using $x-1 \geq 0$.
    The last case to address is $x \in [1/a,1/e]$.
    In this case, $(1-x)^2 \geq (1-1/e)^2$ and $|\log x| \leq \log a$, which completes the proof.
\end{proof}

\section{Reduction to Constant-Accuracy Clustering and Proof of Main Theorem}
\label[section]{sec:constant-accuracy-and-main-theorem}

In this section, we show that a list of candidate rough clusterings suffices to obtain
our desired robust learning algorithm. We start by showing (Proposition~\ref{prop:const-acc})
that if we had a single rough clustering, this would suffice. Unfortunately, our clustering algorithm 
does not provide us with a single set of clustered points. Instead, it gives us a large list of
hypothesis clusterings with the high probability guarantee that at least one of them is correct. 
It turns out that using this weaker guarantee also suffices for our purposes.



Let the overlap of two distributions $P$ and $Q$ be $L(P,Q) = \int \min \{ dp,dq \}$. 
Note that this is $1-\dtv(P,Q)$. Let $h(P,Q) = -\log(L(P,Q))$. We will need the following lemma.
\begin{lemma}\label{largeOverlapLem}
Let $A$ and $B$ be two Gaussians with $h(A,B) = O(1)$. If $D \in \{ A, B\}$, then
\[ \Pr_{x \sim D} \left[1/\eps \geq  \frac{A(x)}{B(x)} \geq \eps \right] \geq 1- \poly(\eps) \;. \]
\end{lemma}
\begin{proof}
Without loss of generality, let $A$ be $\N(0, I)$ and take the probability measure with respect to $A$, i.e., $D = A$. 
The ratio $B(x)/A(x) = \exp(q(x))$ for some quadratic polynomial $q$. 
Since the overlap of $A$ and $B$ is large, $q$ must have constant $L_2$ norm with respect to $A$. 
The result then follows from Gaussian concentration for degree-$2$ polynomials.
\end{proof}

We will need a triangle inequality using ideas from~\cite{DKS17-sq}.

\begin{lemma}\label{triangleLem}
If $A,B,C$ are Gaussian distributions with $h(A,B)=O(1)$, then
\begin{equation}\label{eqn:hell_triangle}
h(A,C) = O(1 + h(B,C)) \;.
\end{equation}
\end{lemma}
\begin{proof}
From~\cite{DKS17-sq} it follows that for Gaussians $G$ and $H$ it holds
$h(G,H) = h_\Sigma(G,H)+h_\mu(G,H)$, 
where $h_\Sigma(G,H) = h(G',H')$ and $G'$, $H'$ are the mean $0$ translates of $G$ and $H$, respectively. 
By Proposition B.6 of that paper,
$$
h_\mu(G,H) = \inf_x ((x-\mu_G)\Sigma_G^{-1}(x-\mu_G)+(x-\mu_H)\Sigma_H^{-1}(x-\mu_H))/4 \;.
$$
From Proposition B.3 of that paper we have that
$$
h_\Sigma(A,C) = O(h_\Sigma(A,B)+h_\Sigma(B,C)) \;.
$$
Therefore, it suffices to prove that
$$
h_\mu(A,C) = O(1 + h_\mu(B,C)) \;.
$$
Note that we have
$$
h_\mu(B,C) = ((x-\mu_B)\Sigma_B^{-1}(x-\mu_B)+(x-\mu_C)\Sigma_C^{-1}(x-\mu_C))/4 \;,
$$
for some appropriately chosen $x$. Using that same value of $x$, we have that
$$
h_\mu(A,C) \leq ((x-\mu_A)\Sigma_A^{-1}(x-\mu_A)+(x-\mu_C)\Sigma_C^{-1}(x-\mu_C))/4 \;.
$$
The latter term is $O(h_\mu(B,C))$, while the former is $E=(x-\mu_A)\Sigma_A^{-1}(x-\mu_A)$. 
Since $h(A,B)=O(1)$, it is clear that $\Sigma_A=\Theta(\Sigma_B)$, 
and hence 
$$E=O((x-\mu_A)\Sigma_B^{-1}(x-\mu_A)) = O((x-\mu_B)\Sigma_B^{-1}(x-\mu_B)+(\mu_A-\mu_B)\Sigma_B^{-1}(\mu_A-\mu_B)) \;,$$
which is clearly $O(h_\mu(B,C)+1)$.
\end{proof}

\begin{proposition}\label[proposition]{prop:const-acc}
Let $X=\sum_{i=1}^k w_i G_i$ be a $k$-mixture of Gaussians with $h(G_i,G_j)$ at least a sufficiently large constant multiple of $\log(k/\eps)$ for all $i\neq j$. Let $X'$ be an $\eps$-corrupted version of $X$ and $N$ a sufficiently large polynomial in $d k/\eps$. Let $S$ be a set of $N$ random samples from $X'$. 
Let $T_1,\ldots,T_k\subset S$ be sets of samples so that, for some sufficiently small constant $c>0$, 
if $S_i$ is the set of samples in $S$ that were drawn from the $i$-th component of $X$, 
then $|T_i\cap S_i| \geq (1-c)\min(|T_i|,|S_i|)$ for all $i \in [k]$. 
There exists an algorithm that given $S$, the $T_i$'s and an additional set of $N$ independent 
samples from $X'$, returns a set of weights $u_i$ such that $\sum_{i=1}^k |u_i - w_i|  = O(\eps)$
and a list of Gaussians $H_1,\ldots,H_k$ 
so that $\sum_{i=1}^k w_i \dtv(H_i,G_i) \leq O(\eps\log(1/\eps))$ for all $i \in [k]$.
\end{proposition}
\begin{proof}
First, we note that we can assume that $w_i\geq \eps/k$ for all $i \in [k]$. 
The clusters with smaller weight can be identified with high probability 
by the size of the corresponding set $T_i$. 
The contribution to our sum of $w_i\dtv(H_i,G_i)$ from such clusters is $O(\eps)$, 
and treating all samples from such clusters as an additional source of error 
keeps the error rate at $O(\eps)$. 
We will assume this throughout the rest of the argument.

For each remaining cluster $T_i$, we note that $T_i$ has $1-c$ overlap 
with a large set of random samples from $G_i$. 
This means that it can be thought of as a set of independent samples from $G_i$ with a $c$-fraction of them adversarially corrupted. 
Applying the algorithm from~\cite{DKKLMS16} for robustly learning a single Gaussian to each component cluster, 
we can learn Gaussians $\tilde H_i$ so that, $\dtv(G_i,\tilde H_i) < 1/10$, and thus $h(G_i,\tilde H_i)=O(1)$.

Note that, by repeated applications of Lemma \ref{triangleLem}, this implies that 
for $i\neq j$ that $h(G_i, \tilde H_j)$ is at least a large multiple of $\log(k/\eps)$. 
This in turn implies that $\int \min(g_i(x),\tilde h_j(x))dx$ is a low-degree polynomial in $\eps/k$, 
where $g_i(x)$ and $\tilde h_j(x)$ are the probability density functions of $G_i$ and $\tilde H_j$, respectively. 
This on the other hand is equal to
$$
\E_{x\sim G_i}[ \min(1,\tilde H_j(x)/G_i(x))] \;.
$$
This means that except with probability $(\eps/k)$ we have that $\tilde H_j(x)/G_i(x)$ is at most 
a small polynomial in $(\eps/k)$.

Given these $\tilde H_i$, we cluster the remaining $N$ samples in the following way. 
We associate each sample $x$ to the cluster $i$ for which the probability density function of $\tilde H_i(x)$ is largest. 
We claim that among the samples drawn from the component $G_i$, all but an $O(\eps/k)$-fraction of them end up in cluster $i$. 
This is because since $h(G_i,\tilde H_i)=O(1)$, Lemma \ref{largeOverlapLem} implies that for all but an $\eps/k$ fraction 
of these samples we have that $\tilde H_i(x)/G_i(x) > \poly(\eps/k)$. However, for all but an $\eps/k$-fraction 
of these samples we have that all other $\tilde H_j(x)/G_i(x)$ are less than a smaller polynomial in $\eps/k$.

Therefore, all but an $\eps/k$ fraction of samples from $G_i$ are assigned to the $i$-th cluster. 
Hence, among all $N$ samples, all but an $\eps$-fraction of them are non-corrupted 
samples that are assigned to the correct cluster. For each $i$, let $\eps_i$ be the fraction of samples 
assigned to the $i$-th cluster coming from other components 
or from errors plus the fraction of samples from the $i$-th component that were corrupted. 
It is easy to see that $\sum_{i=1}^k w_i\eps_i = O(\eps)$. Finally, running the robust estimation algorithm 
of \cite{DKKLMS16}, we can compute estimates $H_i$ so that $\dtv(H_i,G_i)=O(\eps_i \log(1/\eps_i))$ with high probability. 
Thus, the final error is $\sum w_i O(\eps_i \log(1/\eps_i))$, which is $O(\eps \log(1/\eps))$ 
by Jensen's inequality applied to the concave function $x\log(1/x)$.

Finally, we note that learning the desired hypothesis weights $u_i$ so that $\sum |u_i-w_i| = O(\eps)$ is quite simple. 
It suffices to take sufficiently many random samples from $X'$ and define the $u_i$'s based on 
what fraction of these samples lie within each bin. This completes the proof.
\end{proof}

The above proposition would have been sufficient if our clustering algorithm returned
a single rough clustering. 
We can use Proposition~\ref{prop:const-acc} with an appropriate hypothesis testing routine 
(see, e.g.,~\cite{DDS12stoclong, DDS15})
to handle the fact that our clustering algorithm produces a list of candidate clusterings one of which 
is accurate.

\begin{corollary} \label{cor:list}
Letting X and X' be as in Proposition~\ref{prop:const-acc} with $\min w_i$ at
least a sufficiently large multiple of $\eps \log(1/\eps)$, a set of $S$
samples from $X'$, and given a collection $\mathcal{H}$ of $k$-tuples of sets of
samples $T_1,\ldots,T_k$ so that at least one such $k$-tuple satisfies the
conditions in Proposition~\ref{prop:const-acc}, there exists an algorithm that given
$\mathcal{H}$, and $\poly(\log(|\mathcal{H}|)kn/\eps)$ additional samples from
$X'$, returns a list of weights $u_1, \ldots, u_k$ and Gaussians $H_1,\ldots,H_k$ so that 
$\sum_i |u_i - w_i| = O(\eps)$ and $\sum_i w_i \dtv(H_i, G_i) = O(\eps \log(1/\eps))$. The algorithm runs in time
polynomial in the number of samples and $|\mathcal{H}|$.
\end{corollary}
\begin{proof}
For each clustering in $\mathcal{H}$, we can apply the algorithm
of Proposition~\ref{prop:const-acc} to learn a distribution $X_i = \sum u_i H_i$. 
We note that for any clustering with $T_i$ having $(1-c)$-agreement with the
samples from $G_i$ for each $i$, by Proposition~\ref{prop:const-acc}, we have that
$\dtv(X_i,X) = O(\eps \log(1/\eps))$ and thus by the triangle inequality
$\dtv(X_i,X') = O(\eps \log(1/\eps))$. For technical reasons that will
become apparent later, we throw away any hypothesis $X_i$ for which the
components do not have pairwise overlap bounded by a sufficiently
large power of $(\eps/k)$. We note that at least one close hypothesis
should remain. From there, some of the additional samples can be used
to run a tournament to find one of the $X_i$ (call it $X_0$) so that
$\dtv(X_0,X') \leq 3 \min_i \dtv(X_i,X')+\eps = O(\eps \log(1/\eps))$.

If our goal was merely to learn $X$ to small total variation distance,
we would be done. Since we want to do parameter estimation, further work
is required. We note that both $X$ and $X_0$ are mixtures of Gaussians
whose clusters have pairwise overlap bounded by a large power of
$(\eps/k)$, and furthermore that $\dtv(X,X_0) = O(\eps \log(1/\eps))$. 
By Theorem B.3 of~\cite{DKS17-sq}, this implies that for some appropriate permutation
of the clusters we have that $\dtv(u_i H_i, w_i G_i) = O(\eps \log(1/\eps))$, 
which implies that $\dtv(H_i, G_i) < 1/2$, or $h(H_i,G_i) = O(1)$. Given this, 
we can divide our additional samples of $X$ into clusters based on which $H_i$ 
provides them with the highest probability density value. 
As argued in the proof of Proposition~\ref{prop:const-acc} 
all but an $\eps$-fraction of these samples are correctly clustered. 
Feeding these clusters back into Proposition~\ref{prop:const-acc}  
(or just using a robust learning algorithm on them directly) gives our result.
\end{proof}

\begin{proof}[Proof of \cref{thm:main-informal}]
In the case of uniform mixtures, our main theorem follows immediately by putting together \cref{lem:const_acc} and Corollary~\ref{cor:list}.
See \cref{sec:arbit-weights} for the case of nonuniform mixtures.
\end{proof}

\section*{Acknowledgements}
We thank Pravesh Kothari for helpful conversations in the early phases of this project.
We also thank the Simons Institute and especially the program on Foundations of Data Science for hosting 
some of the authors while much of this work was completed.
We are grateful to Ainesh Bakshi and Pravesh Kothari for sharing a draft of \cite{BK20a} with us, which helped us to identify the issue discussed in \cref{sec:concurrent}.


  \phantomsection
  \addcontentsline{toc}{section}{References}
  \bibliographystyle{amsalpha}
  \bibliography{allrefs}

  \appendix

\section*{APPENDIX}


\section{Omitted Proofs from Section~\ref{sec:rough-clustering}} \label[section]{sec:omitted}

\begin{proof}[Proof of \cref{lem:rounding}]
  The algorithm \textsc{split} iterates the following subroutine to build the list of subsets $R_1,\ldots,R_m$.
  \begin{enumerate}
    \item Given current list of subsets $R_1,\ldots,R_t$, let $R(w) = \sum_{i \notin R_1 \cup \ldots \cup R_t} w_i$.
    \item Find $\pE$ which maximizes $\pE R(w)$ over all $\pE$ in $\cC$.
    \item Choose a random $i \sim [n]$ with probability $p_i = \pE w_i / \sum_{j \leq n} \pE w_j$.
    \item Create the set $R_{t+1} \subseteq [n]$ by including each $j \in [n]$ independently with probability $\pE w_i w_j / \pE w_i$.
  \end{enumerate}
  
  Let us first analyze the expected size of each $R_t$.
  At any iteration $t$ of \textsc{split}, we have
  \[
  \E |R_t| = \E_{i} \sum_{j \in [n]} \E \Ind(j \in R_t \text{ given that row $i$ is chosen}) = \frac{\sum_{i \leq n} \pE w_i w_j}{\sum_{i \leq n} w_i} = \frac{(n/k)^2}{n/k} = \frac n {k}
  \]
  where we have used that $\pE$ satisfies $\sum_{i \leq n} w_i = n/k$.

  Next let us analyze the expected fraction of $R_t$ which lies in one of $S$ or $T$.
  Consider first the case that \textsc{split} chooses $i \in S'$.
  Then
  \[
  \E \Brac{ \Abs{R_t \cap T'} \, | \, i \in S'} = \frac 1 { \sum_{i \in S'} \pE w_i} \sum_{i \in S', j \in T'} \pE w_i w_j \leq \frac 1 { \sum_{i \in S'} \pE w_i} \cdot \delta n^2
  \]
  by hypothesis.
  It follows that
  \[
  \E \min( |R_t \cap S'|, |R_t \cap T'| ) \leq \delta n^2\mper
  \]

  Finally, let $R'_t = \cup_{s \leq t} R_t$ be the indices chosen by the algorithm in some iteration up to $t$.
  How many new indices are chosen in iteration $t+1$?
  That is, we want to bound $|R'_{t+1} \setminus R'_t|$
  For this, we need a lower bound on $\pE \sum_{i \in [n], j \notin R'_t} w_i w_j$.

  By hypothesis, the uniform distribution on (indicator vectors for) the sets $S_1,\ldots,S_k$ is in $\cC$.
  Note that
  \[
  \E_{a \sim [k]} |S_a| |S_a \cap R'_t| = \frac n k \cdot \E_{a \sim [k]} |S_a \cap R'_t| = \frac n {k^2} |R'_t|\mper
  \]
  So, we may assume that the $\pE$ found by \textsc{split} at iteration $t+1$ has $\pE \sum_{i \in [n], j \notin R'_t} w_i w_j \geq \frac {n (n - |R'_t|)}{k^2}$.

  By the same analysis as before,
  \[
  \E |R'_{t+1} \setminus R'_t| \geq \frac {n - |R'_t|}{k}\mper
  \]
  Thus, taking $O(k \log 1/\eta)$ iterations,
  \[
  \E |R'_t| \geq n (1-\eta)\mper
  \]
\end{proof}

\begin{proof}[Proof of \cref{cor:TV_param}]
Let $\eta = k^{-(C+2)}/2$.
We start with the case that there is a hyperplane separating the mixture.

\textbf{Case 1:}
  Suppose there is a direction $v$ such that $a,b \in [k]$ which maximize $\iprod{\mu_a - \mu_b, v}^2$ in fact satisfy $\iprod{\mu_a - \mu_b, v}^2 \geq \max_{p \in [k]} \iprod{v, \Sigma_p v} / \eta$.
  Note that
   \[
  \iprod{v, \Sigma v} = \frac 1 {k^2} \sum_{a,b \leq k} \iprod{\mu_a - \mu_b, v}^2 + \frac 1 k \sum_{a \leq k} \iprod{v, \Sigma_a v}\mper
  \]
  In this case, $\frac 1 k \sum_{a \leq k} \iprod{v, \Sigma_a v} \leq \eta \iprod{\mu_a - \mu_b, v}^2$, and hence we have that $\iprod{\mu_a - \mu_b, v}^2 \geq \iprod{v, \Sigma v}/2$.
  Furthermore, by averaging there is a $t \in \R$ such that every $p \in [k]$ has $|\iprod{v, \mu_p} - t| \geq |\iprod{\mu_a - \mu_b, v}|/k \geq (\eta^{1/2} k)^{-1} \max_p \iprod{v,\Sigma_p v}^{1/2}$.
 Let $S = \{ a \in [k] \, : \, \iprod{v, \mu_p} \leq t \}$ and $T = [k] \setminus S$.
 Then for all $a \in S, b \in T$, we have
 \[
 \iprod{\mu_a  - \mu_b, v}^2 \geq \max \left \{ \frac{\iprod{v, \Sigma_a v} + \iprod{v, \Sigma_b v}}{2 \eta k^2}, \frac{\iprod{v, \Sigma v}}{ k^2 } \right \} \geq \left \{ \frac{\iprod{v, \Sigma_a v} + \iprod{v, \Sigma_b v}}{k^C}, \frac{\iprod{v, \Sigma v}}{ k^2 } \right \}
 \]

 \textbf{Case 2:} Suppose there is a direction $v \in \R^d$ such that if we assume WLOG that $\iprod{v, \Sigma_1 v} \geq \ldots \geq \iprod{v, \Sigma_k v}$ then there is an index $i > 1$ such that $\iprod{v, \Sigma_i v} \leq k^{-(C')^i} \iprod{v, \Sigma v}$.
 Then if we let $S = \{j \in [k] \, : \, j < i \}$ and $T = [k] \setminus S$, we have that for all $a \in S$ and $b \in T$ that
 \[
 \frac{\iprod{v, \Sigma_b v}}{\iprod{v, \Sigma v}} \leq k^{-(C')^i} \leq (k^{-(C')^{i-1}})^{C'} \leq \Paren{\frac{\iprod{v, \Sigma_a v}}{\iprod{v, \Sigma v}}}^{C'}\mcom
 \]
 and, furthermore, that $\iprod{v, \Sigma_b v} \leq k^{-C'} \iprod{v, \Sigma_a v}$.

  \textbf{Case 3:} Now suppose that cases 1 and 2 fail.
  By the failure of case 2, in every direction $v$, if we order $\iprod{v, \Sigma_1 v} \geq \ldots \geq \iprod{v, \Sigma_k v}$, then we have both $\iprod{v, \Sigma_1 v} \leq k \iprod{v, \Sigma v}$ and $\iprod{v, \Sigma_k v} \geq k^{-(C')^k} \iprod{v, \Sigma v}$.
  So $\iprod{v, \Sigma_1 v} / \iprod{v, \Sigma_k v} \leq k^{(C')^k + 1}$.
  Thus, no pair of Gaussians in the mixture is $\epsilon$-variance separated (see \cref{lem:TV_param}) for $\epsilon = 2^{-k^{6(C')^k + 6}}$.

  By the failure of case $1$, in every direction $v$ the maximal mean separation $\iprod{\mu_a - \mu_b, v}^2$ is bounded by
  \[
  \iprod{\mu_a - \mu_b,v}^2 \leq k^{C+1} \iprod{v, \Sigma_1 v} \leq k^{C+1} \cdot k \cdot \iprod{v, \Sigma v} \leq k \cdot k^{C+1} \cdot k^{(C')^k} \iprod{v, \Sigma_k v}\mper
  \]
  So, no pair of Gaussians in the mixture can be $\epsilon$-mean separated, for $\epsilon = 2^{- k^{6(C')^k + 6C + 6}}$.
  
  Let $\epsilon = 2^{-k^{6(C')^k + 6C + 6}}$.
  Since we have assumed that every pair of Gaussians in the mixture has overlap at most $= 2^{- k^{10(C')^k + 10C + 10}}$, all pairs must be $\epsilon$-covariance separated.

  Since every pair fails to be $\e$-variance separated, every pair $a,b \in [k]$ has $\Omega(\log \e)^{-1/6} \Sigma_a \preceq \Sigma_b \preceq O(\log 1/\e)^{1/6} \Sigma_a$.
  And, by \cref{lem:TV_param}, every pair $a \neq b$ has $\|I - \Sigma_a^{-1/2} \Sigma_b \Sigma_a^{-1/2}\|_F^2 \geq \Omega( \log 1/\e / \log \log 1/\e)$.
  Furthermore, we claim that $\Sigma \preceq O(\log 1/\e) \Sigma_a$.
  To see this, note that for any $v \in \R^d$,
  \begin{align*}
    \iprod{v, \Sigma v}
    & = \frac 1 {k^2} \sum_{p,q \leq k} \iprod{\mu_p - \mu_q, v}^2 + \frac 1 k \sum_{p \leq k} \iprod{v, \Sigma_p v}\\
    & \leq \frac 1 {k^2} \sum_{p,q \leq k} O(\log 1/\e)^{1/6} (\iprod{v, \Sigma_p v} + \iprod{v, \Sigma_q v}) + \frac 1 k \sum_{p \leq k} \iprod{v, \Sigma_p v} \\
    & \leq O(\log 1/\e)^{1/3} \iprod{v, \Sigma_a v}
  \end{align*}
  where we have used twice that case 1 fails.
  It follows that $\|\Sigma^{1/2} \Sigma_a^{-1/2}\|_2^2 \leq O(\log 1/\e)^{1/3}$, and $\|\Sigma_b^{-1/2} \Sigma_a^{1/2}\|_2^2 \leq O(\log 1/\e)^{1/3}$.

  Now we can simply compute that
  \[
  \|\Sigma_a^{1/2} A_{ab} \Sigma_a^{1/2} \|_F = \|I - \Sigma_a^{-1/2} \Sigma_b \Sigma_a^{-1/2} \|^2 \leq O(1) \cdot \frac{\|I - \Sigma_a^{-1/2} \Sigma_b \Sigma_a^{-1/2}\|_F}{ \log 1/\e / \log \log 1/\e}
  \]
  and
  \[
  \|\Sigma^{1/2} A_{ab} \Sigma^{1/2}\|_F \leq \|\Sigma_{1/2} \Sigma_a^{-1/2}\|_2^2 \|I - \Sigma_a^{-1/2} \Sigma_b \Sigma_a^{-1/2} \|_F \leq O\Paren{\frac{\|I - \Sigma_a^{-1/2} \Sigma_b \Sigma_a^{-1/2}\|_F^2}{(\log 1/\e)^{1/3}}}
  \]
  and similarly for $\|\Sigma_b^{1/2} A_{ab} \Sigma_b^{1/2}\|$.
  \end{proof}

\begin{proof}[Proof of Lemma~\ref{lem:concentration}]
An application of Hoeffding's inequality and a union bound implies that in $n$ samples, each $S_a$, for $a \in [k]$, 
has $((1 \pm \delta)/k) n$ samples with probability $1- O(k)\exp(-O(1) \cdot (\delta/k)^2 \cdot n)$. We will henceforth assume 
that this event occurs. 
	
We now bound from above the probability of each of the following events defined in~\cref{def:deterministic} 
over $X_1, \dots, X_{|S_a|}$ drawn from $G_a$. 
\begin{enumerate}
\item $E_1(a) := \exists v \in \R^d.~\Pr_{X \sim \{X_1, \dots, X_{|S_a|} \} } [\iprod{X - \mu_a, v}^2 \geq \log(1/\xi) \cdot \iprod{v, \Sigma_a v}] \geq (1-\xi)$. 

\item $E_2(a) := \exists v \in \R^d.~\Pr_{X, Y \sim \{X_1 \dots, X_{|S_a|} \}}[\iprod{X- Y, v}^2 \leq \xi^2 \iprod{v, \Sigma_a v} ] \geq (1-\xi)$. 

\item $E_3(a) := \exists A \in \R^{d \times d}.~\Pr_{X, Y \sim \{X_1 \dots, X_{|S_a|} \}}[| \iprod{(X - Y)^{\otimes 2}, A} - 2\iprod{\Sigma_a, A} | \geq O(1) \log(1/\xi) \cdot \| \Sigma_a A\|_F] \geq (1-\xi)$. 

\item $E_4(a) :=$ for all even $s \leq t$ for all $i \in [k]$,
	\[
	\Norm{\frac 1 {|S_a|} \sum_{i \in S_a} [\overline{\Sigma}_a^{-1/2} (X_i - \overline{\mu}_a)]^{\tensor s} - M_s}^2 \leq d^{-2t} \delta \;,\] 
	where $M_s = \E_{g \sim \cN(0,\Id)} g^{\tensor s}$ is the $s$-th moment tensor of the standard Gaussian and
	\[
	\overline{\mu}_a = \frac 1 {|S_a|} \sum_{i \in S_i} X_i \, , \text{ and } \overline{\Sigma}_a = \frac 1 {|S_a|} \sum_{i \in S_a} (X_i - \overline{\mu}_i)(X_i - \overline{\mu}_a)^\top \;.
	\]
\end{enumerate}

Events $E_1, E_2, E_3$ occur with vanishing probability by standard arguments on uniform convergence of empirical CDFs to population CDFs.
(To obtain uniform convergence across all $v \in \R^d, A \in \R^{d \times d}$, it suffices to observe that the VC dimension of halfspaces in $\R^m$ is at most $\poly(m)$.)

Getting bounds on the probability of $E_4(a)$ is more complicated. 
Recall that $X$ is drawn from $\cN(\mu_i, \Sigma_i)$. Let $Y \sim \cN(0, I)$. 
Then $X$ can be simulated by drawing $\Sigma_i^{1/2}Y + \mu_i$.  
Since $\overline{\Sigma}_i^{-1/2} (X - \overline{\mu}_i) = \overline{\Sigma}_i^{-1/2} \Sigma_i^{1/2} Y +  \overline{\Sigma}_i^{-1/2} (\mu_i - \overline{\mu}_i)$. The question hence reduces to showing 
\[
\Norm{\frac 1 {|S_i|} \sum_{j \in S_i} [\overline{\Sigma}_i^{-1/2} \Sigma_i^{1/2} Y_j +  \overline{\Sigma}_i^{-1/2} (\mu_i - \overline{\mu}_i)]^{\tensor s} - M_s}_F^2 \leq d^{-2t} \delta \;, 
\] 
where $Y_j \sim \cN(0, I)$. We will make use of the following concentration inequality (see, e.g.,~\cite{ODonnell:BFA}): 
\begin{lemma}\label[lemma]{lem:hypercontractive-concentration}
Any degree-$d$ polynomial $f(A_1, \dots, A_n)$ of independent centered Gaussian random variables $A_1, \dots, A_n$ satisfies
\[ \Pr \Brac{ \left|f(A) - \mathbb{E}[f(A)] \right| > \tau } \leq O(1) \cdot  e^{-\left(\frac{\tau^2}{R \textbf{Var}(f(A))}\right)^{1/d}} \;,\] 
where $R$ is a universal constant. 
\end{lemma}
We start by showing that the singular values of $\overline{\Sigma_a}^{-1/2} \Sigma_a^{1/2}$ are close to $1$, 
and $\norm{\overline{\Sigma}_a(\mu-a - \overline{\mu}_a)}$ is close to $0$. 
	
For the second claim, note that $ \mu_a - \overline{\mu}_a =  \frac{1}{|S_a|} \sum_{j \in S_a} (X_j - \mu_a) $. 
Applying \cref{lem:hypercontractive-concentration} with $f_l(X_1, \dots, X_n) := \frac{1}{|S_a|} \sum_{j \in S_a} (X_j - \mu_a)_l$ gives us 
\[ \Pr \Brac{ \left|(\mu_a - \overline{\mu}_a)_l \right| > \tau (\Sigma_a)_{l,l}^{1/2} } \leq O(1) \cdot  e^{-\left(O(1)  \cdot  (n/k) \cdot \tau^2 \right)} \;.\]
A union bound over each index then gives 
\[ 
\Pr \Brac{ \| \mu_a - \overline{\mu}_a\|_2^2 > \tau^2 \sum_l (\Sigma_a)_{l,l} }\leq O(1) \cdot \sum_l e^{-\left(O(1)  \cdot  (n/k) \cdot \tau^2\right)} \;. 
\]
Now observe that for a unit vector $v$, we have 
\[\| \Sigma_a^{1/2} v\|_2^2 = \sum_{i,j} v_i v_j (\Sigma_a)_{i,j} \leq\sum_{i,j} v_i v_j \sqrt{ (\Sigma_a)_{i,i} \cdot (\Sigma_a)_{j,j}} \leq \Paren{ \sum_i v_i \sqrt{(\Sigma_a)_{i,i}} }^2 \leq \|v\|_2^2 \sum_i (\Sigma_a)_{i,i} \;.
\] 
Scaling $\tau$ again by $1/(\log(d))$, we see for some of unit norm vector $v$,  
\[ 
\Pr \Brac{ \| \mu_a - \overline{\mu}_a\|_2^2 > O(\tau^2) \| \Sigma_a^{1/2} v \|_2^2 } \leq O(1) \cdot  \sum_l e^{-\left(O(1)  \cdot  (n/k) \cdot \tau^2/\log(d)\right)} \leq O(1) \cdot \exp(-O(n/k) \tau^2) \;. 
\]
Thus, with probability $1-\exp(-O(n/k) \tau^2)$, we have 
\[ 
\| \mu_a - \overline{\mu}_a\|_2^2 \leq  \| \Sigma^{1/2}_a v \|_2^2 \;, 
\] 
for some $v$ with norm at most $\tau$. Since $\Sigma_a$ is full rank, we can assume that 
$\mu_a - \overline{\mu}_a = \Sigma_a^{1/2} v$, for some $v$ whose norm is bounded by $\tau$.

For the first claim, observe that by standard matrix concentration results (see, e.g.,~\cite{vershynin2010introduction}) 
we have that $\widehat{\Sigma}_a := \frac{1}{|S_a|} \sum_{j \in S_a} (X_j - \mu_a)^{\otimes 2}$ satisfies 
$(1-\tau) \Sigma_a\preccurlyeq \widehat{\Sigma}_a \preccurlyeq (1+\tau) \Sigma_a$ with probability 
$1-\exp(-O(1) \cdot (n/k) \cdot \tau^2/d)$. 
However, in our case, we have $\overline{\Sigma}_a = \frac 1 {|S_a|} \sum_{i \in S_a} (X_i - \overline{\mu}_a)(X_i - \overline{\mu}_a)^\top$. 
Adding and subtracting $\mu_a$, a small calculation gives us
\begin{align*}
\overline{\Sigma}_a &= \frac 1 {|S_a|} \sum_{i \in S_a} (X_i - \mu_a + \mu_a - \overline{\mu}_a)^{\otimes 2} \\
&= \frac 1 {|S_a|} \sum_{i \in S_a} (X_i - \mu_a)^{\otimes 2} + 3 (\mu_a - \overline{\mu}_a)^{\otimes 2} \;.
\end{align*} 
Hence, we see $\overline{\Sigma}_a = \widehat{\Sigma}_a \pm 3 (\mu_a - \overline{\mu}_a)^{\otimes 2}$. 
This means for any $v$, and some $u$ having norm $\leq \tau$
\begin{align*}
	\iprod{ v, \overline{\Sigma}_a v} &= \iprod{ v, \widehat{\Sigma}_a v} + 3 \iprod{v, \mu_a - \overline{\mu}_a}^2 \\
	& = \iprod{ v, \widehat{\Sigma}_a v} + 3 \iprod{v,\Sigma_a^{1/2} u}^2 \\
	& = \iprod{ v, \widehat{\Sigma}_a v} + 3 \iprod{\Sigma_a^{1/2} v, u}^2 \\
	& \leq  \| \Sigma_a^{1/2} v \|_2^2  + 3 \| \Sigma_a^{1/2}v\|^2_2 \| u\|_2^2 \\
	&\leq  \| \Sigma_a^{1/2} v \|_2^2  \Paren{1 + 3 \tau} = \iprod{v, \Sigma_a v} (1+O(\tau)) \;.
\end{align*}
A similar argument holds to give us a lower bound, ultimately showing that 
\begin{align}\label{eqn:covar_concentration}
(1-\tau) \Sigma_a\preccurlyeq
\overline{\Sigma}_a \preccurlyeq (1+\tau) \Sigma_a \;.
\end{align}
This in turn implies that the singular values of $\overline{\Sigma}_a^{-1/2} \Sigma_a \overline{\Sigma}_a^{-1/2}$ are $1 \pm \tau$, 
and so the singular values of $\overline{\Sigma}_a^{-1/2} \Sigma_a^{1/2}$ are $1 \pm \sqrt{\tau}$. 
Hence, we get 
\[ 
\norm{\overline{\Sigma}_a^{-1/2} (\mu_a - \overline{\mu}_a)}^2 = \norm{\overline{\Sigma}_a^{-1/2} \Sigma_a^{1/2} v}^2 \leq O(\tau^2) \;.
\]

Finally, we see that $\overline{\Sigma}_a^{-1/2} \Sigma_a^{1/2} Y_j +  \overline{\Sigma}_a^{-1/2} (\mu_a - \overline{\mu}_a) = Y_j \pm \tau u$ 
for some unit vector $u$, and so 
\[
[\overline{\Sigma}_a^{-1/2} \Sigma_a^{1/2} Y_j +  \overline{\Sigma}_a^{-1/2} (\mu_a - \overline{\mu}_a)]^{\tensor s} = 
[Y_j \pm O(\tau) u]^{\tensor s} = Y_j^{\otimes s} \pm D_j \;, 
\]
where $\| D_j \|_F^2 \leq \exp(O(s)) \cdot \tau \cdot \|Y_j\|_2^{2(s-1)}$. Therefore,
\begin{align*}
\Norm{\frac 1 {|S_a|} \sum_{j \in S_a} [\overline{\Sigma}_a^{-1/2} (X_j - \overline{\mu}_a)]^{\tensor s} - \frac{1}{|S_a|} \sum_{j \in S_a} Y_j^{\otimes s}}_F^2 
&\leq \exp(O(s)) \cdot \tau \cdot \frac{1}{|S_a|} \sum_{j \in S_a} \|Y_j\|_2^{2(s-1)} \\ 
&\leq \exp(O(s \log(s))) \cdot \tau  \;,
	\end{align*}
where the final inequality follows from the fact that the $2(s-1)$-th moment of a Gaussian random variable is at most $\exp(O(s \log(s)))$, and from concentration similar to what we show below.
	
We now bound the Frobenius norm difference between the empirical and population expectations of $Y_j^{\otimes s}$. 
Applying \cref{lem:hypercontractive-concentration} entrywise, we see for entry $ T \in [d]^s$ we get
\[ 
\Pr\Brac{ \left( \Paren{\frac{1}{|S_a|} \sum_{j \in S_a} Y_j^{\otimes s}}_T - \E[Y^{\otimes s}]_T \right)^2 \geq \tau^2 } \leq O(1) \cdot  e^{-\left( O(n/k) \cdot  \frac{\tau^2}{ \exp(O(s))  }\right)^{1/s}}  \;. \] 
A union bound over all entries gives us 
\[ 
\Pr\Brac{ \sum_{T \in [d]^s} \left( \Paren{\frac{1}{|S_a|} \sum_{j \in S_a} Y_j^{\otimes s}}_T - \E[Y^{\otimes s}]_T \right)^2 \geq d^s \tau^2 } \leq O(1) \cdot  e^{-\left( O(n/k) \cdot  \frac{\tau^2}{ \exp(O(s))  }\right)^{1/s}} \;. \]
Finally, rescaling $\tau$, a union bound over all entries gives us 
\[ \Pr\Brac{ \sum_{T \in [d]^s} \left( \Paren{\frac{1}{|S_a|} \sum_{j \in S_a} Y_j^{\otimes s}}_T - \E[Y^{\otimes s}]_T \right)^2 \geq d^{-2t}\delta } \leq O(1) \cdot  e^{-\left( O(n/k) \cdot  \frac{\delta^2}{ \exp(O(s))  d^{4t}}\right)^{1/t}} \;. \]
Hence, overall we see the probability of $E_j(a)$ for any fixed $a \in [k]$ is 
$1 - O(k) \cdot \exp(-O(1) \cdot \frac{n}{k} \cdot \frac{\delta^2}{\exp(O(t)) d^{4t}})^{1/t} - \frac{O(1)k}{n}$. 
If $n \geq \poly(k) \cdot \log(1/\gamma)^{O(t)} \cdot  \frac{(2d)^{O(t)}}{\delta^2}$, 
then the deterministic conditions hold with probability $1-\gamma$. 
This completes the proof.	
\end{proof}

\section{Omitted Proofs from Section~\ref{sec:tools}} \label[section]{sec:proofs-tools-section}

\subsection{Proof of Lemmas~\ref{lem:emp-pop} and~\ref{lem:emp-pop-2}}

We give the proof of \cref{lem:emp-pop}.
The proof of \cref{lem:emp-pop-2} is almost identical.

\begin{proof}[Proof of \cref{lem:emp-pop}]
  We start with the first statement, \cref{eq:emp-to-pop-1}.
  By the axioms of $\cA$, we know that for every homogeneous degree-$s$ polynomial $q$,
  \begin{equation}\label{eq:emp-to-pop-2}
 \cA \proves{2t} \Paren{\Iprod{\frac 2 n \sum_{i \in [n]} w_i [\Sigma^{-1/2} (X_i' -\mu(w))]^{\tensor s}, q} - \E_{g \sim \cN(0,\Id)} q(g)}^2 \leq \delta d^{-2t} \cdot \sum_{\alpha \subseteq [n], |\alpha| \leq s} \widehat{q}(\alpha)^2\mper
  \end{equation}
  (By abuse of notation, $q$ also denotes the $s$-tensor of $q$'s coefficients.)

  Set $q(X) = p(\Sigma^{1/2}X + \mu(w))$ (which is not necessarily homogeneous), so that
  \[
  \cA \proves{2t} \E_{g \sim \cN(0,\Id)} q(g) = \E_{g \sim \cN(\mu(w),\Sigma(w))} p(g)
  \]
  and
  \[
  \cA \proves{2t} w_i p(X_i') = w_i q(\Sigma^{-1/2}(X_i' - \mu(w)))
  \]
  (here we used the axioms $\Sigma^{1/2} \Sigma^{-1/2} w_i (X_i' - \mu(w)) = w_i (X_i' - \mu(w))$).
  So it is sufficient to bound
  \begin{equation*}
  \Paren{\Iprod{\frac 2 n \sum_{i \in [n]} w_i [\Sigma^{-1/2} (X_i' -\mu(w))]^{\tensor s}, q} - \E_{g \sim \cN(0,\Id)} q(g)}^2\mper
  \end{equation*}

  Split $q$ into homogeneous parts (with respect to the $x$ variables) $q = q_0 + q_1 + \ldots + q_{t/2}$.
  Then by Cauchy-Schwarz it will be enough to bound each of
\begin{equation*}
  \Paren{\Iprod{\frac 2 n \sum_{i \in [n]} w_i [\Sigma^{-1/2} (X_i' -\mu(w))]^{\tensor s}, q_j} - \E_{g \sim \cN(0,\Id)} q_j(g)}^2
  \end{equation*}
  separately, losing a factor of $2^t$.
  In each case, by the moment axioms of $\cA$,
  \begin{equation}\label{eq:emp-to-pop-2}
 \cA \proves{2t} \Paren{\Iprod{\frac 2 n \sum_{i \in [n]} w_i [\Sigma^{-1/2} (X_i' -\mu(w))]^{\tensor s}, q_j} - \E_{g \sim \cN(0,\Id)} q_j(g)}^2 \leq \delta d^{-2t} \cdot \sum_{\alpha \subseteq [n], |\alpha| = j} \widehat{q}(\alpha)^2\mper
  \end{equation}
  Now we are done by observing that $\sum_{\alpha \subseteq [n], |\alpha| \leq t/2} \widehat{q}(\alpha)^2 \preceq d^{2t} t^{O(t)} \E_{g \sim \cN(0,\Id)} q(g)^2 = \E_{g \sim \cN(\mu(w),\Sigma(w))} p(g)^2$.
  \Snote{TODO: expand justification here; this uses equivalence up to $O_t(1)$ factors of monomial and Hermite bases.}
  \SKnote{NOTE: Perhaps remove $t^t$ as well}

  Next we prove \cref{eq:emp-to-pop-diff-1} as a corollary of \cref{eq:emp-to-pop-1}.
	For each $i \in [n]$, we apply ~\cref{eq:emp-to-pop-1} to $p(X_i' - X_j')$. For each $i$, we get 
	\begin{equation*}
	\cA \proves{O(t)} \Paren{\frac k n \sum_{j=1}^n w_j p(X_i' - X_j')
		- \E_{g \sim \cN(\mu(w),\Sigma(w))} p(X_i'-g)}^2 \leq  t^{O(t)} \cdot \delta \cdot \Paren{ \E_{g \sim \cN(\mu(w),\Sigma(w))} p(X_i'-g)^2 }\mper
	\end{equation*}
	Similarly, applying \cref{eq:emp-to-pop-1} once again to $\E_{g \sim \cN(\mu(w),\Sigma(w))} p(X_i'-g)$, we get 
	\begin{align*}
	&\cA \proves{O(t)} \Paren{\frac k n \sum_{i=1}^n w_i \E_{g \sim \cN(\mu(w),\Sigma(w))} p(X_i'-g)
	- \E_{g', g \sim \cN(\mu(w),\Sigma(w))} p(g'-g)}^2 \\
	&\leq  (2t)^{O(t)} \cdot \delta  \cdot  \E_{g' \sim \cN(\mu(w),\Sigma(w))} \Paren{\E_{g \sim \cN(\mu(w),\Sigma(w))} p(g'-g)}^2 \\
&\leq  (2t)^{O(t)} \cdot \delta  \cdot  \E_{g' \sim \cN(\mu(w),\Sigma(w))} \Paren{\E_{g \sim \cN(\mu(w),\Sigma(w))} p(g'-g)^2}	 \mcom
\end{align*}
where the final inequality follows from applying Cauchy-Schwarz to the inner expectation.
To finish the proof of \eqref{eq:emp-to-pop-diff-1}, we use the triangle inequality by adding and subtracting $\frac k n \sum_{i=1}^n w_i \E_{g \sim \cN(\mu(w),\Sigma(w))} p(X_i'-g)$ and using the inequalities above. By a calculation, it can be verified that
\begin{align*} 
\cA &\proves{O(t)} \Paren{\Paren{\frac k n}^2 \sum_{i, j=1}^n w_i w_j p(X_i' - X_j')	- \E_{g \sim \cN(0,2\Sigma(w))} p(g)}^2 \\
&\leq (2t)^{O(t)} \cdot \delta^2  \cdot \frac k n \sum_{i=1}^n w_i \E_{g \sim \cN(\mu(w),\Sigma(w))} p^2(X_i'-g) +  \E_{g', g \sim \cN(\mu(w),\Sigma(w))} p(g'-g)^2  \\
& = (2t)^{O(t)} \cdot \delta^2 \cdot \E_{g \sim \cN(0,2\Sigma(w))} p(g)^2 \mper
\end{align*} 
\SKnote{More justificantion -- make the triangle inequality and hypercontractivity argument clearer.}
Finally, if $p$ is a square, we can bound the RHS from above by
\begin{align*}
(2t)^{O(t)} \cdot \delta^2 \cdot \Paren{\E_{g \sim \cN(0,2\Sigma(w))} p(g)}^4
\end{align*}
via \cref{fact:hypercontractivity}. This completes the proof of \cref{lem:emp-pop}. 
\end{proof}

\subsection{Proof of \cref{lem:mom_decorruption_diff}}

We turn now to the proof of \cref{lem:mom_decorruption_diff}, which will use three claims, all proved below.

The first claim establishes bounds on the difference between the empirical distribution of a polynomial on the uncorrupted samples $X_1,\ldots,X_n$ and the ghost samples $X_1',\ldots,X_n'$.
For technical convenience, the claim allows for polynomials in multiple samples (the case $u > 1$ below).

\begin{claim}[Glorified Cauchy-Schwarz]
\label[claim]{clm:mom_decorruption_diff_1}
Let $f: (\R^{d})^u \rightarrow \R$ be a symbolic polynomial of degree-$t/4$.
For all $u \in \N$ and $S_1,\ldots,S_u \subseteq [n]$ each having size $n/k$, we have 
\begin{align*}
\cA \proves{O(t)} & \Paren{\frac{k^u}{n^u} \sum_{j_1 \in S_1,\ldots,j_u \in S_u} (\Pi_{i \leq u} w_{j_i}) (f(Y_{j_1},\ldots,Y_{j_u}) - f(X'_{j_1},\ldots,X'_{j_u}))}^2 \\
& \leq  O_u(k^{2u}) \cdot \eps^2 \cdot \frac{k^u}{n^u} \sum_{j_1 \in S_1,\ldots,j_u \in S_u} (\Pi_{i \leq u} w_{j_i}) f(Y_{j_1},\ldots,Y_{j_u})^2 \\
& + O_u(k^{2u}) \cdot \eps^2 \cdot \frac{k^u}{n^u} \sum_{j_1 \in S_1,\ldots,j_u \in S_u} (\Pi_{i \leq u} w_{j_i}) f(X'_{j_1},\ldots,X'_{j_u})^2  \mper
\end{align*}
\end{claim}

The second claim bounds the terms on the right-hand side of the conclusion of \cref{clm:mom_decorruption_diff_1} in the case that $f(x,y) := \iprod{x-y,v}^{s}$ for some $v \in \R^d$.

\begin{claim}
\label[claim]{clm:mom_decorruption_diff_2}
  For all $v \in \R^d$, $a,b \in [k]$, and $s \leq t/10$,
	\begin{align}
\cA &\proves{t} \Paren{\frac{k^2}{n^2} \sum_{i \in S_a,j \in S_b} w_i w_j \iprod{Y_i - Y_j, v}^{2s}}^2 \leq k^4 \cdot (2s)^{O(s)} \cdot \Paren{\iprod{v, \Sigma_a v}^{2s} + \iprod{v, \Sigma_b v}^{2s} + \iprod{ \mu_a - \mu_b, v}^{4s}} \mper \label{eq:mom_decorruption_1} \\
\cA & \proves{t} \Paren{\frac{k^2}{n^2} \sum_{i \in S_a,j \in S_b} w_i w_j \iprod{X_i' - X_j', v}^{2s}}^2 \leq k^4 \cdot (2s)^{O(s)} \cdot \iprod{v, \Sigma(w) v}^{2s}\mper \label{eq:mom_decorruption_2}
\end{align}
\end{claim}

The last claim bounds the right-hand side of \cref{clm:mom_decorruption_diff_1} in the case that $f(x,y,z,w) = \Paren{ \iprod{x-y, A(x-y)} - \iprod{z-w, A(z-w)} }^2$, for some $A \in \R^{d \times d}$.

\begin{claim}
\label[claim]{clm:mom_decorruption_diff_3}
  For all $A \in \R^{d \times d}$, $a,b \in [k]$, and $s \leq t/10$, if $g(x,y,v,z) = \Paren{ \iprod{x-y, A(x-y)} - \iprod{v-z, A(v-z)} }^2$, we have that
\begin{align}
\cA &\proves{t} \Paren{\frac{k^4}{n^4} \sum_{(i,j) \in S_a,(r,l) \in S_b} w_i w_j w_r w_l g(Y_i,Y_j,Y_r,Y_l)^2 }^2\\
&\leq  k^8 \cdot O\Paren{  \|\Sigma_a^{1/2} A \Sigma_a^{1/2}\|_F^8  + \|\Sigma_b^{1/2} A \Sigma_b^{1/2}\|_F^8 + \iprod{\Sigma_a - \Sigma_b, A}^8 }. \label{eq:mom_decorruption_3} \\
\cA & \proves{t} \Paren{\frac{k^4}{n^4} \sum_{(i,j) \in S_a,(r,l) \in S_b} w_i w_j w_r w_l g(X_i',X_j',X_r',X_l')^2 }^2 \leq k^8 \cdot O(\|\Sigma^{1/2} A \Sigma^{1/2}\|_F^8). \label{eq:mom_decorruption_4}
\end{align}
\end{claim}

\begin{proof}[Proof of \cref{lem:mom_decorruption_diff}]
  Items (1) and (2) follow immediately from \cref{clm:mom_decorruption_diff_1,clm:mom_decorruption_diff_2,clm:mom_decorruption_diff_3}.
  For item (3), we use a similar argument as in the proof of \cref{clm:mom_decorruption_diff_1},
  presented here for completeness.
  Following the same logic as in \cref{clm:mom_decorruption_diff_1}, we can write 
\begin{align*}
& \Paren{ \Tr(B^T \Sigma'_S B) - \Tr(B^T \Sigma_S B) }^2 \\
&= \Brac{ \frac{1}{n^2} \sum_i w_i w_j \Paren{ \Tr(B^T (X_i' - X_j')(X_i'-X_j')^T B) -\Tr(B^T (Y_i - Y_j)(Y_i-Y_j)^T B) }  }^2\\
&\leq \frac{O(\eps k^2)}{n^2}\Brac{\sum_i w_i w_j ( \|B(X_i' -X_j')\|^4 + \|B(Y_i - Y_j)\|^4) } \mper
\end{align*}
Applying \cref{lemma:hypercontractivity-sos-polys} and \cref{lem:emp-pop,lem:emp-pop-2} finishes the argument.
\end{proof}

Now we turn to the proofs of \cref{clm:mom_decorruption_diff_1,clm:mom_decorruption_diff_2,clm:mom_decorruption_diff_3}.

\begin{proof}[Proof of \cref{clm:mom_decorruption_diff_1}]
For $T \in S_1 \times \ldots \times S_u$, let $Y_T = (Y_{i_1},\ldots,Y_{i_u})$, where $T = (i_1,\ldots,i_u)$ and similarly for $X_T$.
Let $z'(T) = \prod_{i \in T}z_i \mathbbm{1}_{Y_i = X_i}$ where $z_i$ are the indeterminates from $\cA$.
Then $\cA \proves{t} z'(T) (f(Y_T) - f(X'_T)) = 0$, since
	\begin{align*} 
	& z'(T) \Paren{ f(Y_T) - f(X'_T) }\\
        & \qquad = z'(T) \Paren{f(Y_T) - f(X_T) + f(X_T) - f(X'_T) } \\
	& \qquad = \prod_{i \in T}z_i\mathbbm{1}_{Y_i = X_i} \Paren{f(Y_T) - f(X_T) } + \prod_{i \in T}z_i \mathbbm{1}_{Y_i = X_i} \Paren{ f(X_T) - f(X'_T) }\\
	& \qquad =0\mcom
	\end{align*}
	where the final inequality follows from  $\Pi_{i \in T}\mathbbm{1}_{Y_i = X_i}(f(X_T) - f(Y_T)) = 0$ and the axioms $\cA_{\text{corruptions}}$.
        Hence,
        \begin{align*}
        \cA \proves{O(t)} & \Paren{\frac{k^u}{n^u} \sum_{T \in S_1 \times \ldots \times S_u} (\prod_{\ell \in T} w_\ell) (f(Y_T) - f(X'_T))}^2 \\
        & = \Paren{\frac{k^u}{n^u} \sum_{T \in S_1 \times \ldots \times S_u} (1-z'(T)) (\prod_{\ell \in T} w_\ell) (f(Y_T) - f(X'_T))}^2\mper
        \end{align*}
        By Cauchy-Schwarz and the axioms $w_i^2 = w_i$, we obtain
        \begin{align}
        \cA \proves{O(t)} & \Paren{\frac{k^u}{n^u} \sum_{T \in S_1 \times \ldots \times S_u} (\prod_{\ell \in T} w_\ell) (f(Y_T) - f(X'_T))}^2 \\
        & \leq 8 \Paren{\frac{k^u}{n^u}\sum_{T \in S_1 \times \ldots \times S_u} (1-z'(T))^2} \Paren{\frac{k^u}{n^u} \sum_{T \in S_1 \times \ldots \times S_u} (\prod_{\ell \in T} w_\ell) (f(Y_T)^2 + f(X'_T)^2)}\mper \label{eq:decorrupt-1}
        \end{align}
	Since $\cA \proves{2} \sum_i z_i = (1-\eps) n$ and $\sum_{i} \mathbbm{1}_{X_i = Y_i} \geq (1-\eps) n$, we see that
        \[
        \cA \proves{2} \sum_{i} ( 1- z_i \mathbbm{1}_{X_i = Y_i}) = \sum_i z_i (1- \mathbbm{1}_{X_i = Y_i}) + (1-z_i) \leq \sum_i (1- \mathbbm{1}_{X_i = Y_i}) + \sum_i (1-z_i) \leq 2\eps n\mper
        \]
        Rearranging, we see $\cA \proves{2} \sum z_i \mathbbm{1}_{X_i = Y_i} \geq (1-2\eps) n$.
        Repeating this for $S_1,\ldots,S_u$, we get
	\begin{align*}
	\cA \proves{t} \sum_{T \in S_1 \times \ldots \times S_u} z'(T) = \sum_{T \in S_1 \times \ldots \times S_u} \prod_{i \in T}z_i\mathbbm{1}_{Y_i = X_i} = \prod_{i \leq u} \Paren{ \sum_{j \in S_i} z_j \mathbbm{1}_{Y_j = X_j} } \geq (1/k-2\eps)^u n^u
	\end{align*}
	and so 
	\begin{align*}
	\sum_{T \in S_1 \times \ldots \times S_u} 1-z'(T) \leq (2u\eps) n^u \mper
	\end{align*}
	
        Combining with \eqref{eq:decorrupt-1}, we obtain
\begin{align*}
        \cA \proves{O(t)} & \Paren{\frac{k^u}{n^u} \sum_{T \in S_1 \times \ldots \times S_u} (\prod_{\ell \in T} w_\ell) (f(Y_T) - f(X'_T))}^2 \\
        & \leq  8 k^{2u} (2u\eps)^2 \Paren{\frac{k^u}{n^u} \sum_{T \in S_1 \times \ldots \times S_u} (\prod_{\ell \in T} w_\ell) (f(Y_T)^2 + f(X'_T)^2)}\mper
        \end{align*}
\end{proof}

\begin{proof}[Proof of \cref{clm:mom_decorruption_diff_2}]
We start with \eqref{eq:mom_decorruption_1}.
Observe that $\proves{}(a+b+c)^{2s} < \exp( O(s)) \cdot \Paren{a^{2s} + b^{2s} + c^{2s}}$ implies that
	\begin{align*}
	\iprod{Y_i - Y_j, v}^{2s} \leq 2^{O(s)} \cdot (\iprod{Y_i-\mu_a, v}^{2s} + \iprod{\mu_a- \mu_b, v}^{2s} + \iprod{Y_j- \mu_b, v}^{2s}) \; .
	\end{align*}
         So, letting $f(x) = \iprod{x,v}^s$, we have
	\begin{align*}
	\cA \proves{t} &\Paren{ \frac{k^2}{n^2} \sum_{i \in S_a, j \in S_b} w_i w_j f^2(Y_i- Y_j) }^2 \\
	&\leq \frac{k^4\exp(O(s))}{n^4} \cdot \Paren{ \alpha_a(w) \sum_{i \in S_b} w_i f^2(Y_i- \mu_b) + \alpha_b(w) \sum_{i \in S_a} w_i f^2(Y_i- \mu_a)  +  \alpha_a(w) \alpha_b(w) f^2(\mu_a- \mu_b) }^2\\
	&\leq \exp(O(s)) \cdot \Paren{  \frac{k^2 \alpha_a(w)^2}{n^2} \Paren{ \frac{k}{n}  \sum_{i \in S_b}   f^2(Y_i- \mu_b)}^2 +  \frac{k^2\alpha_b(w)^2}{n^2} \Paren{ \frac{k}{n} \sum_{i \in S_a}   f^2(Y_i- \mu_a)}^2 } \\
	&\qquad  +  \exp(O(s)) \cdot \frac{k^2\alpha_a(w)^2}{n^2} \frac{k^2\alpha_b(w)^2}{n^2} f^4(\mu_a- \mu_b) \; .
	\end{align*}
	Since $f$ is a degree-$\leq t/4$ polynomial, we use \cref{lem:emp-pop-2,fact:hypercontractivity}. Since $\delta \leq 0.1$ and $\alpha_a(w)^2,\alpha_b(w)^2 \leq n^2$, we get 
	\begin{align*}
	\cA \proves{t} &\Paren{ \frac{k^2}{n^2} \sum_{i \in S_a, j \in S_b} w_i w_j f^2(Y_i - Y_j) }^2 \\
	&\leq (2s)^{O(s)} \cdot \Paren{ \E_{x \sim \cN(0, 2\Sigma_a)}[f(x)^2]^2 +  \E_{x \sim \cN(0, 2\Sigma_b)}[f(x)^2]^2 +  f^4(\mu_a - \mu_b) } \; . 
	\end{align*}
	Since $f(x)^2 = \iprod{x, v}^{2s}$, we see
	\begin{align*}
	 \E_{x \sim \cN(0, 2\Sigma_a)}[f(x)^2]^2 \leq s^{O(s)} \cdot \iprod{v, \Sigma v}^{2s}
	\end{align*}
        and similarly for $\E_{x \sim \cN(0, 2\Sigma_b)}[f(x)^2]^2$, which finishes the proof of \eqref{eq:mom_decorruption_1}.

        Moving on to \eqref{eq:mom_decorruption_2}, the reasoning is simpler.
        First of all, since $\iprod{X_i' - X_j',v}^{2s}$ is a square,
        \[
        \cA \proves{t} \Paren{\frac{k^2}{n^2} \sum_{i \in S_a,j \in S_b} w_i w_j \iprod{X_i' - X_j', v}^{2s}}^2
        \leq \Paren{\frac{k^2}{n^2} \sum_{i,j \in [n] } w_i w_j \iprod{X_i' - X_j', v}^{2s}}^2\mper
        \]
        Now using \cref{lem:emp-pop}, we can bound the right-hand side above as
        \[
        \cA \proves{t} \Paren{\frac{k^2}{n^2} \sum_{i,j \in [n] } w_i w_j \iprod{X_i' - X_j', v}^{2s}}^2 \leq k^4 \cdot (2s)^{O(s)} \cdot \Paren{\E_{g \sim \cN(0,2\Sigma(w))} \iprod{g,v}^{2s}}^2\mper
        \]
        \cref{fact:hypercontractivity} concludes the proof of \eqref{eq:mom_decorruption_2}.
        \end{proof}
       
\begin{proof}[Proof of \cref{clm:mom_decorruption_diff_3}]
        The argument is similar to the proof of \cref{clm:mom_decorruption_diff_2}.
        We start with \eqref{eq:mom_decorruption_3}.
	\begin{align*}
	 \cA & \proves{t} g(Y_i,Y_j,Y_r,Y_l)^2 = \iprod{(Y_i - Y_j)^{\otimes 2} -(Y_r - Y_l)^{\otimes 2}, A}^{4} \\
         & = \iprod{(Y_i-Y_j)^{\otimes 2} - \Sigma_a + \Sigma_a - \Sigma_b + \Sigma_b - (Y_r - Y_l)^{\otimes 2}, A}^4 \\
	 &\leq O(\iprod{(Y_i - Y_j)^{\otimes 2}-\Sigma_a, A}^{4} + \iprod{\Sigma_a - \Sigma_b, A}^{4} + \iprod{\Sigma_b-(Y_r - Y_l)^{\otimes 2}, A}^{4})\mper
	\end{align*}
        Proceeding as in the proof of \cref{clm:mom_decorruption_diff_2},
	\begin{align*}
	\cA \proves{t} &\Paren{ \frac{k^4}{n^4} \sum_{(i,j) \in S_a^2, (r,l) \in S_b^2} w_i w_j w_r w_l g(Y_i,Y_j,Y_r,Y_l)^2 }^2 \\
	&\leq O\Paren{\frac{k^8}{n^8}} \cdot \Paren{ \alpha_a(w)^2 \sum_{(r,l) \in S_b^2} w_r w_l \iprod{(Y_r - Y_l)^{\otimes 2} - \Sigma_b, A}^4 }^2 \\
	&\qquad  + O\Paren{\frac{k^8}{n^8}} \cdot \Paren{ \alpha_b(w)^2 \sum_{(i,j) \in S_a^2} w_iw_j \iprod{(Y_i - Y_j)^{\otimes 2} - \Sigma_a, A}^4  }^2 \\
	&\qquad +  O\Paren{\frac{k^8}{n^8}} \cdot \alpha_a(w)^4 \alpha_b(w)^4 \cdot \iprod{\Sigma_a - \Sigma_b, A}^8 \; . \\
	\end{align*}
        To finish the proof of \eqref{eq:mom_decorruption_3}, we just need to bound $\Paren{\frac{k^2}{n^2} \sum_{(i,j) \in S_a^2} w_i w_j \iprod{(Y_i - Y_j)^{\otimes 2} - \Sigma_a, A}^4}^2$. The analogous argument for the $S_b$ term above will be identical.
        By \cref{lem:emp-pop-2,fact:hypercontractivity},
        \begin{align*}
        \cA \proves{O(1)} & \Paren{\frac{k^2}{n^2} \sum_{(i,j) \in S_a^2} w_i w_j \iprod{(Y_i - Y_j)^{\otimes 2} - \Sigma_a, A}^4}^2\\
        & \leq O \Paren{ \E_{g \sim \cN(0,2\Sigma_a)}\Paren{ \iprod{gg^\top,A} - \E_{g \sim \cN(0,2 \Sigma_a)} \iprod{gg^\top, A}}^2 }^4 \\
        & = O \Paren{ \|\Sigma_a^{1/2} A \Sigma_a^{1/2} \|_F^8 }\mcom
        \end{align*}
        where the last step follows from the identity $\E_{x \sim \cN(0,I)} (\iprod{x,Bx} - \E_{x \sim \cN(0,I)} \iprod{x,Bx})^2 = \|A\|_F^2$ for any matrix $B$.

        Finally, \eqref{eq:mom_decorruption_4} follows from the same reasoning as in the proof of \eqref{eq:mom_decorruption_2} from \cref{clm:mom_decorruption_diff_2}, again using the identity in variables $\Sigma, \Sigma^{1/2}$ that $\E_{g \sim \cN(0,\Sigma)} (\iprod{gg^\top,B} - \E_{g \sim \cN(0,\Sigma)} \iprod{gg^\top,B})^2 = \|\Sigma^{1/2} B \Sigma^{1/2}\|_F$ if $(\Sigma^{1/2})^2 = \Sigma$.
\end{proof}

Lastly, we can prove \cref{lem:mixture_var} using the same ideas:

\begin{proof}[Proof of \cref{lem:mixture_var}]
  We can expand $2\iprod{v, \Sigma(w) v} = \sum_{i,j \in [n]} w_i w_j \iprod{X_i' - X_j',v}^2$.
  Then the lemma follows by the same reasoning as in \cref{lem:mom_decorruption_diff}, case (1).
\end{proof}

\subsection{Certifiable Anti-concentration: Proof of \cref{lem:q-main}} \label[section]{sec:anticoncentration}

In this section, we prove \cref{lem:q-main}.
As in \cite{raghavendra2020list,karmalkar2019list}, we prove the lemma using tools from approximation theory.
For completeness we provide a proof from basic tools in approximation theory, following ideas from \cite{DGJ+:10}.

The starting point is the following elementary fact from approximation theory.

\begin{theorem}[Jackson's Theorem] \label[theorem]{thm:jacksons}
For every bounded continuous function $f \, : \, [-1,1] \rightarrow \R$ and every integer $\ell \geq 1$, 
there is a degree-$\ell$ polynomial $J$ such that
  \[
  \sup_{x \in [-1,1]} |J(x) - f(x)| \leq 6 \cdot \sup_{x,y \in [-1,1]} \{ |f(x) - f(y)| \, : \, |x-y| \leq \delta \}\mper
  \]
\end{theorem}

We will also need the following amplifying polynomial, as in \cite{DGJ+:10} -- the following fact follows from the Chernoff bound.

\begin{fact}[See \cite{DGJ+:10}, Claim 4.3]
  \label[fact]{fact:amplify}
  Let $k \geq 0$ be an even integer and let
  \[
  A_k(u) = \sum_{j \geq k/2} {k \choose j} \Paren{\frac{1+u}{2}}^j \Paren{\frac{1-u}{2}}^{k-j} \, .
  \]
  The degree-$k$ polynomial $A_k$ satisfies
  \begin{enumerate}
    \item if $u \in [3/5,1]$ then $A_k(u) \in [1-e^{-k/6}, 1]$, and
    \item if $u \in [-1,-3/5]$ then $A_k(u) \in [0,e^{-k/6}]$,
    \item if $u \in [-1,1]$ then $A_k(u) \in [0,1]$.
  \end{enumerate}
\end{fact}

We also need a bound on the rate of growth of univariate polynomials.

\begin{fact}[\cite{rivlin1974chebyshev}]
\label[fact]{fact:growth-bound}
  Let $a(t)$ be a polynomial of degree at most $d$ such that $|a(t)| \leq 1$ for all $t \in [-1,1]$.
  Then for all $t \geq 1$, $|a(t)| \leq |2t|^d$.
\end{fact}

Now we prove \cref{lem:q-main}.

\begin{proof}[Proof of \cref{lem:q-main}]
  Let $\delta > 0$ and let $f(x) \, : \, [-1,1] \rightarrow [0,1]$ be the following piecewise-linear function:
  \[
  f(x) = \begin{cases} 0, & |x| \geq 2\delta \\
                       (x+2\delta)/\delta, & -2\delta \leq x \leq -\delta \\
                       1, & |x|, \leq \delta\\
                       (-x+2\delta)/\delta, & \delta \leq x \leq 2\delta
          \end{cases} \, .
    \]
  The maximum slope of $0.9(f+0.1)$ is $1/\delta$.
  By \cref{thm:jacksons}, there is a degree-$100/\delta$ polynomial $J(x)$ 
  such that $|J(x) - 0.9(f(x)+0.1)| \leq 0.01$, for all $x \in [-1,1]$.
  Thus,
  \begin{enumerate}
  \item $J(x) \in [0,1]$ for all $x \in [-1,1]$.
  \item $J(x) \leq 0.2$ if $x \in [-1,-2\delta] \cup [2\delta,1]$.
  \item $J(x) \geq 0.8$ if $x \in [-\delta,\delta]$.
  \end{enumerate}

  For $k$ to be chosen later, let $A_k$ be as in \cref{fact:amplify}.
  Then the degree $d = 100k/\delta$ polynomial $A_k(J(x))$ satisfies
  \begin{enumerate}
  \item $A_k(J(x)) \in [0,1]$ for all $x \in [-1,1]$.
  \item $A_k(J(x)) \in [0,e^{-k/6}]$ if $x \in [-1,-2\delta] \cup [2\delta,1]$.
  \item $A_k(J(x)) \in [1,1-e^{-k/6}]$ if $|x| \leq \delta$.
  \item $|A_k(J(x))| \leq |2x|^d$ for all $|x| \geq 1$.
  \end{enumerate}
  We note that the last item follows from \cref{fact:growth-bound}.
  By construction, the same is true for $A_k(J(-x))$.

  For $L \geq 0$ to be chosen later, let
  \[
    p_{\delta,k,L}(x) = \Paren{\frac{A_k(J(x/L)) + A_k(J(-x/L))}{2}}^2\mper
  \]
  Then
  \begin{enumerate}
  \item $p_{\delta,k,L}(x) \in [0,1]$ for all $x \in [-L,L]$.
  \item $p_{\delta,k,L}(x) \in [0,e^{-k/6}]$ if $x \in [-L,-2\delta L] \cup [2\delta L, L]$.
  \item $p_{\delta,k,L}(x) \in [1,1-2e^{-k/6}]$ if $|x| \leq \delta L$.
  \item $|p_{\delta,k,L}(x)| \leq |2x/L|^{2d}$ for $|x| \geq L$.
  \item $p_{\delta,k,L}$ is even.
  \end{enumerate}

  Now we choose parameters in the following order.
  Let $k(\epsilon) = 6 \log(1/\e) + 1$ and let $L(\e) = (\log 1/e)^C / \e$ for $C > 0$ a sufficiently large constant.
  Let $\delta(\e) = \e/L(\e)$.
  Let $q_\e(x) = p_{\delta(\e),k(\e),L(\e)}(x)$.
  Then the following holds:
  \begin{enumerate}
  \item $q_\e(x) \in [0,1]$ for all $x \in [-L,L]$.
  \item $q_\e(x) \in [0,\e]$ if $x \in [-L,-2\e] \cup [2\e, L]$.
  \item $q_{\e}(x) \in [1,1-\e]$ if $|x| \leq \e$.
  \item $|q_{\e}(x)| \leq |2x|^{O((\log(1/\e)/\e^2))}$ for $|x| \geq L$.
  \item $q_{\e}$ is even.
  \end{enumerate}
  It just remains to verify that $\E_{x \sim \cN(0,1)} q_\epsilon(x) \leq O(\epsilon)$.
  We break up the expectation as
  \[
  \E_{x \sim \cN(0,1)} q_\e(x) \cdot (1_{|x| \leq \e} + 1_{|x| \in [\e,2\e]} + 1_{|x| \in [2\e,L]} + 1_{|x| \geq L})\mper
  \]
  The first three terms are all at most $O(\e)$ using our bounds above on $q_\e$.
  For the last term we have
  \[
  \E_{x \sim \cN(0,1)} q_\e(x) 1_{|x| \geq L} \leq \Pr(|x| \geq L)^{1/2} \cdot \Paren{ \E_{x \sim \cN(0,1)} (2x)^{O(\log(1/\e)/\e^2)}}^{1/2} \leq e^{-L^2/10} \cdot (\log(1/\e)/\e^2)^{O(\log(1/\e)/\e^2)}\mper
  \]
  Since $L^2 \geq (1/\e^2) (\log 1/\e)^C$ for sufficiently large $C$, this is at most $O(\e)$.
\end{proof}

\section{Mixtures with Arbitrary Weights}\label[section]{sec:arbit-weights}

In our main exposition, we have described an algorithm which learns a mixture of Gaussians $G_1,\ldots,G_k$ \emph{with uniform weights} from corrupted samples.
We sketch here how this algorithm and its analysis may be adapted to learn a mixture of the form $\sum_{i=1}^k \lambda_i G_i$ with (nonnegative) weights $\lambda_1 + \ldots + \lambda_k = 1$.

\paragraph{General Weights to Uniform Weights}
We first describe a mapping from weighted mixtures to non-weighted (i.e., uniform) mixtures.
Given a weighted mixture with weights $\lambda_1 \leq \lambda_2 \leq \ldots \leq \lambda_k$, we instead consider the un-weighted mixture with $k' \geq k$ components created as follows: for each $i \in [k]$, the unweighted mixture contains $\lambda_i / \lambda_1$ copies of $G_i$, each with weight $\lambda_1 = 1/k'$.
It is not hard to check that the $k'$-mixture described above produces is statistically indistinguishable from the mixture with weights $\lambda_i$.

\paragraph{Rough Clustering Algorithm: Modified Rounding}
Since the final ``re-clustering'' step of our parameter-learning algorithm is straightforward to adapt to general weighted mixtures, we focus on how to adapt the constant-accuracy clustering algorithm of \cref{sec:rough-clustering}.

The algorithm itself is run as written in \cref{sec:rough-clustering}, using parameter $k' = 1/\lambda_1$ as the number of clusters to recover, with just one significant modification.
The algorithm can no longer rely only on size of a set of samples (as in line (2) of \textsc{cluster}) to determine whether a single cluster (rather than a nontrivial submixture) has been found.
Instead, in addition to recursively calling \textsc{split}, we can also treat the entire set of samples $\{X_1,\ldots,X_n\}$ as a candidate cluster, growing the size of our list of candidate clusterings only by some $F(w_{\min})$.

\section{Moment Upper Bounds Alone are Unlikely to Suffice} \label{sec:mom-upperbounds-not-enough}

In our SDP we insist on both upper as well as lower bounds on our moments being satisfied. In this section, we show that it is unlikely to be the case that moment upper bounds alone are sufficient to identify Gaussian sub-clusters. We do this by showing that any single moment constraint is not enough to ensure that a one-dimensional Gaussian mixture is trivial. 

Consider the mixture $M := \alpha \cN((1-\alpha)\Delta, 1) + (1-\alpha)\cN(-\alpha \Delta, 1)$. 
Note that the mean of the mixture is $0$ and the distance between the means is $\Delta$.
Suppose now that we insist that, for some $t$, \[\E_{x \sim M}[x^{2t}] = (2t-1)!! \E_{x \sim M}[x^2]^t.\]
Observe now that as $\Delta \rightarrow \infty$ we see that 
$\E_{x \sim M}[x^{2t}] \approx \Paren{ \alpha (1-\alpha)^{2t} + (1-\alpha) \alpha^{2t} } \Delta^{2t}.$ Hence, for small $\alpha$, we have that 
$$\E_{x \sim M}[x^{2t}] \approx \Paren{ \alpha (1-\alpha)^{2t} + (1-\alpha) \alpha^{2t} } \Delta^{2t} \gg (2t-1)!! \cdot \alpha^t(1-\alpha)^t \Delta^{2t} \approx (2t-1)!! \E_{x \sim M}[x^2]^t \;.$$ 
If $\alpha = 1/2$, then 
$$\E_{x \sim M}[x^{2t}] \approx \frac{1}{2^{2t+1}} \Delta^{2t} \ll (2t-1)!! \cdot \frac{1}{2^{2t}} \Delta^{2t} \approx (2t-1)!! \E_{x \sim M}[x^2]^t \;.$$ 
Hence, there is some $\alpha$ between $0$ and $1/2$ that satisfies the equation. 

When $t = 2$, it can be seen that $\alpha \approx 0.211$ is a solution as $\Delta \rightarrow \infty$. 

This implies that \emph{any single moment} matching that of a gaussian does not suffice, however an application of Holder then implies almost tight moments bounds for all moments smaller than $t$, i.e., $1, \dots, t-1$. This strongly suggests that it is not sufficient to have just moment upper bounds to uniquely identify a Gaussian. 

When $t=2$ and under the additional condition that $\E_{x \sim M}[x^3] = 0$, $\alpha$ is indeed forced to be either $1$ or $0$. Hence insisting on even one more moment has appropriate upper \emph{as well as lower} bounds seems to ensure that the mixture is trivial.

\end{document}